%% file: main.tex
\documentclass[11pt]{article}
\input{math_commands.tex}
\usepackage[top=1in,bottom=1in,left=1.25in,right=1.25in]{geometry}

\usepackage{palatino}
\usepackage[T1]{fontenc}
\usepackage[parfill]{parskip}
\usepackage{graphicx} 
\usepackage{float}
\usepackage{booktabs}
\usepackage{threeparttable}
\usepackage{adjustbox}
\usepackage{algorithm}
\usepackage{algorithmicx}
\usepackage{algpseudocode}
\usepackage{amsmath,amssymb}
\usepackage{amsfonts}
\usepackage{amsthm}
\usepackage{bbm}
\usepackage[dvipsnames]{xcolor}
\usepackage[table]{xcolor}
\usepackage{xr}
\usepackage{chngcntr}
\usepackage{multirow}
\usepackage{accents}
\usepackage{tcolorbox}
\usepackage{hyperref}
\hypersetup{colorlinks=true, urlcolor=cyan}
\usepackage[font=footnotesize]{caption}
\usepackage[sort&compress]{natbib}
\setcitestyle{authoryear}

\usepackage[resetlabels]{multibib}
\newcites{sec}{References}
\nocitesec{*}
\bibliographystylesec{abbrvnat}

\theoremstyle{definition}
\newtheorem{definition}{Definition}[section]
\newtheorem{assumption}{Assumption}[section]
\newtheorem{example}{Example}[section]
\newtheorem{proposition}{Proposition}[section]

\theoremstyle{plain}
\newtheorem{lemma}{Lemma}[section]
\newtheorem{theorem}{Theorem}[section]

\definecolor{lightgreen}{RGB} {126, 153, 103} 
\definecolor{lightred}{RGB} {173, 83, 73} 

\title{\textsc{Evaluating Algorithm-Assisted Human Decision-Making Over Repeated Algorithm Exposure: Recommendations for Effect Estimands and Experimental Design}}
\author{
    Maggie Wang \thanks{Stanford University, Department of Biomedical Data Science, Email: \href{mailto:mwang102@stanford.edu}{mwang102@stanford.edu}} \and Mike Baiocchi \thanks{Stanford University, Department of Epidemiology \& Population Health} 
    }
\date{}

\begin{document}

\maketitle

\begin{abstract}
    In algorithm-assisted decision-making in high-stakes settings like healthcare, an algorithmic decision support tool provides a recommendation, but the human ultimately makes the decision. Determining whether algorithm assistance actually improves the quality of human decision-making prior to deployment is critical, and randomized experiments are one way to collect robust evidence. Historically, however, experimental designs and analyses ignore how decision-making behavior adapts with repeated algorithm exposure. In this work, we define a set of effect estimands that account for and characterize human behavior adaptation under repeated exposure and justify why these estimands are useful to target for developing a better understanding of the impact of algorithm assistance. We propose a minimax stepped double wedge design that facilitates estimating these target estimands. Finally, we compare our proposed design to two common alternative designs identified through a review of historical randomized trials of algorithm assistance. We show that these designs are less amenable to estimating the target estimands and produce biased estimates under three different forms of behavioral adaptation inspired by dynamics observed in the real world -- automation bias, alert fatigue, and calibrated reliance.
\end{abstract}

\clearpage

\section{Introduction}

\subsection{Motivation}
Algorithms, especially those powered by modern machine learning or large language models, have the potential to help humans make decisions and complete tasks with greater efficiency, accuracy, and quality. In many high stakes settings, algorithms are designed and used as augmentative tools, rather than replacements, for human action and judgment; that is, while algorithms are allowed to provide advice or input on what decision the human should make, the human has the final say. One such high stakes setting, and the core motivation for this work, is healthcare. Algorithmic assistance tools have been developed for a variety of use cases in healthcare, including assisting radiologists and dermatologists in identifying pathologies from imaging \citep{Yu2024Heterogeneity, Groh2024Skin}, improving early identification of low ejection fraction from electrocardiograms \citep{Yao2021ECG}, recognizing neonatal seizures from electroencephalography data \citep{Pavel2020Neonatal}, assisting clinicians with planning and conducting stroke care for patients \citep{Zhang2026Golden}, and prescreening patients for eligibility for clinical trial enrollment \citep{Unlu2025Prescreening, Parikh2026Prescreening}. Recent years have also seen interest in the use of general-purpose large language models for assisting clinicians with performing differential diagnoses \citep{Qazi2026LMIC, Goh2025GPT}.

When algorithms are first developed, they are generally evaluated \textit{in silico} on benchmark datasets, according to metrics like accuracy, sensitivity, and specificity. These evaluations reflect the algorithm's own standalone performance, but do not reveal what would happen if the algorithm were to be given to a human as an assistive tool. One of the ways we can collect robust evidence on the effect of algorithm assistance on human decision-making is by running randomized trials in which the provision of algorithm assistance is randomized.

Algorithm assistance as an intervention is different from more canonical interventions, e.g. a novel drug, and these distinct properties should influence how we approach designing randomized trials for algorithm assistance interventions. We highlight two properties in particular. First, the recipient of algorithm assistance is often not the same unit on which the outcome is measured. For instance, for an algorithmic decision support tool for physicians that displays an alert when their patient needs additional medical attention, the \textit{physician} is the one who sees the alert, but what we want to know is whether showing such an alert actually improves \textit{patient} outcomes. Second, depending on the trial design and setting, users may interact with the algorithm multiple times over the course of the trial. However, the user's trust in and familiarity with the algorithm the very first time they interact with it could be quite different from the $t^{th}$ time they interact with it. Similarly, how they make decisions and complete tasks without algorithm assistance could change after they have had some prior exposure to the algorithm. 

Notably, the second property -- the dynamics of human behavior -- is often ignored in trials of algorithm assistance (see our literature review in Appendix \ref{app:sec:elicit}, where only 51 of 274 (18.6\%) included trials dealt with it), leading to potential under- or over-estimation of the intended target effect estimand. The first property -- the discrepancy between unit of intervention and unit of outcome measurement -- is commonly left implicit, which obscures the fact that the impact of algorithm assistance on the outcome(s) of interest is mediated by the behavior of the human being assisted. Hence, sometimes a null finding does not reflect that the algorithm assistance is  unhelpful in general, but rather that the algorithm assistance tool is delivering its assistance in an unhelpful way. We focus on settings where the human makes a discrete decision or completes a discrete task, where the algorithm's output is a singular recommended decision or action, and where the algorithm stays fixed (i.e. is not updated or retrained) over the course of the trial. While many algorithms (e.g. generative and large language models) fall outside this scope, the same concepts about behavioral mediation and behavior dynamics apply.

\subsection{Contributions}

We provide a unified framework for attending to the behavioral mediation and behavior dynamics properties when designing and analyzing randomized trials for algorithm assistance. We specifically make the following contributions:
\begin{enumerate}
    \item \textbf{We define a set of target effect estimands (Section \ref{sec:estimands}) } -- global, immediate, habituation, and skilling effects -- that collectively characterize the impact of algorithm assistance when the effect may change with repeated algorithm exposure (or, equivalently, when there is between-decision, within decision-maker interference). Additionally, in Appendix \ref{app:sec:conditional}, we propose conditional effect estimands, localizing on decisions where decision-makers comply with algorithm assistance.
    \item \textbf{We propose a minimax stepped ``double wedge" design (Section \ref{sec:minimax-design})}, where human decision-makers are randomized to always-treated, always-control, and staggered onset and offset times. This design facilitates unbiased estimation of each of the target estimands and is conceptually simple. We additionally derive a minimax design that specifies the number of decision-makers to allocate to each arm of the trial, given a fixed total budget of decision-makers $N$ and trial length $T$.
    \item \textbf{We compare the stepped double wedge design to two common alternatives (Section \ref{sec:alternatives})}, a parallel decision-maker randomized design and a decision-randomized design, as identified in our literature review (Appendix \ref{app:sec:elicit}). The decision-maker randomized design can be more sample efficient for estimating global effects, but is not amenable to estimating immediate, habituation, and skilling effects. We additionally demonstrate that the decision-randomized design is biased across three stylized data-generating processes that reflect possible behavioral dynamics that can be encountered in the real-world: increasing automation bias, increasing alert fatigue, and gradual calibration. We then use simulations to demonstrate the tradeoffs between eliminating bias and increasing variance with the stepped double wedge design.
\end{enumerate}

In our review of the literature (Appendix \ref{app:sec:elicit}), a majority of randomized trials of algorithm assistance have been in the medical domain (149 of 274 trials). While the framework we propose is general and can be applied to the evaluation of algorithm assistance in domains beyond medicine, we emphasize that our contributions may be especially relevant for improving how randomized trials are being designed for evaluating algorithm-assisted decision-making in medicine.
 
\section{Related Work}
Our work is inspired by recently developed methods for randomized experimental evaluation of algorithms assisting humans. \cite{Imai2023Pretrial} used a principal strata approach to analyze a field experiment of a judge making pretrial release decisions for defendants with and without randomized algorithm assistance. \cite{BenMichael2025AIExperiment} developed a framework for identifying and comparing the performance of human-with-AI, human-alone, and AI-alone decision-making, applied to the same trial as in \cite{Imai2023Pretrial}. To classify the human and AI's decision-making ability, they rely on classification metrics (true / false positives, true / false negatives) defined by comparing the potential outcome under a baseline decision to a hypothetical decision that the human / AI makes. Neither \cite{Imai2023Pretrial} nor \cite{BenMichael2025AIExperiment} address repeated algorithm exposure in their methods (while \cite{Imai2023Pretrial} conduct a conditional randomization test for spillover effects across repeated decisions and find no evidence of it, they acknowledge that the test may be underpowered). \cite{Raji2025Prediction} highlight how experimental design choices, cognitive biases of decision-makers, and spillover effects under repeated algorithm exposure combine to result in biased estimation of average treatment effects. However, they do not propose alternative estimands and alternative experimental designs that can remove this bias. In addition, similar to our work, they derive bias expressions for illustrative data-generating processes, but their bias expressions are for effect estimands defined on the potential \textit{decisions}, whereas ours are for effect estimands defined on the potential \textit{outcomes}, mediated by the potential decisions. Our proposed local effect estimands are defined with respect to concordance classes that resemble the compliance classes developed by \cite{McLaughlin2024Complementarity}. However, while \cite{McLaughlin2024Complementarity} focus on utilizing these compliance classes to decompose the objective function for the algorithm and design an algorithm that produces better recommendations for the human, we focus on estimating local effect estimands. Additionally, our concordance classes allow for spillover effects under repeated algorithm exposure, which are not covered by \cite{McLaughlin2024Complementarity}. 

We also draw upon prior work in online experimentation in industry. \cite{Hohnhold2015Longterm} devised a ``post-period" experimental design and a ``cookie-cookie-day" experimental design to measure user learning effects, with an application to how users change their interactions with ads over repeated exposure. Their ``post-period" design with lagged starts is conceptually equivalent to the staggered onsets portion of our proposed stepped double wedge design, and the ``cookie-cookie-day" experiment is a modified version of staggered onsets, where individuals have staggered introductions to treatment but do not stay on treatment for more than one time period. \cite{Bojinov2019TimeSeries} also propose estimands and develop estimation and inference methodology for time-varying treatment effects, but focus on $p$-lag estimands, which contrast the potential outcomes at time $t$ under equivalent treatment assignment vectors apart from the assignment at the $(t-p)^{th}$ point in time. Our minimax stepped double wedge design is most similar to the minimax wedge design proposed by \cite{Basse2023Minimax}, but accommodates the estimation of two additional estimands (the global and skilling effects) by including a staggered offsets portion in addition to the staggered onsets portion and by including always-treated and always-control decision-makers.

Finally, our work is related to the literature on how humans may interact with algorithms in suboptimal or non-synergistic ways \citep{Vaccaro2024Combinations, Kostick2022Imperfect}, including algorithm overreliance \citep{Logg2019Appreciation, Suresh2020Trust, Vicente2023Biases, Bucinca2021CognitiveForcing} (sometimes also called automation bias \citep{Khera2023AutomationBias} and anchoring bias \citep{Rastogi2022CognitiveBias}) and underreliance (often called algorithmic aversion \citep{Dietvorst2015Aversion, Gaube2021Susceptibility} or, in medical settings, alert fatigue \citep{Weingart2003Override, Ancker2017AlertFatigue}). In addition, several works have shown that these behaviors are not static and can be influenced by prior exposure to and experience with the algorithm. Early mistakes made by the algorithm, for instance, can lead to skepticism and underreliance in later decisions \citep{Nourani2020FirstImpressions, Chacon2022Longitudinal}. \cite{Yu2017TrustDynamics, Li2023MarkovianTrust} propose and fit dynamic trust models to algorithm-assisted decision-making data, showing that a human's trust in algorithm output can evolve with algorithm performance and other contextual factors; \cite{Kahr2024TrustDevelopment} assess the correlation between cumulative trust experience in earlier interactions with the algorithm with subsequent trust experiences; and \cite{Yang2021TrustDynamics} evaluate how humans adjust their trust across each interaction with an algorithm, depending on whether the algorithm's recommendation results in a desirable or undesirable outcome. Though there is much empirical work providing evidence for dynamic decision-maker behavior under repeated algorithm exposure, to the best of our knowledge, there is no unifying causal inference framework for defining and estimating causal effects that capture these dynamics.

\section{Notation, Setup, and Assumptions}
We consider a setting where one set of people, the decision-makers, are tasked with making binary decisions regarding another set of people, the decision-recipients, e.g. physicians making diagnoses for patients. An algorithmic assistance tool has been developed to assist in making this decision, and we are interested in evaluating whether assistance leads to better outcomes for the recipients. The overall setup and key notation is summarized in Figure \ref{fig:setup}. Our notation and assumptions are adapted from \cite{Imai2023Pretrial}, with modifications made to accommodate interference between decisions, within decision-makers.

To perform this evaluation, we conduct a randomized experiment in which assistance is randomized. Suppose there are $N$ total decision-makers, indexed $i=1, 2, \dots, N$. We define timepoints $t = 1, 2, \dots, T$, where time 1 is the start of the randomized experiment and time $T$ marks the end of the experiment. We denote $R_{it}$ as the recommendation for the decision-recipient seen at time $t$ by decision-maker $i$. We assume that the algorithm is fixed, e.g. not re-trained or updated, between $t=1$ and $t=T$. Note that $R_{it}$ is well-defined and observable by the study analyst even when the tool recommendation is not shown to the decision-maker, as long as the tool is ``silently deployed" on such cases and the recommendation it would have shown is recorded.

Let $Z_{it}$ denote the treatment assignment variable for the decision-recipient seen by decision-maker $i$ at time $t$, where $Z_{it} = 1$ indicates that the algorithm output is shown to the decision-maker and $Z_{it} = 0$ indicates the algorithm output is hidden. Denote the random assignment vector for decision-maker $i$'s decision-recipients as $\rmZ_{i}$. We assume partial interference, so that decision-maker $i$'s decision depends on the treatment assignments of their own decision-recipients but not on the treatment assignments for other decision-makers' recipients.
\begin{assumption}[Partial interference]\label{assump-partial-interf}
    $D_{it}(\rmZ) = D_{it}(\rmZ')$ for all $\rmZ, \ \rmZ'$ such that $\rmZ_i = \rmZ'_i$.
\end{assumption}

Under Assumption \ref{assump-partial-interf}, we can write the potential decision made by decision-maker $i$ at time $t$ as $D_{it}(\rmZ_{i}) \in \{0, 1\}$. We are also interested in understanding the outcomes for the decision-recipients, not just the decisions made by the decision-makers. We denote the potential outcome for the decision-recipient seen by decision-maker $i$ at time $t$ as $Y_{it}(\mathbf{Z}_i, D_{it}(\mathbf{Z}_i))$. Using the setup developed by \cite{Imai2023Pretrial}, we assume that the effect of the treatment on a recipient outcome is fully mediated by the decision-maker's decision for that recipient. We additionally assume non-anticipatory effects, meaning a potential decision can only be affected by past and current treatment assignment, not future treatment assignments. Since treatment assignments are random (i.e. not dependent on potential decisions or outcomes), we assume ignorability holds.

\begin{assumption}[Full mediation]\label{assump-mediation}
\begin{align*}
Y_{it}(\mathbf{Z}_{i}, D_{it}) = Y_{it}(\mathbf{Z'}, D_{it}) \text{ for all } \mathbf{Z}_i^{\leq t}, \ {\mathbf{Z}'}_i^{\leq t}, D_{it} \in \{0, 1\}
\end{align*}
\end{assumption}

\begin{assumption}[Non-anticipation]\label{assump-nonanticip}
\begin{align*}
    D_{it}(\rmZ) = D_{it}(\rmZ') \text{ for all } \rmZ, \ \rmZ' \text{ such that } \rmZ^{\leq t} = \rmZ'^{\leq t}
\end{align*}
where the superscript $^{\leq t}$ corresponds to recipients seen up to time $t$. 
\end{assumption}

\begin{assumption}[Ignorability]\label{assump:ignorability}
    \begin{align*}
        (D_{it}(\mathbf{z}^{\leq t}), Y(D_{it}(\mathbf{z}^{\leq t}))) \indep \mathbf{Z}_i^{\leq t} \text{ for all } \mathbf{z}^{\leq t}
    \end{align*}
\end{assumption}

\begin{figure}[ht!]
    \centering
    \includegraphics[trim={4cm 0 2cm 0}, width=\linewidth]{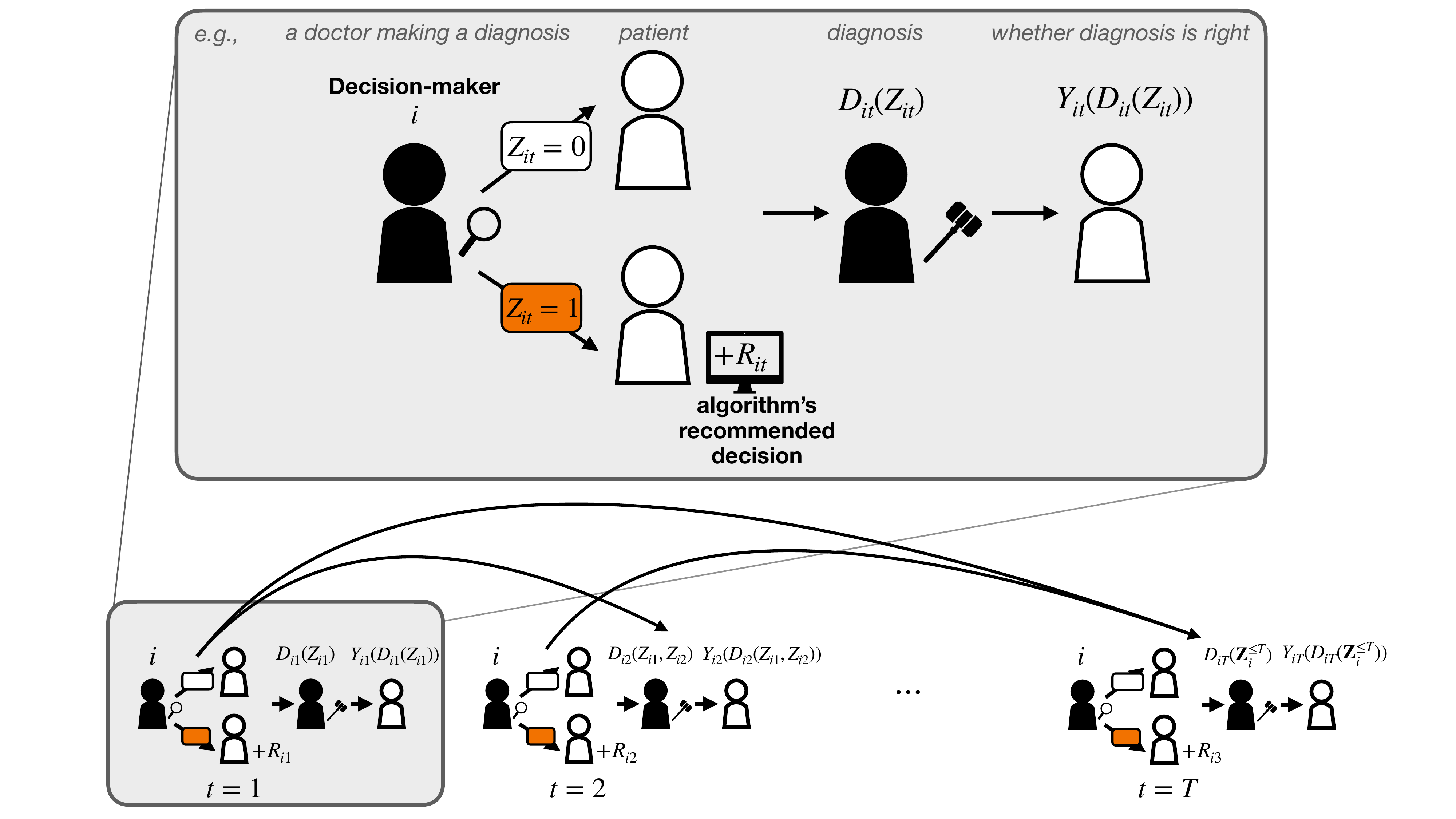}
    \caption{\textbf{Overview of algorithm-assisted decision-making setup}. (\textit{top, single decision with no prior decisions)} A human decision-maker makes a decision for a decision-recipient either with ($Z_{it}=1$) or without ($Z_{it}=0$) algorithm assistance, resulting in a potential decision ($D_{it}(Z_{it})$) and a potential outcome mediated by that potential decision ($Y_{it}(D_{it}(Z_{it}))$). \textit{(bottom)} The decision-maker makes repeated decisions on new decision-recipients, and prior exposure to algorithm assistance ($\mathbf{Z}_i^{\leq t})$ can affect the current potential decision ($D_{it}(\mathbf{Z}_{i}^{\leq t})$) and potential outcome ($Y_{it}(D_{it}(\mathbf{Z}_{i}^{\leq t}))$).}
    \label{fig:setup}
\end{figure}

Under Assumption \ref{assump-mediation}, we rewrite the potential outcome for the decision-recipient seen by decision-maker $i$ at time $t$ as $Y_{it}(D_{it}(\rmZ_i)) \in \mathbb{R}$. Let $\rmZ_i^{\leq t}$ denote the vector of treatment assignments for all recipients seen by decision-maker $i$ up to time $t$. Under Assumption \ref{assump-nonanticip}, we can then express the potential decision at time $t$ as $D_{it}(\rmZ_i^{\leq t})$ and the potential outcome as $Y_{it}(D_{it}(\rmZ_i^{\leq t}))$. 

Lastly, let $\mathbf{w}_{t, t'}$ denote the length $T$ vector corresponding to treatment onset at time $t$ and treatment offset at time $t'$. That is, 
$$\mathbf{w}_{t, t'} = (0, 0, \dots, 0, \underbrace{1}_{t}, 1, \dots, 1, \underbrace{0}_{t'}, 0, \dots 0)$$ With a slight abuse of notation, let $\mathbf{w}_{t, \infty}$ denote the treatment assignment where all decisions up to and including time $t-1$ are control and the subsequent decisions through the remainder of the experiment are treated. Similarly, let $\mathbf{w}_{1, t}$ denote the treatment assignment where all decisions up to and including time $t-1$ are treated and the subsequent decisions are control. In this work, we will be most interested in the following possible values for $\mathbf{Z}_i^{\leq t}$: $\mathbf{0}^{\leq t}$ (the length $t$ all-zero's vector), $\mathbf{1}^{\leq t}$ (the length $t$ all-one's vector), $\mathbf{w}_{t, \infty}$, and $\mathbf{w}_{1, t}$.

\section{What effects to estimate? Global, immediate, habituation, and skilling effects}\label{sec:estimands}
We propose the following four complementary treatment effect estimands to evaluate algorithm-assisted decision-making. We focus in this section on defining sample average treatment effects, which assume that the only source of randomness comes from the treatment assignment vector, $\mathbf{Z}$, because they lend themselves to the minimax design proof in Section \ref{sec:minimax-design}. The definitions can be easily extended to population average treatment effects by taking an expectation over the sample of potential outcomes and potential decisions (see, for instance, Section \ref{sec:alternatives}, where we use a population-level definition for $\tau$, and Appendix \ref{app:sec:conditional}, where we shift to the population-level for conditional effect estimands).

\subsection{Estimand definitions}

\begin{definition}[Global Treatment Effect at $t$]
    \begin{align}\label{eq:tau-global}
    \tau_{t} := \frac{1}{N}
    \sum_{i=1}^N
        Y_{it}(D_{it}(\mathbf{1}^{\leq t}))
        - Y_{it}(D_{it}(\mathbf{0}^{\leq t})) 
\end{align}
\end{definition}
The global treatment effect is often an estimand of interest in settings with interference and answers the question of what the outcomes would be if we were to overhaul the status quo (no algorithm assistance at all) and deploy algorithm assistance ubiquitously across all decisions. $\tau_t $ is, in many cases, the implicit target estimand of randomized trials of algorithm assistance. The next three effect estimands, however, are less commonly targeted and estimated, whether implicitly or explicitly.

\begin{definition}[Immediate Treatment Effect at $t$] 
    \begin{align}\label{eq:delta-immediate}
    \delta_{t} := 
    \frac{1}{N}
    \sum_{i=1}^N
        Y_{it}(D_{it}(\mathbf{0}^{<t}, 1))
        - Y_{it}(D_{it}(\mathbf{0}^{\leq t}))
\end{align}
\end{definition}
The immediate treatment effect answers the question of what happens when a decision-maker receives algorithm assistance for the very first time at $t$. 

\begin{definition}[Habituation Effect at $t$]
    \begin{align}\label{eq:Delta-habituated}
    \Delta_{t} := 
    \frac{1}{N}
    \sum_{i=1}^N
        Y_{it}(D_{it}(\mathbf{1}^{\leq t}))
        - Y_{it}(D_{it}(\mathbf{0}^{< t}, 1))
\end{align}
\end{definition}
The habituation effect answers the question of what happens when a decision-maker receives algorithm assistance on every decision up to $t$ versus when they receive assistance for the very first time at $t$. In other words, the effect assesses the impact of the decision-maker ``habituating" to algorithm assistance under repeated exposure.

\begin{definition}[Skilling Effect at $t$]
    \begin{align}\label{eq:lambda-skilling}
        \lambda_{t} :=
        \frac{1}{N}
        \sum_{i=1}^N
        Y_{it}(D_{it}(\mathbf{1}^{< t}, 0))
        - Y_{it}(D_{it}(\mathbf{0}^{\leq t}))
    \end{align}
\end{definition}
The skilling effect answers the question of what happens when a decision-maker receives algorithm assistance until $t$ but has it taken away at $t$ versus when they never receive algorithm assistance at all. In other words, the effect assesses the impact of the decision-maker losing access to the algorithm after repeated exposure.

\subsection{Why estimate these effects?}
We provide three arguments for why it is useful to estimate $\tau_t, \delta_t, \Delta_t$, and $\lambda_t$: (1) there is information in the joint signs of $\delta_t$, $\Delta_t$, and $\tau_t$ that would otherwise be obscured if we only estimated $\tau_t$, (2) $\lambda_t$ can be used to address concerns about the use of artificial intelligence tools leading to the deterioration or loss of human skills, and (3) examining both $\delta_t$ and $\Delta_t$ helps to disentangle the factors driving temporal trends in treatment effects.

\paragraph{Reason 1: Implications of the joint signs of $\delta_t$, $\Delta_t$, and $\tau_t$}
We argue that the joint signs of $\delta_t$, $\Delta_t$, and $\tau_t$ have different implications on how we ought to interpret the effect of the algorithm assistance intervention. In Table \ref{tab:joint-signs}, we show six possible joint sign combinations for $\delta_t$, $\Delta_t$, and $\tau_t$ (assuming none are exactly zero) and interpret each combination.

\begin{table}[ht!]
    \centering
    \begin{adjustbox}{max width=\textwidth}
    \begin{tabular}{cccl}
    \toprule
    $\delta_t$ & $\Delta_t$ & $\tau_t$ & Interpretation \\
    \midrule
     + & + &  + & Immediately beneficial, amplified by repeated exposure \\
     + &  $-$ & + & Immediately beneficial, diminished by repeated exposure, but still beneficial at $t$ \\
     + & $-$ &  $-$ & Immediately beneficial, diminished by repeated exposure, no longer beneficial at $t$ \\
    $-$ & $-$ & $-$ & Immediately harmful, worsened by repeated exposure \\
    $-$ & + & + & Immediately harmful, ameliorated by repeated exposure, becomes beneficial by $t$\\
    $-$ & + & $-$ & Immediately harmful, ameliorated by repeated exposure, but still harmful at $t$ \\
    \bottomrule
    \end{tabular}
    \end{adjustbox}
    \caption{\textbf{Joint sign combinations for $\tau_t$, $\delta_t$, and $\Delta_t$ that have different interpretations.} Note that, since $\delta_t + \Delta_t = \tau_t$, it is not possible for both $\delta_t$ and $\Delta_t$ to be negative while $\tau_t$ is positive. Likewise, it is impossible for $\delta_t$ and $\Delta_t$ to both be positive while $\tau_t$ is negative.}
    \label{tab:joint-signs}
\end{table}

Not only do $\delta_t$, $\Delta_t$, and $\tau_t$ tell us about the effect of algorithm assistance, they also can provide valuable insights into what steps might be taken to improve the intervention. These changes could be applied to the intervention itself, e.g. refining the user interface, or to the surrounding infrastructure and workflows, e.g. administering training to the decision-makers so that they are more skillful in when and how they utilize algorithm assistance. Of particular interest is when $\delta_t$ and $\Delta_t$ have opposing signs:

\begin{enumerate}
    \item {$\mathbf{\boldsymbol{\delta}_t < 0, \boldsymbol{\Delta}_t > 0.}$} This sign combination suggests that algorithm assistance is immediately harmful, i.e. it causes human decision-makers to make worse decisions than they would have had algorithm assistance not been available, but that repeated exposure helps to overcome this harm. In other words, there may be a ``burn-in" period during which the human decision-makers learn how to use the algorithm assistance to their advantage. Two natural responses are (1) making sure that the decision-makers progress through their burn-in period via realistic training sessions before they are permitted to use algorithm assistance on actual decisions, and (2) redesigning the algorithm assistance interface so that it is easier to learn and more intuitive to use, which could reduce the length of the burn-in.
    \item {$\mathbf{\boldsymbol{\delta}_t > 0, \boldsymbol{\Delta}_t < 0.}$} This sign combination suggests that algorithm assistance is immediately beneficial, i.e. the first time human decision-makers use algorithm assistance at time $t$, it leads to better decisions; however, over repeated exposure, this benefit diminishes. There are several reasons why benefit could diminish through repeated exposure, including a gradual over-reliance on the algorithm (a tendency to agree with the algorithm's recommendation no matter what, instead of exercising critical judgment about whether the recommendation is correct) or a gradual ignoring of the algorithm (a tendency to grow annoyed with or skeptical of the algorithm's recommendations and to ignore what it recommends, even when it is giving correct recommendations). Appropriate remedial action could include administering regular supplemental training sessions to target over- or under-reliant behavior and considering offering \textit{selective} algorithm assistance on a subset of decisions where human decision-makers are most likely to find the recommendation actually helpful \citep{McLaughlin2024Complementarity}.
\end{enumerate}
The joint signs of $\delta_t$ and $\Delta_t$ are a summary metric for complex human behavioral processes, but do not provide the full picture into \textit{why} human decision-makers behave in certain ways. To better understand these behaviors -- e.g., determining whether the $\mathbf{\boldsymbol{\delta}}_t > 0, \boldsymbol{\Delta}_t < 0$ case is a result of gradual over-reliance or gradual ignoring -- and to design solutions -- e.g., preventing over-reliance -- conducting follow-up interviews would be helpful. We encourage consideration of mixed methods approaches to evaluating algorithm assistance interventions, with estimation of $\delta_t$ and $\Delta_t$ serving as the quantitative component.

\paragraph{Reason 2: Addressing concerns about ``deskilling" through $\lambda_t$}
Alongside the surge in development of artificial intelligence tools, especially ones powered by large language models, there has been mounting concern that the use of artificial intelligence to augment tasks traditionally done by humans alone may cause humans to lose fundamental problem-solving and critical-thinking skills \citep{Liu2026Deskilling, Budzyn2025EndoscopistDeskilling}. Evidence for whether such deskilling actually occurs, however, is nascent; indeed, how to even measure deskilling rigorously remains an open question. We present the skilling effect estimand, $\lambda_t$, as a way to shed light on whether deskilling exists, specifically as it pertains to decision-making tasks. When $\lambda_t < 0$, this suggests that humans become worse decision-makers after repeated algorithm assistance -- when algorithm assistance is removed, they make poorer decisions than they would have had they never received algorithm assistance at all. A negative $\lambda_t$ estimate therefore suggests ``deskilling." In contrast, when $\lambda_t > 0$, this suggests that algorithm assistance leads to humans becoming better decision-makers even in the absence of assistance, suggesting ``upskilling". By clearly defining skilling in formal causal inference terms, we provide a concrete concept to ground conversations and research around skilling effects. 

\paragraph{Reason 3: Disentangling habituation from changes in the immediate effect}
Suppose $\tau_t \neq \tau_1$, meaning the global treatment effect has changed with time. Consider the following decomposition:
\begin{align*}
    \tau_t - \tau_1 
        &= (\delta_t + \Delta_t) - (\delta_1 + \underbrace{\Delta_1}_{=0})
        = (\delta_t - \delta_1) + \Delta_t
\end{align*}
The above makes it clear that the difference between $\tau_1$ and $\tau_t$ is driven by two factors: (1) a change in the immediate effect from $\delta_1$ to $\delta_t$, and (2) the habituation effect. These two factors represent fundamentally different processes. Under our assumptions, the change in (1) arises when there are sources of temporal trends apart from the intervention, including decision-makers becoming more experienced as they make more decisions; drifts in the characteristics of the decision-recipient population that affect the quality of decisions and/or the decision-recipient outcomes; and changes in the environment, such as reduced resource availability forcing decision-makers to make certain decisions over others in order to conserve resources. In contrast, the change driven by (2) comes specifically from repeated algorithm exposure.

\subsection{Extensions to conditional effect estimands} In Appendix \ref{app:sec:conditional}, we present \textit{conditional} effect estimands under repeated exposure, inspired by the literature on instrumental variables and local average treatment effects. In particular, these conditional effect estimands dig into cases where the decision-maker would have disagreed with the algorithm's recommendation if they were making the decision alone (without assistance) but would comply with the algorithm when they are provided algorithm assistance. Additionally, in Appendix \ref{app:sec:profiling}, we propose an approach for covariate ``profiling" of complying decision-makers, which surfaces characteristics of decision-makers that may be responsible for driving compliance behavior.

\section{How to estimate these effects? Making use of a stepped ``double" wedge trial design}\label{sec:minimax-design}

Drawing upon \cite{Basse2023Minimax}, we derive a minimax stepped ``double wedge" trial design, where the difference between our design and that of \cite{Basse2023Minimax} is the accommodation of two additional effect estimands -- $\tau_t$ and $\lambda_t$. We specifically consider the class of designs where the treatment assignment vector $\mathbf{Z}_i$ for decision-maker $i$ must take on one of the following values: $\mathbf{1}$ (always-treated), $\mathbf{0}$ (always-control), $\{\mathbf{w}_{t, \infty}\}_{t=2}^{T}$ (treatment onset at $t=2, \dots, T$), and $\{\mathbf{w}_{1, t}\}_{t=2}^{T}$ (treatment offset at $t=2, \dots, T$). A standard stepped wedge trial does not include staggered treatment \textit{offsets}, only treatment onsets. The ``double wedge" term refers to the fact that our design contains both. Let $\pi(\mathbf{Z})$ denote the design, and let $\mathbb{H}$ denote the class of all possible designs where $\mathbf{Z}$ takes on the aforementioned possible values. Our goal, then, is to determine the minimax design, $\pi^{opt}$, within the class $\mathbb{H}$.

\begin{figure}[ht!]
    \centering
    \includegraphics[trim={4cm 0 4cm 0}, width=\linewidth]{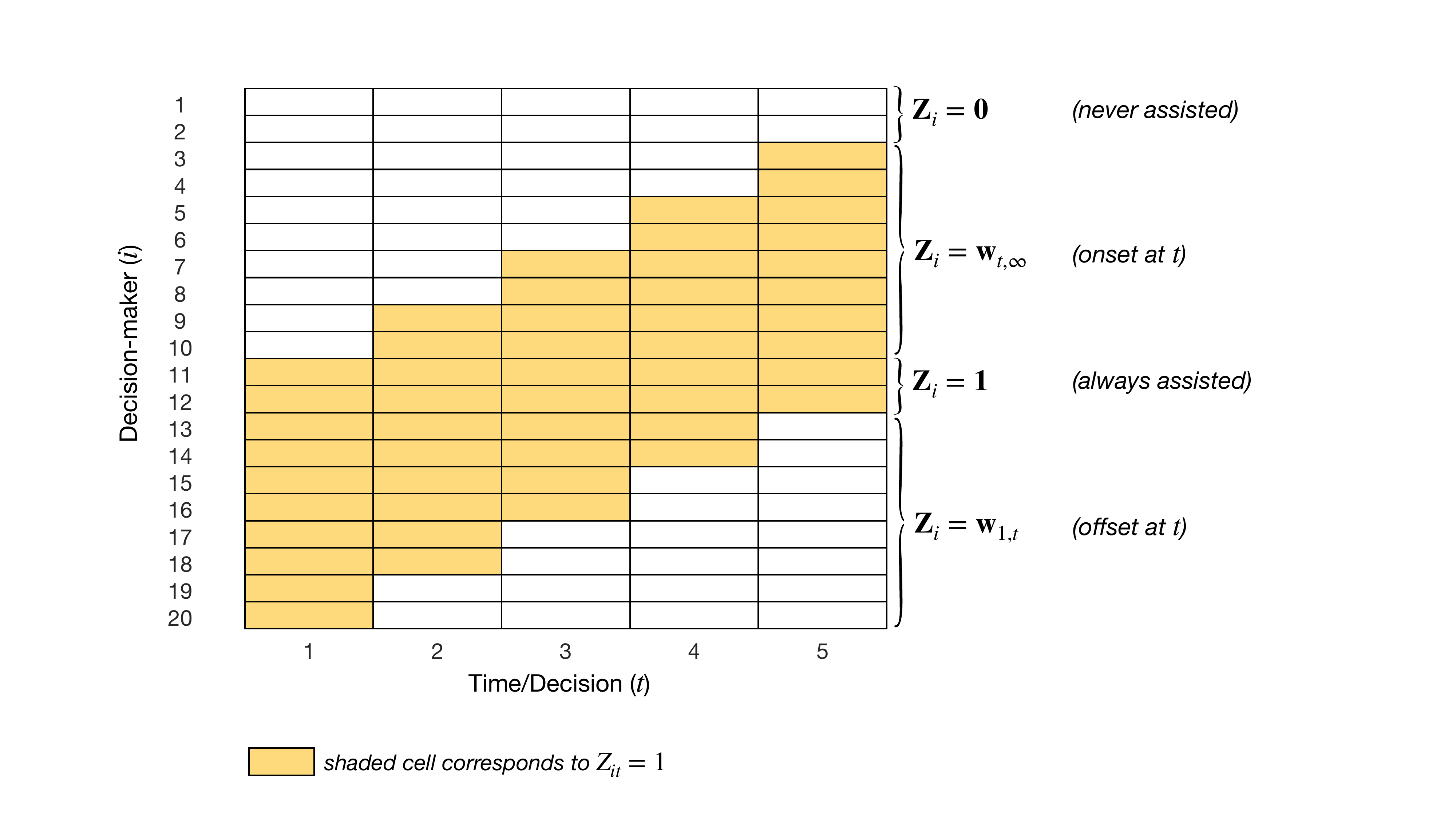}
    \caption{\textbf{Stepped double wedge trial design for a length $T=5$ experiment with 20 decision-makers.} Each row corresponds to a decision-maker and each column to a point in time. (We assume every decision-maker makes a single decision at each point in time). Shaded cells correspond to decisions where the decision-maker is assigned to receive algorithm assistance. A subset of decision-makers are assigned to always-assisted ($\mathbf{Z}_{i} = \mathbf{1}$), a subset to never-assisted ($\mathbf{Z}_i = \mathbf{0}$), a subset to assistance onset at $t=2, \dots, T$ ($\mathbf{Z}_i = \mathbf{w}_{t, \infty}$), and a subset to assistance offset at $t=2, \dots, T$ ($\mathbf{Z}_i = \mathbf{w}_{1, t}$). Our minimax theorem (Theorem \ref{thm:minimax}) provides guidance on how to select the sizes of each of these subsets.}
    \label{fig:mod-sw}
\end{figure}

For simplicity, we assume that each decision-maker makes a single decision for a single decision-recipient at each point in time, for a total of $T$ decisions per decision-maker and $NT$ decisions overall. We largely adopt the same notation and terminology from \cite{Basse2023Minimax}, restated here for clarity and completeness. Let $\mathbf{Y}(D(\mathbf{z})) \in \mathbb{R}^{N \times T}$ denote the potential outcomes matrix where every decision-maker is assigned $\mathbf{z}$. We denote the full schedule of potential outcomes as
    $$\underline{\mathbf{Y}} = [\mathbf{Y}(D(\mathbf{1})), \mathbf{Y}(D(\mathbf{0})), \mathbf{Y}(D(\mathbf{w}_{2, \infty})), \dots \mathbf{Y}(D(\mathbf{w}_{T, \infty}))), \mathbf{Y}(D(\mathbf{w}_{1, 2})), \dots, \mathbf{Y}(D(\mathbf{w}_{1, T}))],$$
and $\underline{\mathbb{Y}}$ denotes the support for the potential outcomes schedule $\underline{\mathbf{Y}}$.

Let $\mathcal{T}_t = \{i: \mathbf{Z}_i = \mathbf{1} \text{ or } \mathbf{Z}_i = \mathbf{w}_{1, t'} , t' > t\}$ denote the set of all decision-makers who make assisted decisions for every decision up until time $t$. This set combines the decision-makers that are assigned to the always-treated assignment vector and those that are assigned to treatment offset at some time $t'$ that is later than $t$. Similarly, let $\mathcal{C}_t = \{i: \mathbf{Z}_i = \mathbf{0} \text{ or } \mathbf{Z}_i = \mathbf{w}_{t', \infty}, t' > t\}$ denote the set of all decision-makers who make unassisted decisions for every decision up until time $t$, which combines the decision-makers that are assigned to the always-control assignment vector and those that are assigned to treatment onset at some time $t'$ that is later than $t$. Let $N_{1, t} := |\mathcal{T}_t| = N_1 + \sum_{t'=t+1}^{T} N_{w_{1, t'}}$ and $N_{0, t} := |\mathcal{C}_t| = N_0 + \sum_{t'=t+1}^{T} N_{w_{t', \infty}}$. We define the following estimators for $\tau_t$, $\Delta_t$, $\delta_t$, and $\lambda_t$, respectively.
\begin{align*}
    \hat{\tau}_t^{sdw} &= 
        \frac{1}{N_{1, t}}\sum_{i=1}^{N} Y_{it} \mathbbm{1}\{i\in \mathcal{T}_t\} -
        \frac{1}{N_{0, t}}\sum_{i=1}^{N} Y_{it} \mathbbm{1}\{i\in \mathcal{C}_t\} \\
    \hat{\Delta}_t^{sdw} &= 
        \frac{1}{N_{1, t}}\sum_{i=1}^{N} Y_{it} \mathbbm{1}\{i\in \mathcal{T}_t\}  -
        \frac{1}{N_{w_{t, \infty}}}\sum_{i=1}^{N} Y_{it} \mathbbm{1} \{\mathbf{Z}_i = \mathbf{w}_{t, \infty}\} \\
    \hat{\delta}_t^{sdw} &= 
        \frac{1}{N_{w_{t, \infty}}}\sum_{i=1}^{N} Y_{it} \mathbbm{1} \{\mathbf{Z}_i = \mathbf{w}_{t, \infty}\} -
        \frac{1}{N_{0, t}}\sum_{i=1}^{N} Y_{it} \mathbbm{1}\{i\in \mathcal{C}_t\}
        \\
    \hat{\lambda}_t^{sdw} &= 
        \frac{1}{N_{w_{1, t}}}\sum_{i=1}^{N} Y_{it} \mathbbm{1} \{\mathbf{Z}_i = \mathbf{w}_{1, t}\} -         \frac{1}{N_{0, t}}\sum_{i=1}^{N} Y_{it} \mathbbm{1}\{i\in \mathcal{C}_t\}.
\end{align*}
Under Assumptions \ref{assump-partial-interf} - \ref{assump:ignorability}, these estimators are unbiased for $\tau_t$, $\Delta_t$, $\delta_t$, and $\lambda_t$ and are all computable by plugging in observed data from the stepped double wedge trial.

\begin{definition}[Loss, adapted from \citep{Basse2023Minimax}]
    We define the \textit{loss} under a particular random assignment $\mathbf{Z}$ and potential outcomes schedule $\underline{\mathbf{Y}}$ as
    \begin{align*}
        L(\mathbf{Z}, \underline{\mathbf{Y}}) =
        \sum_{t=2}^{T} (\tau_t - \hat{\tau}_t^{sdw})^2 + (\delta_t - \hat{\delta}_t^{sdw})^2 + (\Delta_t - \hat{\Delta}_t^{sdw})^2 + (\lambda_t - \hat{\lambda}_t^{sdw})^2
    \end{align*}
\end{definition}

\begin{definition}[Risk, adapted from \citep{Basse2023Minimax}]
    We define the \textit{risk} of an assignment mechanism $\pi \in \mathbb{H}$ as 
    \begin{align*}
        r(\pi; \underline{\mathbb{Y}}) = \mathbb{E}_{\pi}[L(\mathbf{Z}, \underline{\mathbf{Y}})]
    \end{align*}
\end{definition}

Let $\mathcal{Z}$ denote the support for $\mathbf{Z}$. We define the \textit{minimax optimal design} $\pi^{\text{opt}}$ as the design the minimizes the maximum risk over the support of potential outcomes, $\mathbb{Y}$, over the class of designs $\mathbb{H}$.

\begin{definition}[Minimax objective, adapted from \citep{Basse2023Minimax}]
    \begin{align*}
        \pi^{\text{opt}} = \underset{\pi \in \mathbb{H}}{\arg \min} \max_{\underline{\mathbf{Y}} \in \mathbb{Y}} r(\pi; \underline{\mathbf{Y}})
    \end{align*}
\end{definition}

\begin{theorem}[Minimax stepped double wedge design]\label{thm:minimax}
    Under an assumption on the support of the potential outcomes (see Appendix \ref{app:proofs:minimax}), the approximate (integer-relaxed) minimax optimal design is a completely randomized design that assigns $N_1^{\text{opt}}$, $N_0^{\text{opt}}$, $\{N_{w_{t, \infty}}^{\text{opt}}\}_{t=2}^{T}$, $\{N_{w_{1, t}}^{\text{opt}}\}_{t=2}^T$ to always-treated, always-control, treatment onset at $t$, and treatment offset at $t$ trajectories, where
    \begin{align*}
        N_0^{opt} &= N \cdot \left(1 + \sqrt{\frac{2}{3}}\frac{c_2}{d_2} + \sqrt{\frac{2}{3}}\sum_{t=2}^T c_t + \sqrt{\frac{1}{3}}\frac{c_2}{d_2}\sum_{t=2}^T d_t \right)^{-1} \\
        N_1^{opt} &= \sqrt{\frac{2}{3}}\frac{c_2}{d_2}N_0^{opt} \\
        \quad N_{w_{t, \infty}}^{\text{opt}} &= \sqrt{\frac{2}{3}}N_0^{opt} c_t \\
        \quad N_{w_{1, t}}^{\text{opt}} &= \sqrt{\frac{1}{2}}N_1^{opt} d_t \\
        c_t &= 
        \left[ 
        \frac{1}{c_{t+1}^2} + \frac{1}{(1 + \sqrt{\frac{2}{3}} \sum_{l=t+1}^T c_l)^2}
        \right]^{-1/2}
        ,\quad d_t = 
        \left[ 
        \frac{1}{d_{t+1}^2} + \frac{1}{(1 + \sqrt{\frac{1}{2}} \sum_{l=t+1}^T d_l)^2}
        \right]^{-1/2} \\
        c_T &= d_T = 1
    \end{align*}
\end{theorem}
\begin{proof}
    See Appendix \ref{app:proofs:minimax}.
\end{proof}
In practical implementation, we would need to apply a rounding rule, e.g., a floor, to obtain integer-valued $N_0^{opt}, N_1^{opt}, N_{w_{t, \infty}}^{opt}, N_{w_{1,t}}^{opt}$. Given a $T$, the total budget $N$ must be large enough such that, even after rounding, every assignment trajectory has a non-zero number of decision-makers assigned to it. In Figure \ref{fig:min-sample-size}, we show the relationship between the minimum $N$ budget and $T$.

\begin{figure}[ht!]
    \centering
    \includegraphics[width=0.6\linewidth]{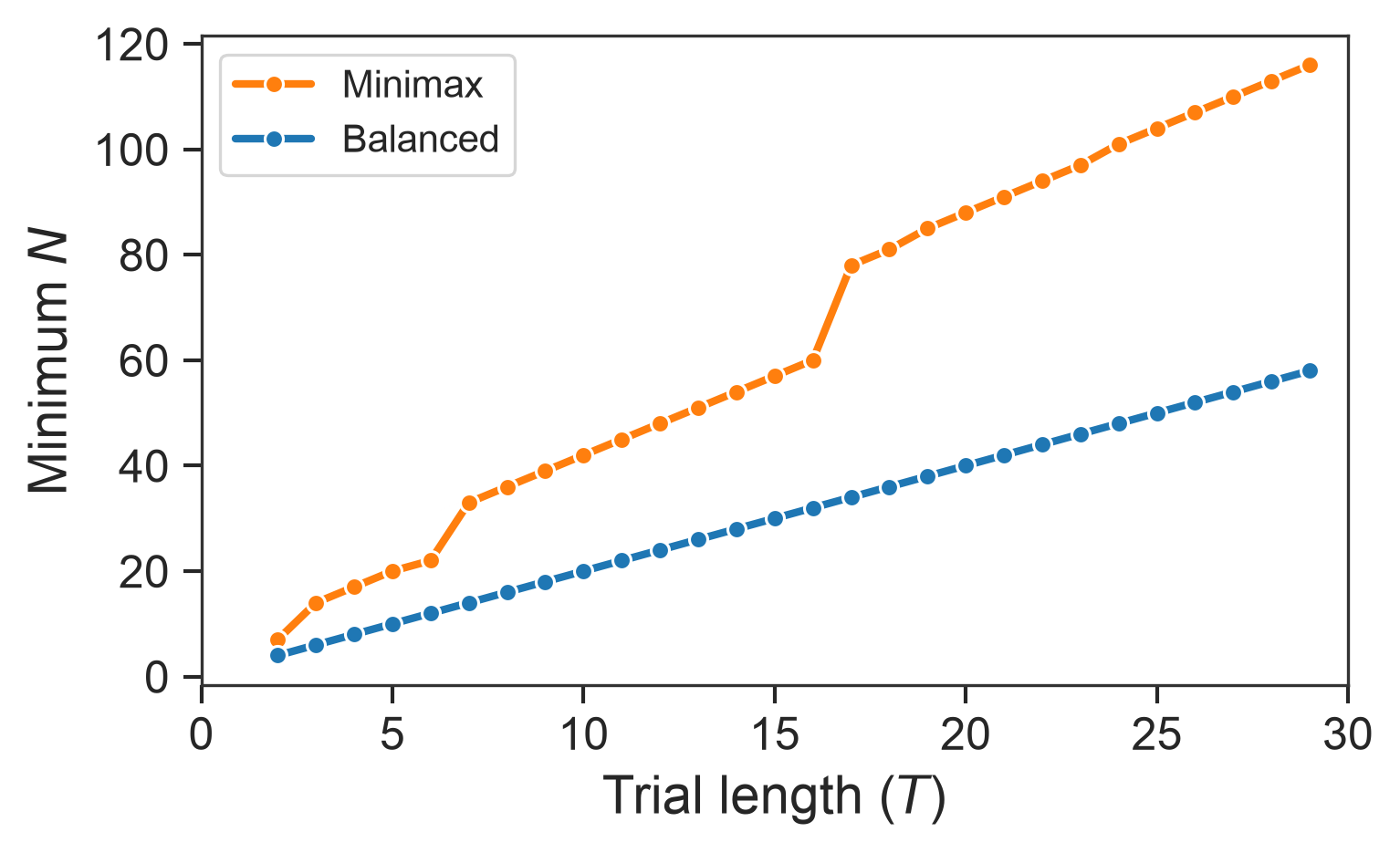}
    \caption{\textbf{Minimum $N$ to guarantee non-empty arms.} We show the minimum $N$ needed in the minimax stepped double wedge design to ensure that each assignment trajectory gets a non-zero number of decision-makers, across varying $T$. We compare this to a baseline stepped double wedge design that is not minimax, where the minimum $N$ for any $T$ is exactly $2T$.}
    \label{fig:min-sample-size}
\end{figure}

\section{Comparing the stepped double wedge to alternatives}\label{sec:alternatives}
In this section, we examine the tradeoffs we make when choosing to run a stepped double wedge design in comparison to two alternative designs: a design where all decision-makers are randomized to either always-assisted or never-assisted, and a per-decision randomized design where each decision is randomized to be assisted or unassisted. We choose these two designs as alternatives for comparison because they were two of the common approaches found in our literature review of published randomized trials of algorithm assistance (Appendix \ref{app:sec:elicit}).

\subsection{Sample efficiency loss (stepped double wedge versus a decision-maker-randomized design)}
We propose running the stepped double wedge trial in order to facilitate estimating $\delta_t$, $\Delta_t$, and $\lambda_t$. To estimate $\tau_t$ only, it would be more efficient to run a standard decision-maker randomized trial where all decision-makers are randomized to either always receive assistance or never receive assistance, rather than having some decision-makers be allocated to staggered treatment onsets and offsets. One way to quantify the loss in efficiency is to examine, at each $t$, the fraction of the decision-makers that are assigned to treatment trajectories that are used to estimate $\tau_t$. More concretely, we calculate $(N_{1, t} + N_{0, t}) / N$: in a decision-maker randomized trial without any staggered onsets and offsets, this ratio will always be equal to $1$. However, in our stepped double wedge trial design, the ratio decreases as $t$ increases (Figure \ref{fig:tradeoff-n}). In settings where the budget of decision-makers, $N$, is small, trialists may want to consider the parallel decision-maker randomized design. On the other hand, when $N$ is large, and when there may be prior reasons to believe that behavior dynamics under repeated algorithm exposure are to be expected, e.g. from anecdotal evidence or pilot studies, then a stepped double wedge design may be preferred.

\begin{figure}[ht!]
    \centering
    \includegraphics[width=0.6\linewidth]{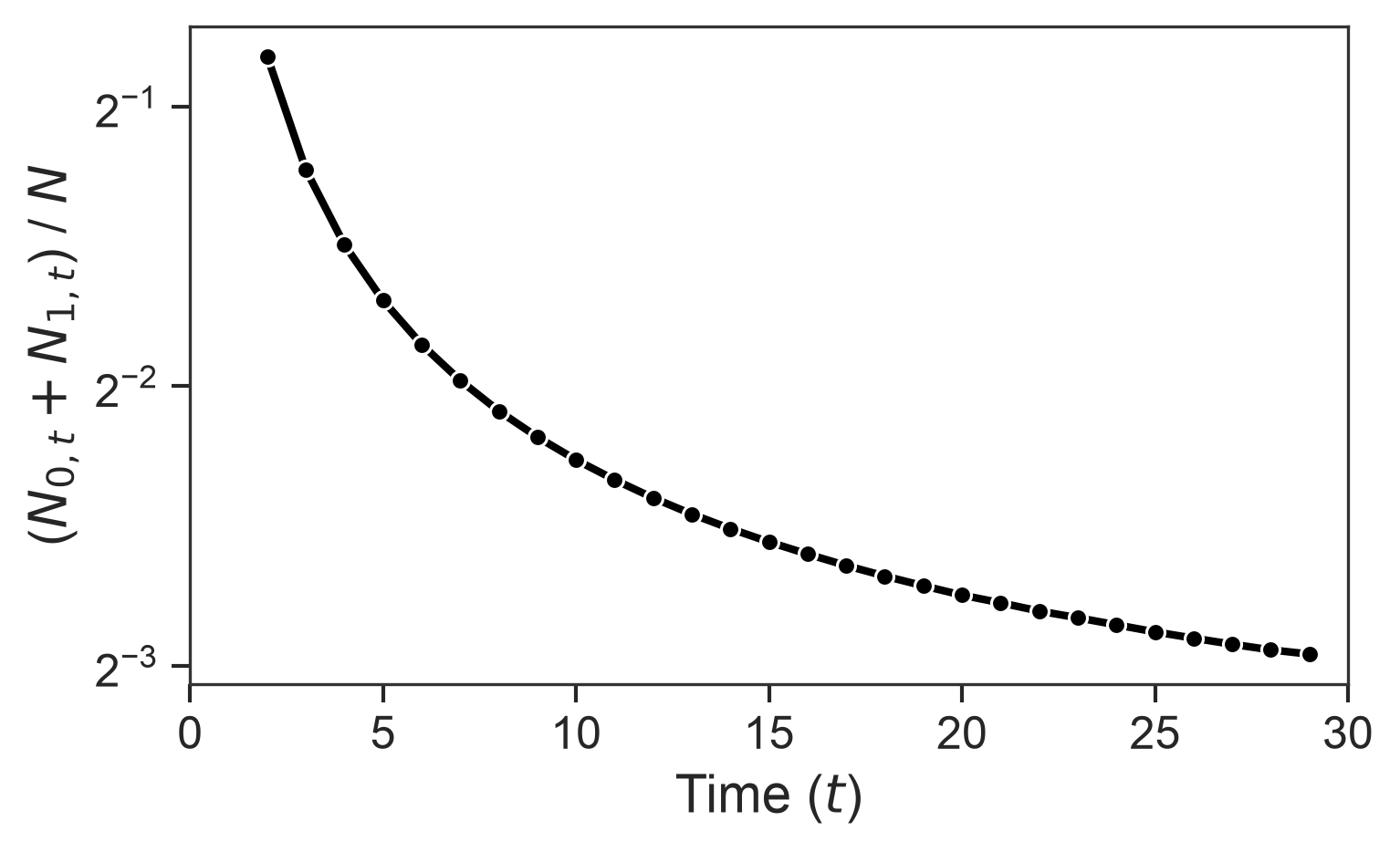}
    \caption{\textbf{Sample efficiency tradeoff of running a stepped double wedge design to estimate $\tau_t$.} As the time horizon of the experiment increases, more decision-makers are allocated to staggered treatment onset and offset trajectories, and fewer are able to be used to estimate $\tau_t$, thus decreasing sample efficiency. However, a stepped double wedge design may still be preferred to facilitate estimation of $\delta_t$, $\Delta_t$, and $\lambda_t$.}
    \label{fig:tradeoff-n}
\end{figure}

\subsection{Bias-variance tradeoffs (stepped double wedge versus a decision-randomized design)}

Another alternative design is one where each decision is randomized to be assisted or unassisted. Note that, from here onwards, we transition to focus on a population-level definition for $\tau_t$, i.e. $\tau_t := \mathbb{E}[Y_t(D(\mathbf{1}^{\leq t})) - Y_t(D(\mathbf{0}^{\leq t}))]$ which results in more interpretable bias expressions for illustration purposes. We also introduce $\tau$ as the global effect averaged over all $t$, $\tau := 1 / T \sum_{t=1}^{T} \tau_t$.

When a per-decision randomized design is performed, a canonical simple estimator is the Horvitz-Thompson difference-in-means estimator that averages across all decisions,
\begin{align*}
    \hat{\tau}^{dr} &=  
        \frac{1}{NT} \sum_{i=1}^N \sum_{t=1}^T 
         \frac{Y_{it}Z_{it}}{p}
         -
         \frac{Y_{it}(1 - Z_{it})}{1 - p}
\end{align*}
where $dr$ stands for ``decision-randomized" and $p := P(Z_{it} = 1)$. The estimator $\hat{\tau}^{dr}$ is generally biased for $\tau$ unless there is no between-decision interference. For instance, it is straightforward to show that bias is guaranteed to be non-zero under a strict monotonicity condition.

\begin{proposition}\label{prop:bias-monotonic}
     If the potential outcomes $Y_{it}(D_{it}(\mathbf{Z}^{\leq t}))$ are strictly monotonically increasing or decreasing in $\sum_{t'=1}^t Z_{it'}$, and the decision-randomized design satisfies $P(\mathbf{Z}_i = \mathbf{z}) > 0$ for all $\mathbf{z} \in \{0, 1\}^T$, then $|\mathbb{E}[\hat{\tau}^{dr}] - \tau| > 0$ for all $T \geq 2$.
\end{proposition}
\begin{proof}
    See Appendix \ref{app:proofs:bias-monotonic}.
\end{proof}

In contrast, we can estimate $\tau_t$ and $\tau$ unbiasedly using the stepped double wedge trial: the timepoint-specific estimate $\hat{\tau}^{sdw}_t$ is unbiased for $\tau_t$, and the averaged estimator $\hat{\tau}^{sdw} = 1 / T \sum_{t=1}^T \hat{\tau}^{sdw}_t$
is unbiased for $\tau$. While the stepped double wedge design facilitates unbiased estimation, it also results in higher variance. This is both due to within decision-maker correlation across different decisions, and because fewer decisions are used for the estimation of $\tau_t$ and $\tau$. Thus, there is a bias-variance tradeoff between choosing the per-decision design versus the stepped double wedge design. 

To show this bias-variance tradeoff more concretely, we first derive exact bias expressions for $\hat{\tau}^{dr}$ across three illustrative data-generative processes. We focus on bias for $\tau_T$, the global effect in the last time period of the experiment, which can be revealing of what the effect of the intervention would be in the time periods following the experiment, and for $\tau$, the averaged global effect. We examine both bias as a function of $T$, and also what value the bias would take on if $T$ were to approach infinity, which helps shed light on what to expect in long experiments. We then use simulations to examine the bias-variance tradeoff when choosing between $\hat{\tau}^{dr}$ and $\hat{\tau}^{sdw}$.

\subsubsection{Bias of the decision-randomized design under three data-generating processes}

Inspired by \cite{Raji2025Prediction}, we define three data-generating processes that demonstrate three different ways in which decision-making behavior may change with repeated exposure to the algorithm -- by \textbf{developing Automation Bias}, by \textbf{developing Alert Fatigue}, and by \textbf{developing Calibrated Reliance}. Whereas \cite{Raji2025Prediction} define data-generating processes over potential decisions, we define ours over both decisions and outcomes. These three data-generating processes differ in how agreement between the potential decision and the algorithm recommendation depends on prior treatment assignments and prior observed outcomes. While the data-generating processes are stylized and real-world processes are likely much more complicated, they nevertheless provide a sense of how even very simple forms of decision-maker behavior adaptation across repeated exposure can lead to non-negligible bias in $\hat{\tau}^{dr}$. For simplicity and tractability, we construct data-generating processes where there is a habituation effect, but no skilling effect (i.e. $\lambda_t = 0$). 

We first introduce some additional notation and assumptions.
Suppose that, for each decision, there exists a correct decision $D^*_{it}$. We define
\begin{align*}
    Y_{it}(\mathbf{Z}_i^{\leq t}) 
    &:= \mathbbm{1} \{ D(\mathbf{Z}_i^{\leq t}) = D_{it}^*\} \\
    A_{it} &:= \mathbbm{1} \{R_{it} = D_{it}^*\} \\
    \mu_A &:= \mathbb{E}[A_{it}] \\
    Q_{it}(\mathbf{Z}_i^{\leq t}) &:= \mathbbm{1} \{D_{it}(\mathbf{Z}_i^{\leq t}) = R_{it}\} \\
    q_{0, i} &:= P(D_{it}(\mathbf{0}^{\leq t}) = R_{it}) \\
    \mathbb{E}[q_{0, i}] &:= q_0 \\
    p &:= P(Z_{it} = 1)
\end{align*}
$Y_{it}(\mathbf{Z}_i^{\leq t})$ is an indicator for whether the potential decision matches the correct decision, and $A_{it}$ is an indicator for whether the algorithm's recommendation matches the correct decision. We additionally assume $A_{it}$ are independent and identically distributed across all $i, t$, with constant $\mu_A$ as the probability that the algorithm's recommendation is correct. $Q_{it}(\mathbf{Z}^{\leq t})$ is an indicator for whether the potential decision matches the algorithm's recommendation, $q_{0, i}$ is the probability of agreement between a potential decision made by decision-maker $i$ with no exposure to the algorithm, and $q_0$ is the mean of $q_{0, i}$ across all decision-makers. Finally, we let $p$ denote the probability of being assigned to treatment in the decision-randomized design.

\begin{example}[Developing Automation Bias]\label{ex:ab-dgp}
\begin{align*}
    p_{Q, \text{ab}}(Z_{it}, S_{it}) &:= q_0 + \beta_Z Z_{it} + \beta_S     S_{it}Z_{it} \\
    S_{it} &:= \rho S_{it-1} + (1-\rho)Z_{it-1}, \quad \text{for $t \geq 2$}, \ S_{i1} = 0 \\
    P(Q_{it}(\mathbf{Z}_{it}) = 1) &= \text{clip}(p_{Q, ab}, 0, 1)
    \end{align*}
    where $\beta_S > 0$, $\rho \in [0, 1)$, $\beta_Z \in \mathbb{R}$. At $\rho$ = 1, there is no dependence on prior exposure, so we exclude this boundary point; the same rationale holds for the two other data-generating processes we introduce below. The $\rho$ scalar weights how much cumulative prior exposure to the algorithm matters, as opposed to the most recent prior exposure. The $\beta_Z$ coefficient captures the immediate effect of algorithm assistance on concordance probability. The $\beta_S$ coefficient captures the influence of prior exposure on concordance probability. We restrict $\beta_S$ to be positive, meaning the decision-maker becomes more likely to concord with the algorithm's recommendation the more times they have interacted with the algorithm, behavior that is indicative of automation bias.
\end{example}
\begin{proposition}[Bias under Automation Bias DGP]\label{prop:ab-bias}
    Assume that the parameters $q_0, \beta_Z, \beta_S$ and the variables $\mathbf{Z}, \mathbf{S}$ take on values such that $p_{Q, {ab}} \in [0, 1]$; that is, the clipping function does not need to be applied. Assume also that $\rho \in [0, 1)$. Then, under the Automation Bias data-generating process,
    \begin{align*}
        \mathbb{E}_{\mathcal{D}_{ab}}[\hat{\tau}^{dr}] - \tau_T 
            &= -(2\mu_A - 1)\beta_S \left( \left(1 - \rho^{T-1}\right) - p \left( 1 - \frac{\rho^T - 1}{T(\rho - 1)}   \right) \right) \\
        \mathbb{E}_{\mathbf{Z},\mathcal{D}_{ab}}[\hat{\tau}^{dr}] - \tau 
            &= -(1-p)(2\mu_A - 1)\beta_S\left(1 - \frac{\rho^{T} - 1}{T(\rho - 1)}\right)
    \end{align*}
    In the limit as $T \rightarrow \infty$, 
    \begin{align*}
        \mathbb{E}_{\mathcal{D}_{ab}}[\hat{\tau}^{dr}] - \tau_T 
        =
        \mathbb{E}_{\mathcal{D}_{ab}}[\hat{\tau}^{dr}] - \tau
        = -(1-p)(2\mu_A - 1)\beta_S
    \end{align*}
\end{proposition}
\begin{proof}
    See Appendix \ref{app:proofs:bias}.
\end{proof}

Notice that, if $\mu_A > 0.5$ (the algorithm is better than random guessing), $\beta_S > 0$, and $\rho \in (0, 1)$, the bias $\mathbb{E}_{\mathbf{Z}, \mathcal{D}_{ab}}[\hat{\tau}^{dr}] - \tau$ becomes more negative as the duration of the experiment increases. Additionally, the absolute magnitudes of the bias for $\tau_T$, the bias for $\tau$, and the bias in the limit increase with both $\mu_A$ and $\beta_S$.

\begin{example}[Developing Alert Fatigue]\label{ex:af-dgp}
    \begin{align*}
        p_{Q, \text{af}}(Z_{it}, S_{it}) &:= q_0 + \beta_Z Z_{it} + \beta_S Z_{it}S_{it} \\
        S_{it} &= \rho S_{it-1} + (1-\rho) Z_{it-1}(1 - A_{it-1}), \quad \text{for $t \geq 2$}, \ S_{i1} = 0 \\
        P(Q_{it}(\mathbf{Z}_{it}) = 1) & = \text{clip}(p_{Q, af}, 0, 1)
    \end{align*}
    where $\beta_S < 0, \ \rho \in [0, 1)$, $\beta_Z \in \mathbb{R}$. Here, we conceptualize $S_{it}$ as latent ``reliance", where the reliance updates (specifically, decreases) each time the decision-maker receives algorithm assistance and the recommendation is incorrect. Under this DGP, the decision-maker will concord with the algorithm less and less over time, as long as the algorithm makes non-zero mistakes. Note that this DGP necessarily assumes that the decision-maker operates in a setting with rapid, observable feedback and can assess the correctness of the algorithm's recommendation before the next time they need to make a decision.
\end{example}

\begin{proposition}[Bias under Alert Fatigue DGP]\label{prop:af-bias}
   Assume that the parameters $q_0, \beta_Z, \beta_S$ and the variables $\mathbf{Z}, \mathbf{S}$ take on values such that $p_{Q, {af}} \in [0, 1]$; that is, the clipping function does not need to be applied. Assume also that $\rho \in [0, 1)$. Then, under the Alert Fatigue data-generating process,
    \begin{align*}
        \mathbb{E}_{\mathcal{D}_{af}}[\hat{\tau}^{dr}] - \tau_T 
            &= -(2\mu_A - 1)\beta_S (1 - \mu_A)
                \left(\left(1 - \rho^{T-1} \right)
                - p\left( 1 - \frac{\rho^T - 1}{T(\rho - 1)} \right) 
                \right) \\
        \mathbb{E}_{\mathcal{D}_{af}}[\hat{\tau}^{dr}] - \tau 
            &= -(1-p)(2\mu_A - 1)\beta_S (1 - \mu_A)
                \left(1 - \frac{\rho^{T} - 1}{T(\rho - 1)}\right)
    \end{align*}
    In the limit as $T \rightarrow \infty$,
    \begin{align*}
        \mathbb{E}_{\mathcal{D}_{af}}[\hat{\tau}^{dr}] - \tau_T 
        =
        \mathbb{E}_{\mathcal{D}_{af}}[\hat{\tau}^{dr}] - \tau
        &= -(1-p)(2\mu_A - 1)\beta_S (1 - \mu_A)
    \end{align*}
\end{proposition}
\begin{proof}
    See Appendix \ref{app:proofs:bias}.
\end{proof}

Notice that, if $\mu_A > 0.5$, $\beta_S < 0$, $\rho \in (0, 1)$, the bias $\mathbb{E}_{\mathcal{D}_{af}}[\hat{\tau}^{dr}]-\tau$ becomes more positive as the duration of the experiment increases and as $|\beta_S|$ increases. In the limit as $T \rightarrow \infty$, the absolute magnitude of the bias is maximal at $\mu_A = 0.75$ if $\mu_A > 0.5$. That is, the per-decision-randomized design is most biased in the limit for a better-than-random algorithm if it makes correct recommendations $75\%$ of the time.

\begin{example}[Developing Calibrated Reliance]\label{ex:cr-dgp}
    \begin{align*}
            p_{Q, \text{cr}}(Z_{it}, S_{it}) &:= q_0 + \beta_Z Z_{it} + \beta_S Z_{it}S_{it}(2A_{it} - 1) \\
            S_{it} &= \rho S_{it-1} + (1 - \rho) Z_{it-1}, \quad \text{for $t \geq 2$}, \ S_{i1} = 0 \\
            P(Q_{it}(\mathbf{Z}_{it}) = 1) & = \text{clip}(p_{Q, cr}, 0, 1)
        \end{align*}
    where $\beta_S > 0, \ \rho \in [0, 1)$, $\beta_Z \in \mathbb{R}$. Here, we conceptualize $S_{it}$ as latent ``calibration", where the calibration updates (specifically, increases) each time the decision-maker interacts with the algorithm. Under this DGP, the decision-maker will gradually concord with the algorithm more when the algorithm is correct and less when the algorithm is incorrect, thus becoming a more calibrated arbiter of the algorithm's recommendations.
\end{example}

\begin{proposition}[Bias under Calibrated Reliance DGP]\label{prop:cr-bias}
    Assume that the parameters $q_0, \beta_Z, \beta_S$ and the variables $\mathbf{Z}, \mathbf{S}$ take on values such that $p_{Q, {cr}} \in [0, 1]$; that is, the clipping function does not need to be applied. Assume also that $\rho \in [0, 1)$. Then, under the Calibrated Reliance data-generating process,
    \begin{align*}
        \mathbb{E}_{\mathcal{D}_{cr}}[\hat{\tau}^{dr}] - \tau_T 
            &=-\beta_S 
                \left (
                \left(1 - \rho^{T-1} \right)
                - p 
                    \left( 1 - \frac{\rho^T - 1}{T(\rho - 1)} \right) 
                \right) \\
        \mathbb{E}_{\mathcal{D}_{cr}}[\hat{\tau}^{dr}] - \tau 
            &= -(1-p)\beta_S
                \left(1 - \frac{\rho^{T} - 1}{T(\rho - 1)}\right)
    \end{align*}
    In the limit as $T \rightarrow \infty$,
    \begin{align*}
        \mathbb{E}_{\mathcal{D}_{cr}}[\hat{\tau}^{dr}] - \tau_T 
        =
        \mathbb{E}_{\mathcal{D}_{cr}}[\hat{\tau}^{dr}] - \tau 
        = -(1-p)\beta_S
    \end{align*}
\end{proposition}
\begin{proof}
    See Appendix \ref{app:proofs:bias}.
\end{proof}
Notice that, if $\beta_S > 0$, $\rho \in (0, 1)$, the bias $\mathbb{E}_{\mathcal{D}_{cr}}[\hat{\tau}^{dr}] - \tau$ becomes more negative as the duration of the experiment increases and as $\beta_S$ increases. Unlike the two other data-generating processes, for the calibrated reliance data-generating process, the biases for $\tau_T$ and for $\tau$ have no dependence on the performance of the algorithm itself, $\mu_A$, but still depend on the carryover coefficient, $\beta_S$.

\subsubsection{Bias-variance tradeoff of switching to the stepped wedge design under three illustrative data-generating processes}

We now illustrate how a bias-variance tradeoff manifests across the three example data-generating processes defined in the previous section (Examples \ref{ex:ab-dgp}, \ref{ex:af-dgp}, \ref{ex:cr-dgp}). We estimate the variance of the estimators through simulation and combine variance estimates with the closed-form bias expressions derived in the previous section to calculate the root mean squared error (RMSE) of the estimators. We then investigate how the RMSE compares between $\hat{\tau}^{dr}$ and $\hat{\tau}^{sdw}$, across four axes: $\beta_S$, the strength of carryover from prior exposure; $\sigma^2_{q0} = \text{Var}(q_{0, i})$, the between decision-maker variance in $q0$; $N$, the total number of decision-makers in the trial; and $T$, the total number of decisions made per decision-maker. The RMSE is calculated for the target estimand $\tau$. We set the parameters of the simulation grid such that the assumptions that $p_{Q, ab}, p_{Q, af}, p_{Q, cr} \in [0, 1]$ all hold, to ensure that the closed-form bias expressions are valid, and we set $p = P(Z_{it} = 1) = 1/2$ in the decision-randomized design.

In Figures \ref{fig:tradeoff-rmse-ab}, \ref{fig:tradeoff-rmse-af}, and \ref{fig:tradeoff-rmse-cr}, we see that, as $\sigma_{q0}^2$ increases, i.e. there is more within decision-maker correlation, the RMSE of $\hat{\tau}^{sdw}$ increases while the RMSE of $\hat{\tau}^{dr}$ remains constant. As $N$ increases, the RMSE of $\hat{\tau}^{sdw}$ decreases more rapidly than the RMSE of $\hat{\tau}^{dr}$. As $T$ increases, the RMSE of $\hat{\tau}^{dr}$ increases when the carryover coefficient $\beta_S$ is large in magnitude (bias dominates), but decreases when $\beta_S$ is small (variance dominates). Across all subplots, larger $|\beta_S|$ corresponds to a larger RMSE for $\hat{\tau}^{dr}$. As a result of these trends, there are some parameter combinations for which the RMSE of $\hat{\tau}^{dr}$ is smaller, and others where the RMSE of $\hat{\tau}^{sdw}$ is favorable. For instance, under the Automation Bias data-generating process, at small $N$, the stepped double wedge design has higher RMSE across all $\beta_S \in \{0.05, 0.1, 0.2\}$, but at large enough $N$ ($\gtrsim 3678, 833, 175$, respectively), it becomes favorable. 

\begin{figure}[ht!]
    \centering
    \includegraphics[width=\linewidth]{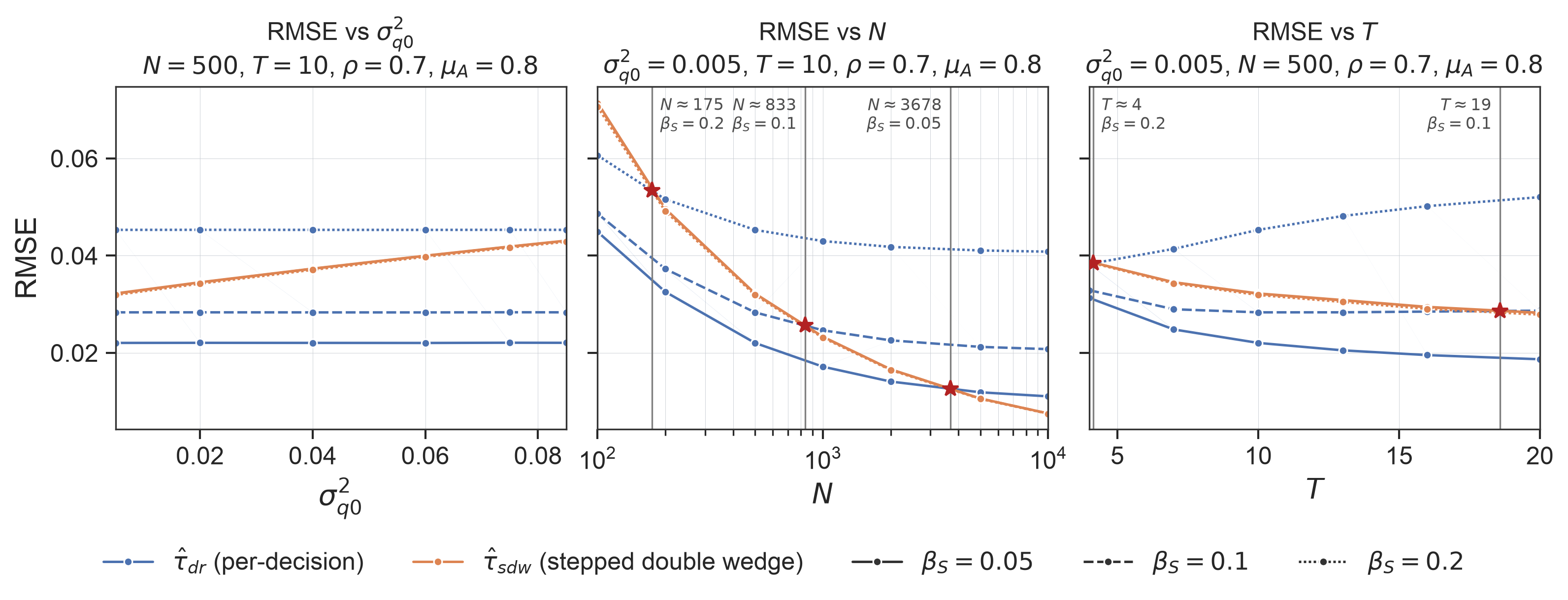}
    \caption{\textbf{RMSE Comparison, Automation Bias DGP.} (left, varying $\sigma_{q0}^2$) $\hat{\tau}^{sdw}$ has consistently higher RMSE than $\hat{\tau}^{dr}$ when $\beta_S = 0.05, 0.1$ and consistently lower RMSE when $\beta_S = 0.2$. (middle, varying $N$) $\hat{\tau}^{sdw}$ has lower RMSE for $\beta_S = 0.05$ when $N \gtrsim 3678$, for $\beta_S = 0.1$ when $N \gtrsim 833$, and for $\beta_S = 0.2$ when $N \gtrsim 175$.  (right, varying $T$) $\hat{\tau}^{sdw}$ has consistently higher RMSE than $\hat{\tau}^{dr}$ when $\beta_S = 0.05$, but crosses from having higher to having lower RMSE at $T \gtrsim 19$ when $\beta_S = 0.1$ and at $T \gtrsim 4$ when $\beta_S = 0.2$.}
    \label{fig:tradeoff-rmse-ab}
\end{figure}
\begin{figure}[ht!]
    \centering
    \includegraphics[width=\linewidth]{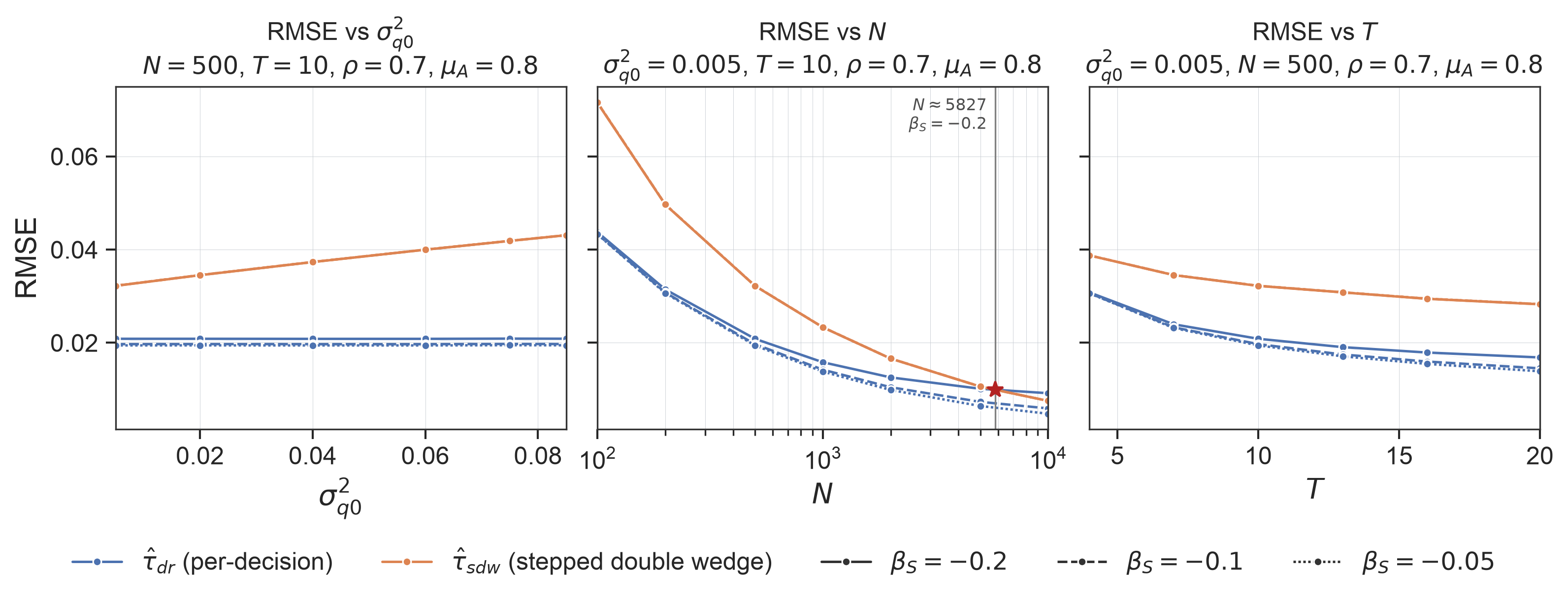}
    \caption{\textbf{RMSE Comparison, Alert Fatigue DGP.} (left, varying $\sigma_{q0}^2$; right, varying $T$) For the simulated parameter grid, the RMSE of $\hat{\tau}^{sdw}$ is always greater than that of $\hat{\tau}^{dr}$. (middle, varying $N$) $\hat{\tau}^{sdw}$ has consistently higher RMSE than $\hat{\tau}^{dr}$ when $\beta_S = -0.05, -0.1$, but crosses to lower RMSE at $N \gtrsim 5827$ when $\beta = -0.2$.}
    \label{fig:tradeoff-rmse-af}
\end{figure}
\begin{figure}[ht!]
    \centering
    \includegraphics[width=\linewidth]{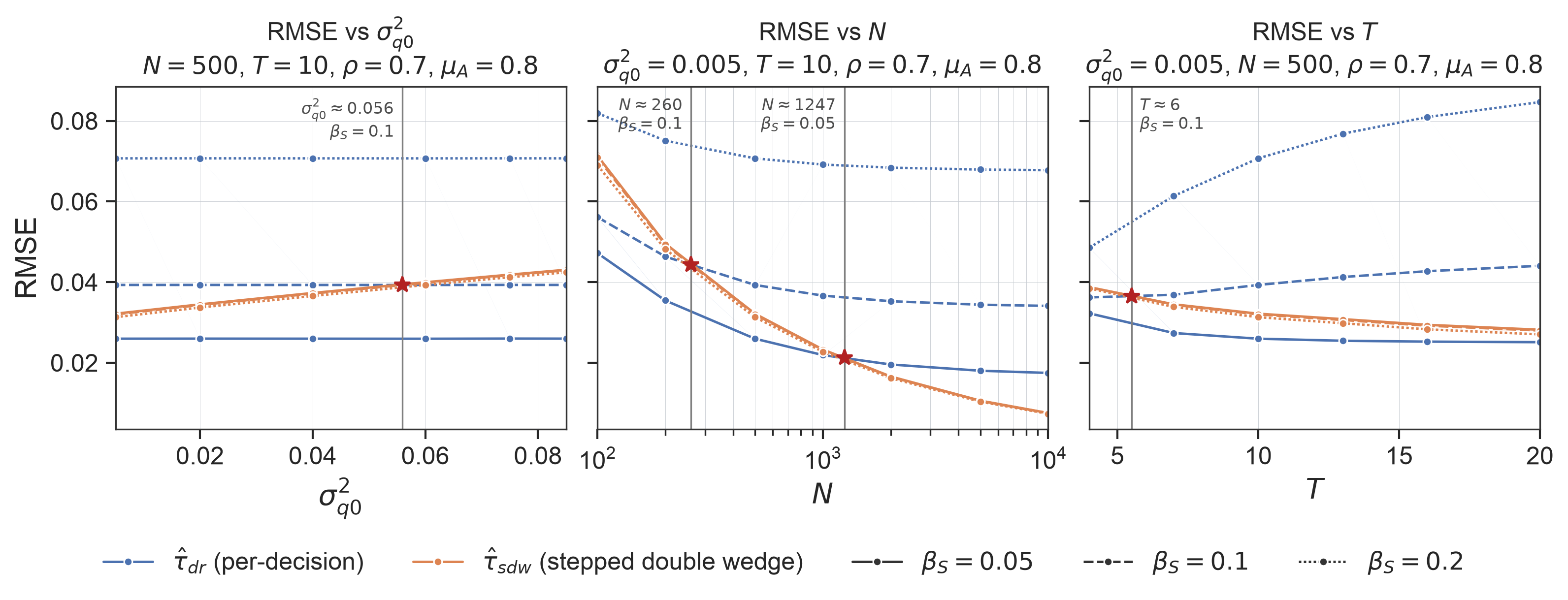}
    \caption{\textbf{RMSE Comparison, Calibrated Reliance DGP} (left, varying $\sigma_{q0}^2$) $\hat{\tau}^{sdw}$ has consistently higher RMSE than $\hat{\tau}^{dr}$ when $\beta_S = 0.05$, crosses from having lower to having higher RMSE at $\sigma_{q0}^2 \approx 0.056$ when $\beta_S = 0.1$, and has consistently lower RMSE when $\beta_S = 0.2$. (middle, varying $N$) $\hat{\tau}^{sdw}$ has lower RMSE for $\beta_S = 0.05$ when $N \gtrsim 1247$, for $\beta_S = 0.1$ when $N \gtrsim 260$, and for $\beta_S = 0.2$ consistenly. (right, varying $T$) $\hat{\tau}^{sdw}$ has consistently higher RMSE than $\hat{\tau}^{dr}$ when $\beta_S = 0.05$, crosses from higher to lower RMSE at $T \gtrsim 6$ when $\beta_S = 0.1$, and has consistently lower RMSE when $\beta_S = 0.2$.}
    \label{fig:tradeoff-rmse-cr}
\end{figure}

We emphasize that, despite the existence of this bias-variance tradeoff, it holds that there is no straightforward way to estimate $\delta_t$, $\Delta_t$, and $\lambda_t$ under the per-decision-randomized design without imposing parametric assumptions. Thus, it may still be desirable to use a stepped double wedge design even if it does not have the lower RMSE for $\tau$.

\section{Considerations for implementation in healthcare settings}\label{sec:health-ops}
Healthcare settings can pose certain challenges to conducting our proposed stepped double wedge design, both from an operational logistics standpoint, and from a statistical analysis standpoint. We highlight a few of these challenges in this section.

One question is how to operationalize the treatment assignment randomization. Often, in trials of clinical decision support tools, randomization occurs at the patient-level. In these instances, treatment assignment can be stored within the patient's electronic health record, either as an additional attribute for the patient, or as a function of an existing patient identifier, e.g. whether the patient's medical record number is even or odd. For randomization at the decision-maker level with staggered onsets and offsets, however, the logistics of randomization can be more challenging. One option is to, \textit{a priori}, store a particular treatment assignment trajectory for each provider, using a mapping between the provider's unique ID and the assignment trajectory, in addition to a provider-specific index to track position on the trajectory. Then, each time the provider encounters a patient and needs to make a decision, they receive algorithm assistance according to their assigned trajectory, and the tracking index advances. 

Another consideration is interference between decision-makers. For example, providers may influence one another's perception of algorithm assistance, both positively and negatively. Competing for a finite resource capacity can also lead to interference between decision-makers when the algorithm produces decisions or recommendations that encourage resource consumption, e.g. ordering lab tests or transferring to the ICU where bed availability is limited. When there is interference between decision-makers, the randomization may need to occur at a cluster- or site-level, such that decision-makers assigned to different treatment assignment trajectories are unlikely to interfere with one another, and such that each cluster itself functions as a ``universe" in which only the cluster-specific assignment trajectory exists.

Finally, a single patient may interact with the same provider more than once. Likewise, a single patient may interact with multiple different providers. In both of these cases, there is additional correlation across decisions that needs to be taken into account when analyzing the trial data.

\section{Conclusion}
We present recommendations for target effect estimands (including conditional effect estimands) and for randomized trial design when evaluating algorithm assistance interventions in the presence of repeated algorithm exposure. We motivate our estimands as necessary for obtaining a better understanding of the impact of algorithm assistance, including how the algorithm intervention could be modified to improve its effect, and we propose our design as a straightforward way to target those estimands. Particularly in medicine, where many trialists and hospitals are familiar with the idea of stepped wedge designs, our stepped double wedge design is an intuitive extension. Additionally, we provide some preliminary guidance on how to decide when running our proposed stepped double wedge design to facilitate estimating immediate, habituation, and skilling effects  is ``worth it", in terms of sample efficiency and bias-variance tradeoffs.

There are many directions for future work. We focus on algorithm assistance in the form of a discrete recommended decision or action, where the algorithm fixed over the course of the trial and is not personalized to the decision-maker that is using them. Large language models violate all of these assumptions -- they provide long-form free text output that depends on user prompts, and they can update their behavior and output as they accumulate context about the user's preferences and goals. A natural extension of our work is to consider how to define estimands and estimate effects when not just the human decision-maker's behavior adapts over time, but when the algorithm and the human co-adapt. In addition, we do not explore estimating long-term effects, e.g. $\lim_{t \rightarrow \infty} \tau_t, \Delta_t$. While these effects can be of interest, since they capture what would happen if we were to deploy algorithm assistance into perpetuity as the new status quo, they also generally require strong parametric assumptions or extremely large datasets, which we wanted to avoid within this paper in order to make the framework relevant to the medical setting, where sample size can be limited. Lastly, as mentioned in Section \ref{sec:health-ops} above, we assume no interference between decision-makers, which often does not hold. Future work includes considering the advantages and disadvantages of running the stepped double wedge design as a clustered design, where clusters of decision-makers that are unlikely to interfere with one another are assigned to each of the different possible assignment trajectories. 

\clearpage

\bibliography{refs}

\clearpage

\appendix
\section{Elicit-assisted review of trial designs used for evaluating algorithm-assisted decision-making}\label{app:sec:elicit}
Elicit is an AI tool for performing literature reviews. We used Elicit's systematic review workflow to assist with reviewing the literature on randomized trial for algorithm-assisted decision-making. This review was performed in the style of a rapid review, with a single AI reviewer and single human reviewer. We used Elicit to conduct the initial paper search, abstract screening, and full text screening. Data extraction from the papers that passed through all screening steps was first performed by Elicit, then manually reviewed, and if needed, corrected, by M.W. 

\subsection{Review methodology}
\paragraph{Paper search (Elicit)} We used the following instruction to prompt Elicit to conduct the initial search for paper sources.

\begin{tcolorbox}[colback=blue!5, colframe=blue!50!black, arc=4mm, boxrule=0.5pt]
\footnotesize{
``What randomized controlled trials or randomized field experiments have been performed to evaluate algorithm-assisted human decision-making, and do they show that algorithm-assisted decision-making improves decision quality (in terms of correctness, efficiency, and/or fairness)? An algorithm is any statistical model, predictive machine learning model, artificial intelligence system, or large language model. The human decision-makers can be domain experts, e.g. a doctor or judge, or laypeople. Randomized controlled trials can be real-world or lab / vignette-based experiments."
}
\end{tcolorbox}

We also applied a date filter for papers published between January, 2023 and June, 2026.

\paragraph{Abstract and full text screen (Elicit)} We defined the following screening criteria for the abstract screening step. Note that these criteria are copied verbatim as they were supplied to Elicit, and we have phrased each criterion as a question, which follows Elicit guidelines. Papers are excluded if the answer to the question is "no".  
\begin{tcolorbox}[colback=blue!5, colframe=blue!50!black, arc=4mm, boxrule=0.5pt]
\begin{itemize}
    \footnotesize{
    \item 
    Is this a full research article that conducts an original RCT(s) rather than a conference abstract, editorial, opinion piece, case report, systematic review, or meta-analysis?
    \item 
    Does the intervention involve algorithmic assistance in decision-making using a statistical model, predictive machine learning model, artificial intelligence system, or large language model, as opposed to naive "algorithms" like checklists, guidelines, or protocols?
    \item 
    Does the study evaluate human decision-making with algorithmic assistance where the human has the final say on the decision rather than evaluating autonomous algorithm decision-making that replaces the human?
    \item 
    Does the study conduct at least one randomized controlled trial?
    \item 
    Does at least one randomized controlled trial in the study compare algorithm-assisted decision-making task to unassisted decision-making, where the comparison can be within or between decision-makers?
    \item 
    Does the study measure at least one of the following outcomes: (1) in terms of accuracy/correctness of the decisions, (2) in terms of speed/efficiency of the decisions, (3) in terms of fairness of the decisions, (4) in terms of the human's level trust or reliance in the algorithm?
    \item 
    Does each decision-maker make at least two decisions, and are there some decision-makers that make repeated decisions with algorithm assistance and some decision-makers that make repeated decisions without algorithm assistance?
    \item
    Does the task that the algorithm assists the human with involve making a binary decision or producing a numerical judgment, as opposed to creating content like art, music, or writing?
    }
\end{itemize}
\end{tcolorbox}

The full text screen included the following additional criteria:
\begin{tcolorbox}[colback=blue!5, colframe=blue!50!black, arc=4mm, boxrule=0.5pt]
    \begin{itemize}
    \footnotesize{
    \item
    Is the article a full text publication, as opposed to just a conference abstract that has no accompanying full text?
    \item
    Is the full text, not just the abstract, in English? 
    }
    \end{itemize}
\end{tcolorbox}

\paragraph{Data extraction (Elicit + Manual)}
We defined the following items for Elicit to extract from the full text. The definitions are copied verbatim from the instructions supplied to Elicit.

\begin{tcolorbox}[colback=blue!5, colframe=blue!50!black, arc=4mm, boxrule=0.5pt]
    \footnotesize{
    \begin{itemize}
    \item 
    \textbf{Type of setting in which the study is conducted.} Options are ``real\_world", indicating an experiment that is performed in real deployment contexts, ``lab\_vignette", indicating an experiment performed in a contrived and stylized setting, or ``other".
    \item 
    \textbf{What domain, discipline, or application area that the algorithm-assisted decision-making takes place in. } Options include: ``medicine" for medicine and clinical care (diagnosis, triage, treatment, screening); ``criminal\_justice" for criminal justice and law (bail, sentencing, recidivism, legal review); ``finance" for finance, lending, and insurance (credit, underwriting, claims, fraud, investment); ``hiring" for hiring and human resources (resume screening, selection, performance review); ``education" for education (admissions, grading, dropout/at-risk prediction); ``content\_moderation" for moderating content (flagging fake or harmful content, misinformation); ``business\_management" for business management and operations (forecasting, pricing, supply chain, marketing); ``public\_sector" for public sector and social services (benefits eligibility, child protection, public administration); ``security" for security, defense, and intelligence (threat detection, surveillance analysis, cybersecurity triage); and ``general" for general or abstract tasks (synthetic or domain-neutral experimental tasks with no real-world application domain).
    \item 
    \textbf{The unit on which randomization occurs.} ``cluster\_of\_decision\_makers" for cluster-level randomization, which means that treatment assignment is randomized at the level of a cluster of decision-makers, such that all decision-makers within that cluster receive the same assignment. ``decision\_maker" for decision-maker-level randomization, which means that treatment assignments are uniquely defined on the decision-maker-level – given a decision-maker and their treatment assignment, the individual treatment assignments of each decision that they make is fully determined. ``decision" for decision-level randomization, which means that each individual decision (or decision-recipient, e.g. patient, defendant) is randomized to treatment or control. ``other" otherwise.
    \item
    \textbf{The type of randomized trial design used.} Options are ``parallel" for a between-subjects parallel group design (where participants are randomized to a constant intervention condition), ``crossover" for a within-subjects crossover design (where participants are randomized to a particular sequence of interventions in a particular order), ``one\_group\_pre\_post" for a one-group pretest-posttest design (where participants repeat the same task first without then with AI assistance, in that order, for every task), ``interleaved" for a decision/task-randomized design (where each participant completes multiple tasks and those tasks are individually randomized), ``interleaved\_in\_blocks" for a blocked decision/task-randomized design (where each participant completes multiple tasks and those tasks are randomized in blocks), ``stepped\_wedge" for a stepped wedge design, or "other"
    \end{itemize}
    }
\end{tcolorbox}
\begin{tcolorbox}[colback=blue!5, colframe=blue!50!black, arc=4mm, boxrule=0.5pt]
    \footnotesize{
    \begin{itemize}
    \item 
    \textbf{Repeated exposure handling.} Whether the study took explicit steps to measure exposure-time-varying effects arising from repeated exposure to the algorithm across multiple decisions within decision-makers. "yes" if explicit steps were taken, "no" if no steps were taken, and "maybe" if unsure. Adjusting for correlation due to repeated measures (e.g. using a mixed effects model) but not estimating or dealing with exposure-time-varying effects does not count as a "yes".
    \item
    \textbf{Number of decision-makers enrolled in the trial}. Provide number only if reported in the paper or able to be estimated from reported numbers (i.e., decisions\_per\_decision\_maker $*$ sample\_size\_decision\_makers). If estimated, provide the number with an ``(estimated)" suffix. If the number is not reported and is not estimable, write ``unknown".
    \item
    \textbf{Number of decisions/tasks completed per individual decision-maker across the duration of the trial.} Provide number only if reported in the paper or able to be estimated from reported numbers (i.e., sample\_size\_decisions / sample\_size\_decision\_makers). If estimated, provide the number with an ``(estimated)" suffix. If not estimated, provide the number only. If the number is not reported and is not estimable, write ``unknown".
    \item
    \textbf{Number of decisions overall}. Total number of decisions/tasks completed in the trial across all decision-makers. Provide number only if reported in the paper or able to be estimated from reported numbers (i.e., decisions\_per\_decision\_maker * sample\_size\_decision\_makers). If estimated, provide the number with an "(estimated)" suffix. If not estimated, provide the number only. If the number is not reported and is not estimable, write "unknown".
    \end{itemize}
    }
\end{tcolorbox}
After Elicit performed an initial extraction, M.W. manually reviewed the extracted content and overrode any extractions that were inconsistent with the paper or with the definition of the extraction item. Additionally, papers that did not actually meet the inclusion criteria were excluded. Finally, papers that included multiple randomized trials that met the inclusion criteria were expanded so that extraction was performed separately for each of the trials.

\subsection{Review findings}
Of 5,000 paper sources found by the initial Elicit search, 246 papers containing 274 unique randomized trials were included. See Figure \ref{fig:prisma} for a flow diagram with exclusion counts and reasons and Appendix \ref{sec:elicit-bib} for a full list of the included papers. Note that, prior to manual review, Elicit included 289 papers; after manual review, 43 of these papers were excluded. In addition, manual review led to many corrections in data extraction; in particular, there were 6 corrections for setting type, 17 for application area, 64 for unit of randomization, 95 for trial design, 36 for repeated exposure handling, 71 for number of decision-makers, 84 for number of decisions per decision-maker, and 137 for number of overall decisions. For manual corrections of the repeated exposure handling element, 14/36 were corrected from ``yes/maybe" to ``no", and 22/36 were corrected from ``no" to ``yes".

\begin{figure}
    \centering
    \includegraphics[trim={12cm 0 12cm 0}, width=\linewidth]{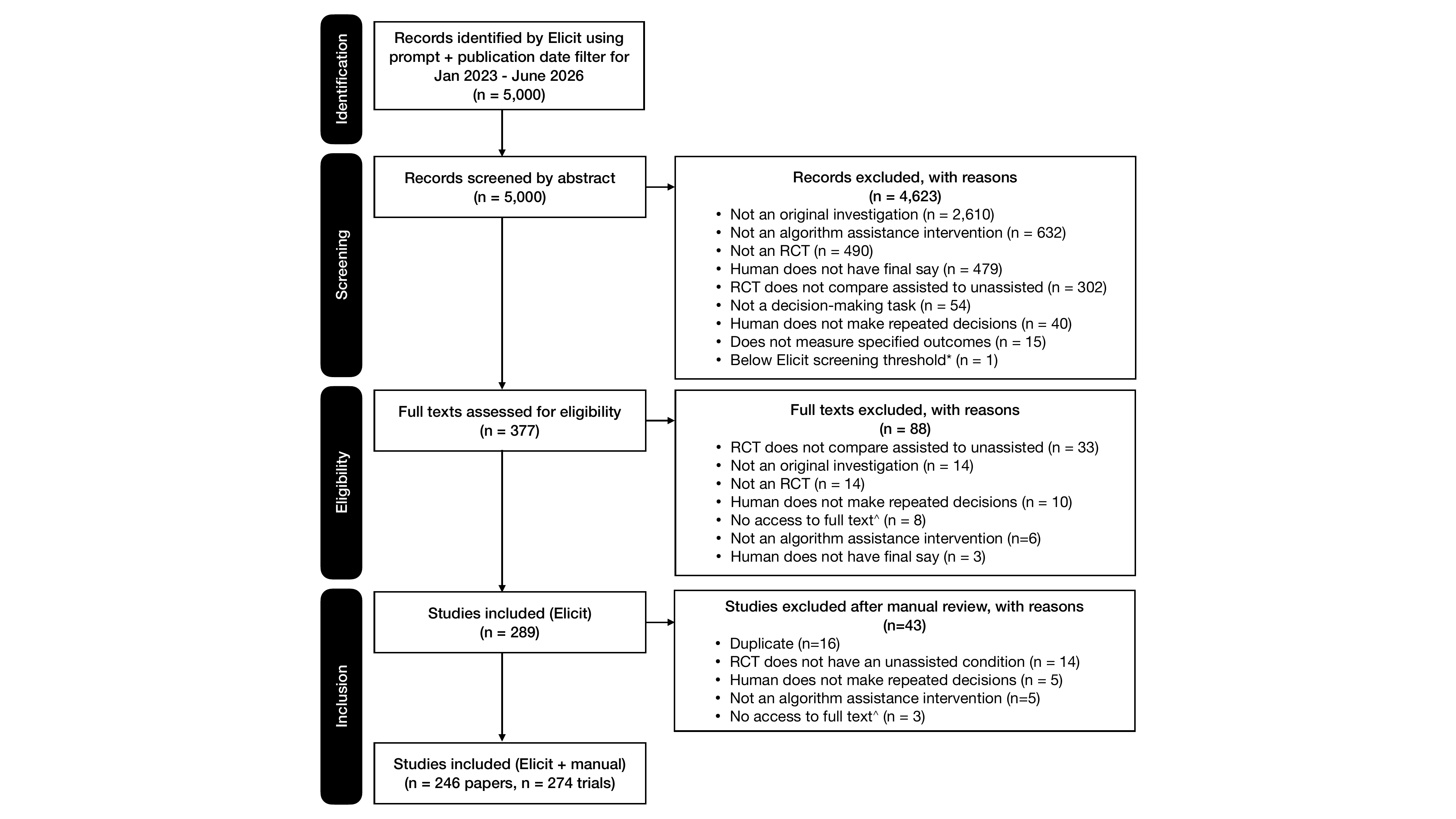}
    \caption{\textbf{PRISMA-style flow diagram for Elicit-assisted literature review}. \\
    {\scriptsize{*``Below Elicit screening threshold" refers to records that received an inclusion ``score" lower than a manually specified threshold. The score reflects how well the record meets the inclusion criteria and is on a scale from 0 (worst) to 5 (best). The threshold was set to 3.} \\
    \^``No access to full text" refers to records that were unable to be accessed due to paywalls and subscription requirements.}}
    \label{fig:prisma}
\end{figure}

Of the 274 trials, 206 (75.2\%) were conducted in a synthetic lab or vignette setting, and 68 (24.8\%) were conducted in a real-world or field setting. Most of the trials were either in a medicine application area (n=149, 54.4\%) or in a general application area (n=62, 22.6\%). Notably, across all trials, only 51 (18.6\%) assessed how the effect of algorithm assistance may evolve with repeated exposure or may have lingering carryover effects after assistance is removed, despite the fact that all trials included at least one decision-maker that was exposed to algorithm assistance more than once.

For the unit of randomization, 189 trials (69.0\%) were randomized at the decision-maker level, 52 (19.0\%) at the decision level, and 11 (4.0\%) at the cluster of decision-makers level. An additional 21 trials (7.7\%) used pre-post designs to estimate the effect of algorithm assistance and therefore did not include randomization, and one trial's randomization design was not specified. Of the trials randomized at the decision-maker level, 124 (65.6\%) used a parallel group design, and 65 (34.4\%) used a crossover design. Of the trials randomized at the decision level, 46 (88.5\%) used an ``interleaved" design, where each decision-maker made an interleaved mixture of assisted and unassisted decisions, and 6 (11.5\%) used an ``interleaved in blocks" design, where each decision-maker made an interleaved mixture of assisted and unassisted decisions, with assisted and unassisted decisions being organized into groups or blocks (e.g., if randomization occurs per day, then the day is a ``block" and all decisions made by all decision-makers on that day are either assisted or unassisted). Of the trials randomized at the cluster of decision-makers level, 8 (72.7\%) used a parallel group design, 2 (18.2\%) used a stepped wedge design, and 1 (9\%) used a crossover design. The median number of decision-makers in a trial was 89 (range: 1-132,199), and the median number of decisions made per decision-maker was 28 (range: 2-4680).

In Figure \ref{fig:elicit-decisions}, we plot the number of decision-makers against the number of decisions per decision-maker for each of the included trials. As the number of decision-makers enrolled in the trial increases, the number of decisions that each decision-maker makes decreases; interestingly, the relationship between the two resembles a power law ($R^2 = 0.22, p < 0.001$). This likely reflects a tradeoff between including more decisions versus more decision-makers to increase statistical power. Additionally, comparing decision-maker-level and decision-level randomization, we observe that trials that randomize at the decision-maker level tend to have more decision-makers (median=100.5, range=2-132,199) and fewer decisions per decision-maker (median=25, range=2-4680), whereas the opposite is true for trials that randomize on the decision-level (number of decision-makers: median=27.5, range=1-1600; number of decisions per decision-maker: median=51.2, range=2.7-2161.7).

\begin{figure}[ht!]
    \centering
    \includegraphics[width=0.7\linewidth]{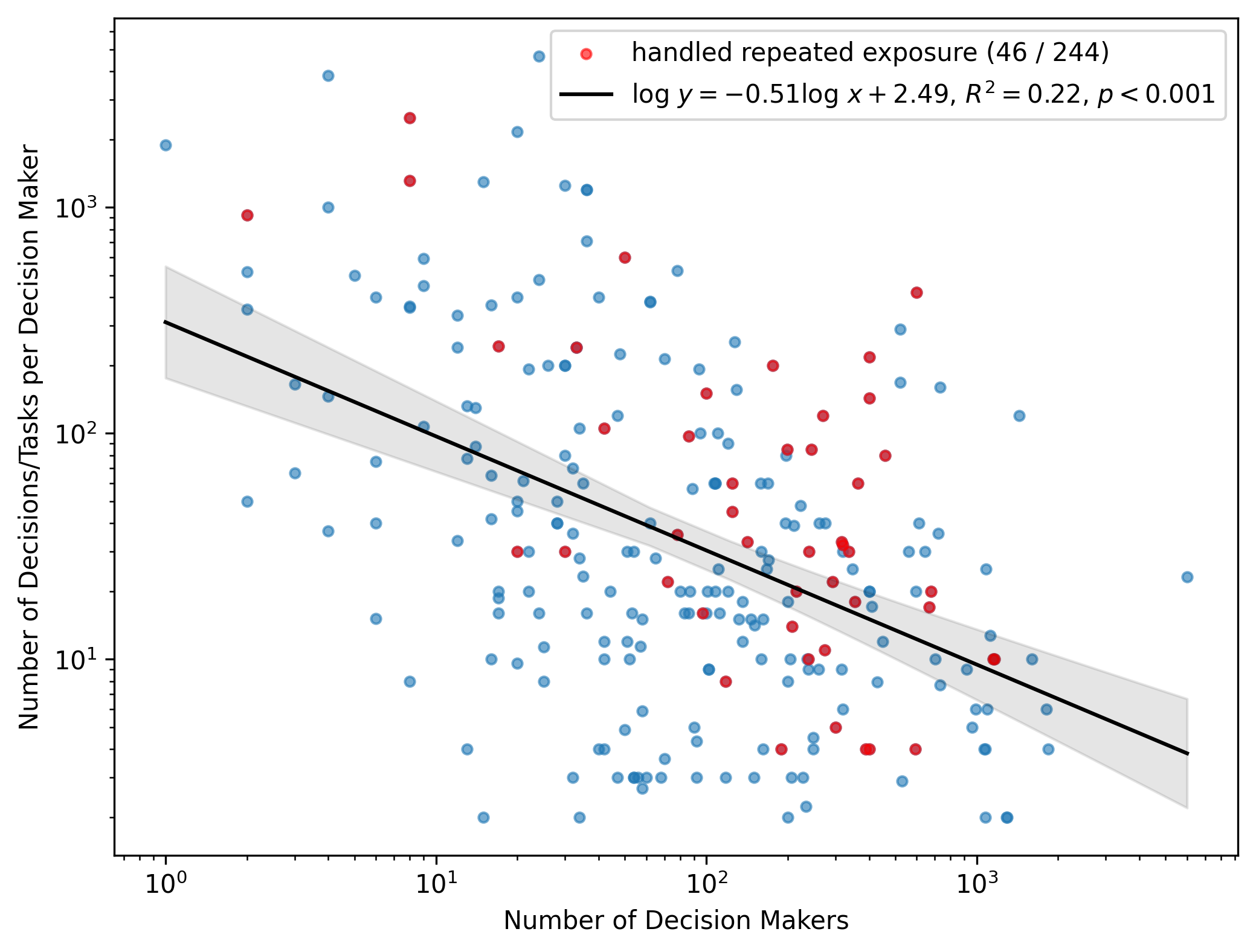}
    \caption{\textbf{Number of decision-makers and decisions-per-decision-maker across trials}. Of the 274 trials, 244 had known or estimable values for both the number of decision-makers and the average number of decisions made per decision-maker. Larger decision-maker sample sizes correlate with smaller numbers of decisions per decision-maker.}
    \label{fig:elicit-decisions}
\end{figure}

\section{Additional effect estimands: conditional effects}\label{app:sec:conditional}

In addition to the estimands presented in the main text ($\tau_t$, $\delta_t$, $\Delta_t$, $\lambda_t$), which measure the expected average difference in outcomes across all decisions, another estimand of interest is the average difference in outcomes conditioned on decisions where the decision-maker takes the algorithm's recommendation. Specifically, we are interested in instances where the decision-maker would have made the non-recommended decision if they had never seen the tool recommendation, but would have yielded to the algorithm's recommendation had they seen it. This conditional effect estimand is related to the complier average causal effect in the instrumental variables literature and is of interest because it reflects whether decision-makers selectively yield to the tool in a way that leads to better outcomes for the decision-recipients.

\subsection{Concordance classes}

We define the potential \textit{concordance} between the decision-maker's decision and the algorithm's recommendation as $Q_{it}(\mathbf{Z}_i^{\leq t}) = \mathbbm{1}\{D_{it}(\rmZ_i^{\leq t}) = R_{it}\}$. Drawing upon the ideas of compliance groups, we define four concordance classes for potential decisions: \textit{always aligned}, \textit{never aligned}, \textit{yielding}, and \textit{contrarian}. An \textit{always aligned} decision aligns with the algorithm's recommendation both if the recommendation were shown and if it were withheld. A \textit{never aligned} decision would not align with the algorithm's recommendation, both if the recommendation were shown and if it were withheld. A \textit{yielding} decision would not align if the recommendation were withheld but would align if the recommendation were shown. Lastly, a \textit{contrarian} decision aligns with the algorithm recommendation if it were withheld, but would not align if it were shown. We use this terminology, as opposed to the compliance class terminology (``always-takers", ``never-takers", ``compliers", and ``defiers") because the core concern is not whether the decision-maker complies with receiving the algorithm assistance intervention (it is assumed that they do) but whether the potential decision they make actually concords with the algorithm's recommendation.

In Table \ref{tab:concordance-cls}, we show the pair of potential concordance values that correspond to each concordance class. Because we allow for within decision-maker interference across time, a decision's concordance class can depend on the decision-maker's prior exposure to the algorithm. We therefore define a concordance class in terms of what decision would be made at time $t$ under $z_t = 1$ (treatment) or $z_t = 0$ (control), given prior exposure vectors, $\mathbf{z}^{< t}$ and $\mathbf{z'}^{< t}$, respectively, and given a fixed algorithm recommendation $r_t$. 

\begin{table}[!ht]
    \centering
    \begin{tabular}{lll}
        \toprule
        Concordance Class & $Q_t(\mathbf{z}^{<t}, 0)$ & $Q_t(\mathbf{z'}^{<t}, 1)$ \\
        \midrule
        Always aligned (AA) & 1 & 1 \\
        Never aligned (NA) & 0 & 0 \\
        Yielding (YI) & 0 & 1 \\
        Contrarian (CO) & 1 & 0 \\
        \bottomrule
    \end{tabular}
    \caption{Concordance classes under prior exposure vectors, $\mathbf{z}^{< t}$ and $\mathbf{z'}^{< t}$, and fixed tool recommendation $r_t$}
    \label{tab:concordance-cls}
\end{table}

To narrow down the space of possible prior exposure vectors, we draw special attention to the concordance classes under two particular pairs of $\mathbf{z}^{<t}$ and $\mathbf{z'}^{<t}$. First, we define the \textit{habituated concordance class} as follows, which uses the pair  $\mathbf{z}^{<t} = \mathbf{0}^{< t}$ and $ \mathbf{z'}^{<t} = \mathbf{1}^{< t}$.

\begin{definition}[Habituated concordance class]
    The \textit{habituated concordance class} is the concordance class identified by the pair of concordance values, $Q(\mathbf{0}^{\leq t})$ and $Q(\mathbf{1}^{\leq t})$. We denote the habituated concordance class at time $t$ as $H_t \in \{\text{AA}, \text{NA}, \text{YI}, \text{CO}\}$.
\end{definition}

In addition to the habituated concordance class, we also define an \textit{immediate concordance class}, which sets the pair of prior exposure vectors to $\mathbf{z}^{<t} = \mathbf{z'}^{<t} = \mathbf{0}^{<t}$.

\begin{definition}[Immediate concordance class]
    The \textit{immediate concordance class} is the concordance class identified by the pair of concordance values, $Q(\mathbf{0}^{\leq t})$ and $Q(\mathbf{0}^{<t}, 1)$. We denote the immediate concordance class at time $t$ as $I_t \in \{\text{AA}, \text{NA}, \text{YI}, \text{CO}\}$.
\end{definition}

\subsection{Conditional effects with habituated and immediate yielding}

We now introduce two conditional effect estimands: the global treatment effect conditioned on the habituated concordance class being yielding and the immediate treatment effect conditioned on the immediate concordance class being yielding. The conditional global treatment effect isolates the differences in outcomes within the subgroup of decisions where the decision-maker changes their decision to be concordant with the algorithm upon sustained exposure to the algorithm. Similarly, the conditional immediate treatment effect isolates the difference in outcomes within the subgroup of decisions where the decision-maker changes their decision to be concordant with the algorithm upon first exposure to the algorithm's assistance.

\begin{definition}[Global treatment effect with habituated yielding]
\begin{align}
    \tau_t^{\text{YI}} &:= 
        \mathbb{E}
    \left[ 
        Y_{t}(D_{t}(\mathbf{1}^{\leq t}))
        - Y_{t}(D_{t}(\mathbf{0}^{\leq t}))
    \mid Q_t(\mathbf{1}^{\leq t})=1, Q_t(\mathbf{0}^{\leq t}) = 0 \right] \\
    &=         \mathbb{E}
    \left[ 
        Y_{t}(D_{t}(\mathbf{1}^{\leq t}))
        - Y_{t}(D_{t}(\mathbf{0}^{\leq t}))
    \mid H_t = \text{YI} \right]
\end{align}
\end{definition}

\begin{definition}[Immediate treatment effect with immediate yielding]
\begin{align}
    \delta_t^{\text{YI}} &:= 
        \mathbb{E}
    \left[ 
        Y_{t}(D_{t}(\mathbf{0}^{< t}, 1))
        - Y_{t}(D_{t}(\mathbf{0}^{\leq t}))
    \mid Q_t(\mathbf{0}^{< t}, 1)=1, Q_t(\mathbf{0}^{\leq t}) = 0 \right] \\
    &= \mathbb{E}
    \left[ 
        Y_{t}(D_{t}(\mathbf{0}^{< t}, 1))
        - Y_{t}(D_{t}(\mathbf{0}^{\leq t}))
    \mid I_t = \text{YI} \right] 
\end{align}
\end{definition}

To make the above estimands identifiable, we impose monotonicity and positivity assumptions analogous to those used to identify complier average causal effects in the instrumental variables literature.

\begin{assumption}[Positivity]\label{assump:positivity}
    $P(H_t = YI) > 0$, $P(I_t = YI) > 0$ for all $t = 1, \dots, T$.
\end{assumption}

\begin{assumption}[Monotonicity, i.e. no contrarian decisions]\label{assump:monotonicity}
    $Q_t(\mathbf{z}^{<t}, 1) - Q_t(\mathbf{z'}^{<t}, 0) \geq 0$ for all $t=1, \dots, T$, $i = 1, \dots, N$ and all $\mathbf{z}^{<t}$, $\mathbf{z'}^{<t}$.
\end{assumption}

Under Assumptions \ref{assump:positivity} and \ref{assump:monotonicity}, as well as the previous Assumptions \ref{assump-partial-interf} -\ref{assump:ignorability}, we have the following identification result.

\begin{proposition}[Identification of conditional effects]\label{prop:cond-ident}
\begin{align*}
    \tau_t^{\text{YI}} &=
        \frac{\tau_t}{P(H_t = \text{YI})} = 
        \frac
        {
        \mathbb{E}[Y_t \mid \mathbf{Z}^{\leq t} = \mathbf{1}^{\leq t}] - E[Y_t \mid \mathbf{Z}^{\leq t} = \mathbf{0}^{\leq t}]}
        {\mathbb{E}[Q_t \mid \mathbf{Z}^{\leq t} = \mathbf{1}^{\leq t}] - E[Q_t \mid \mathbf{Z}^{\leq t} = \mathbf{0}^{\leq t}]}
    \\
    \delta_t^{\text{YI}} &=
        \frac{\delta_t}{P(I_t = \text{YI})} = 
        \frac{
            \mathbb{E}[Y_t \mid \mathbf{Z}^{\leq t} = (\mathbf{0}^{< t}, 1)]
                - \mathbb{E}[Y_t \mid \mathbf{Z}^{\leq t} = \mathbf{0}^{\leq t}]
        }
        {   \mathbb{E}[Q_t \mid \mathbf{Z}^{\leq t} = (\mathbf{0}^{< t}, 1)]
                - \mathbb{E}[Q_t \mid \mathbf{Z}^{\leq t} = \mathbf{0}^{\leq t}]
        }
\end{align*}
\end{proposition}
\begin{proof}
    See Appendix \ref{app:proofs:conditional:identifiability}.
\end{proof}

\begin{example}\label{app:ex:ab-af-yielding-effects}
    Assume that $Q_{it}(\mathbf{Z})$ are sampled such that the positivity (\ref{assump:positivity}) and monotonicity property (\ref{assump:monotonicity}) are satisfied. Then, under both the Automation Bias and Alert Fatigue data-generating processes,
    \begin{align*}
        \tau_t^{YI} &= \delta_t^{YI} = 2\mu_A - 1
    \end{align*} 
    which reflects how much better the algorithm is compared to random guessing.
    On the other hand, under the Calibrated Reliance data-generating process, 
    \begin{align*}
        \tau_t^{YI} &= 
        \frac{(2\mu_A - 1) \beta_Z + \beta_S (1 - \rho^{t-1})}
        {\beta_Z + \beta_S(2\mu_A - 1) (1 - \rho^{t-1})} \\
        \delta_t^{YI} &= 2\mu_A - 1
    \end{align*}
    Notice that, for Calibrated Reliance, $\tau_t^{YI}$ varies with $\beta_Z, \beta_S, \rho$, and $t$. This is expected, since the data-generating process is constructed to produce selective concordance with the algorithm that develops over time and that depends on the calibration dynamics.
\end{example}
\begin{proof}
    See Appendix \ref{app:proofs:conditional}.
\end{proof}

\subsection{Using conditional effects to ``explain" the habituation effect}
Using the decomposition $\tau_t = \delta_t + \Delta_t$ and the conditional effect definitions, we can write the habituation effect in terms of the conditional effects and yielding probabilities as,
\begin{align}\label{eq:htt-decomp}
    \Delta_t = \tau_t^{\text{YI}} \cdot P(H_t = \text{YI}) - \delta_t^{\text{YI}} \cdot P(I_t = \text{YI})
\end{align}
From the above, we see that $\Delta_t$ can be non-zero across several different scenarios. First, most obviously it can be non-zero when $\tau_t^{YI} \neq \delta_t^{YI}$ and $P(H_t = YI) \neq P(I_t = YI)$. However, it can also be non-zero even when the conditional effects $\tau_t^{\text{YI}}$ and $\delta_t^{\text{YI}}$ are equal, provided the yielding probabilities $P(H_t = \text{YI})$ and $P(I_t = \text{YI})$ are different. Similarly,  $\Delta_t$ can be non-zero if $P(H_t = \text{YI}) = P(I_t = \text{YI})$, provided $\tau_t^{\text{YI}} \neq \delta_t^{\text{YI}}$. Informally speaking, if only $P(H_t = \text{YI})$ and $P(I_t = \text{YI})$ are different but $\tau_t^{\text{YI}}$ and $\delta_t^{\text{YI}}$ are the same, then habituation simply changes how often decision-makers yield to the tool. On the other hand, if only $\tau_t^{\text{YI}}$ and $\delta_t^{\text{YI}}$ are different while $P(H_t = \text{YI})$ and $P(I_t = \text{YI})$ are the same, then habituation does not change how often decision-makers yield to the tool; instead, it shifts the ``composition" of decisions where yielding occurs such that the effect of the intervention for decisions where there is immediate yielding is different from the effect when there is habituated yielding. This could happen, for example, if, with repeated exposure, decision-makers are more likely to yield to the tool when the tool is more likely to be correct and less likely to yield when the tool is wrong. 

\begin{example}
    As shown in Example \ref{app:ex:ab-af-yielding-effects} above, $\tau_t^{YI} = \delta_t^{YI}$ for both the Automation Bias and Alert Fatigue data-generating processes. For the Automation Bias data-generating process, 
    \begin{align*}
        P(H_t = YI) &= \beta_Z + \beta_S(1 - \rho^{t-1}) \\
        P(I_t = YI) &= \beta_Z \\
        \Rightarrow P(H_t = YI) &\neq P(I_t = YI) \text{ for all } t > 1, \rho \in [0, 1)
    \end{align*}
    And for the Alert Fatigue data-generating process,
    \begin{align*}
        P(H_t = YI) &= \beta_Z + \beta_S(1 - \mu_A)(1 - \rho^{t-1}) \\
        P(I_t = YI) &= \beta_Z \\
        \Rightarrow P(H_t = YI) &\neq P(I_t = YI) \text{ for all } t > 1, \rho \in [0, 1), \mu_A \neq 1
    \end{align*}
    Thus, both of these data-generating processes are examples of the scenario in which there is a non-zero habituation effect driven by different habituated and immediate yielding probabilities, as opposed to different conditional yielding effects.

    On the other hand, for the Calibrated Reliance data-generating process $\tau_t^{YI} \neq \delta_t^{YI}$. Additionally,
    \begin{align*}
        P(H_t = YI) &= \beta_Z + \beta_S(2\mu_A - 1)(1 - \rho^{t-1}) \\
        P(I_t = YI) &= \beta_Z \\
        \Rightarrow P(H_t = YI) &\neq P(I_t = YI) \text{ for all } t > 1, \rho \in [0, 1), \mu_A \neq \frac{1}{2}
    \end{align*}
    Thus, this data-generating process is an example of the scenario in which the habituation effect is driven by both a difference in habituated and immediate yielding probabilities, as well as different conditional yielding effects.

\end{example}

\section{Profiling ``yielding" decisions}\label{app:sec:profiling}
To better understand yielding decisions, e.g. \textit{when} and \textit{why} yielding may be taking place, we can ``profile" yielding decisions by investigating certain defining characteristics of these types of decisions. Let $X_t \in \mathbb{R}$ be a covariate of interest, which could be a covariate that describes a decision-maker, a decision-recipient, or the decision support tool itself. In addition to the assumption of ignorability for the decision $D_{it}(\mathbf{Z}_i^{\leq t})$,  also assume ignorability for $X_{it}$.

\begin{assumption}[Ignorability for $X$]\label{assump:ignorability-x}
    \begin{align*}
        (X_{it}, D_{it}(\mathbf{z}^{\leq t}), Y(D_{it}(\mathbf{z}^{\leq t}))) \indep \mathbf{Z}_i^{\leq t} \text{ for all } \mathbf{z}^{\leq t}
    \end{align*}
\end{assumption}

Then, together with the monotonicity assumption \ref{assump:monotonicity}, the expected value of $X$ given yielding (either habituated or immediate) is identifiable from observed data. 

\begin{proposition}[Covariate profiling of yielding decisions]\label{prop:profile}
Under Assumptions \ref{assump-partial-interf}- \ref{assump-mediation}, \ref{assump:positivity}, \ref{assump:monotonicity},  and \ref{assump:ignorability-x}, 
    \begin{align*}
        \mathbb{E}[X_t \mid H_t = \text{YI}] &= 
            \frac{1}
            {
                \mathbb{E}[Q_t \mid \mathbf{Z}^{\leq t} = \mathbf{1}^{\leq t}] - 
                \mathbb{E}[Q_t \mid \mathbf{Z}^{\leq t} = \mathbf{0}^{\leq t}]} 
             \\
            &\quad \cdot( \mathbb{E}[X_t] \\
                & \quad\quad - \mathbb{E}[X_t \mid Q_t = 0, \mathbf{Z}^{\leq t} = \mathbf{1}^{\leq t}] \cdot P(Q_t = 0 \mid \mathbf{Z}^{\leq t} = \mathbf{1}^{\leq t}) \\
                & \quad\quad - \mathbb{E}[X_t \mid Q_t = 1, \mathbf{Z}^{\leq t} = \mathbf{0}^{\leq t}] \cdot P(Q_t = 1 \mid \mathbf{Z}^{\leq t} = \mathbf{0}^{\leq t})
            ) \\
    \mathbb{E}[X_t \mid I_t = \text{YI}] &= 
        \frac{1}
        {
            \mathbb{E}[Q_t \mid \mathbf{Z}^{\leq t} = (\mathbf{0}^{< t}, 1)] - 
            \mathbb{E}[Q_t \mid \mathbf{Z}^{\leq t} = \mathbf{0}^{\leq t}]} 
         \\
        &\quad \cdot( \mathbb{E}[X_t] \\
            & \quad\quad - \mathbb{E}[X_t \mid Q_t = 0, \mathbf{Z}^{\leq t} = (\mathbf{0}^{< t}, 1)] \cdot P(Q_t = 0 \mid \mathbf{Z}^{\leq t} = (\mathbf{0}^{< t}, 1)) \\
            & \quad\quad - \mathbb{E}[X_t \mid Q_t = 1, \mathbf{Z}^{\leq t} = \mathbf{0}^{\leq t}] \cdot P(Q_t = 1 \mid \mathbf{Z}^{\leq t} = \mathbf{0}^{\leq t})
        )
\end{align*}
\end{proposition}
\begin{proof}
    See Appendix \ref{app:proofs:conditional}.
\end{proof}

\begin{example}\label{app:ex:ab-af-cr-profile}
    Let $X_{it} \in \{-1, 1\}$ be an observed covariate for each decision such that $P(X_{it} = 1) = \gamma, \ P(X_{it} = -1) = 1-\gamma$. Suppose also that, $A_{it}$, the correctness of the algorithm's recommendation, is sampled as 
    \begin{align*}
        A_{it} &\sim \text{Bernoulli}(p_{A_{it}}) \\
        p_{A_{it}} &= \text{sigmoid}(\alpha + \beta_X X)
    \end{align*}
    where $\alpha, \beta_X \in \mathbb{R}^+$. Thus, the algorithm is more likely to be correct when $X = 1$ and less likely to be correct when $X = -1$. 

    Under the Automation Bias and Alert Fatigue data-generating processes and, using the above sampling process for $A_{it}$, 
    \begin{align*}
        \mathbb{E}[X_t \mid H_t = YI] = \mathbb{E}[X_t \mid I_t = YI] = 2\gamma - 1
    \end{align*}
    That is, the expected value of the covariate on habituated and immediate yielding decisions is the same as the unconditional expected value of the covariate.

    On the other hand, under the Calibrated Reliance data-generating process, the expected value of the covariate on immediate yielding decisions is the same,
    \begin{align*}
        \mathbb{E}[X_t \mid I_t = YI] = 2 \gamma - 1,
    \end{align*}
    but the profile for habituated yielding decisions is
    \begin{align*}
        \mathbb{E}[X_t \mid H_t = YI] = 
        (2 \gamma - 1)
         + \frac{4(1 - \gamma)(\mu_A - \mu_A \mid_{X = -1})\beta_S (1 - \rho^{t-1})}{\beta_Z + \beta_S(2\mu_A - 1)(1 - \rho^{t-1})}
    \end{align*}
    where 
    \begin{align*}
        \mu_A &= \mathbb{E}[A_t] = \gamma \ \text{sigmoid}(\alpha + \beta_X) + (1 - \gamma)\text{sigmoid}(\alpha - \beta_X) \\
        \mu_A \mid_{X=-1} &= \mathbb{E}[A_t \mid X_t = -1] = \text{sigmoid}(\alpha - \beta_X)
    \end{align*}
    Notably, this profile varies with $\mu_A$, $\beta_S$, $\beta_Z$, $\rho$, and time $t$.
\end{example}
\begin{proof}
    See Appendix \ref{app:proofs:conditional}.
\end{proof}

\section{Proofs}\label{app:proofs}
\subsection{Minimax Design}\label{app:proofs:minimax}
The minimax design result follows from the results proven previously in \cite{Basse2023Minimax}. We make small extensions to accommodate the staggered offsets part of the design and the introduction of additional estimators ($\tau_t$ and $\lambda_t$). 

First, we restate the permutation invariance assumption from \cite{Basse2023Minimax} that establishes the support of the potential outcomes. For $\mathcal{Y} \subset \mathbb{R}^N$, let $\mathbb{Y}(\mathcal{Y}) = \{[\mathbf{y}_1 \dots \mathbf{y}_T]: \mathbf{y}_t \in \mathcal{Y}, t = 1, \dots, T\}$ denote the set of $N \times T$ matrices whose columns are in $\mathcal{Y}$, and let $\underline{\mathbb{Y}}(\mathcal{Y})$ denote the set of all potential outcomes schedules whose matrices are all elements of $\mathbb{Y}(\mathcal{Y})$. That is, for any $\underline{\mathbf{Y}} \in \underline{\mathbb{Y}}(\mathcal{Y})$, it holds that $\mathbf{Y}(D(\mathbf{z})) \in \underline{\mathbf{Y}}$ for all $\mathbf{Y}(D(\mathbf{z})) \in \underline{\mathbf{Y}}$.

\begin{assumption}[Permutation invariance]\label{assump:perm-invariance}
      If there exists some $\mathbf{y}_t(\mathbf{z}) \in \underline{\mathbf{Y}} \in \underline{\mathbb{Y}}(\mathcal{Y})$ such that $\mathbf{y}_t(\mathbf{z}) = \mathbf{y}$, then there should also exist $\mathbf{y}_t(\mathbf{z}) \in \underline{\mathbf{Y}}' \in \underline{\mathbb{Y}}(\mathcal{Y})$ such that $\mathbf{y}_t(\mathbf{z}) = \sigma \cdot \mathbf{y}$, for any permutation $\sigma$ in $S_N$ and any $t =1, \dots, T$.
\end{assumption}
    
For the sake of completeness, we restate the results from Lemmas 1-4 in \cite{Basse2023Minimax}, with adaptations for our setting. We use slightly modified notation to avoid overlapping symbols with other parts of this work.
 
\begin{lemma}\label{lemma:loss-perm-invar}
    Let $\sigma \in S_{N}$ denote a permutation over the symmetric group on $N$ elements. For any assignment matrix $\mathbf{Z}$, and any schedule $\underline{\mathbf{Y}}$, 
    \begin{align*}
        L(\sigma \cdot \mathbf{Z}, \sigma \cdot \underline{\mathbf{Y}}) = L(\mathbf{Z}, \underline{\mathbf{Y}})
    \end{align*}
\end{lemma}
\begin{proof}
    We can apply the same proof argument as in \cite{Basse2023Minimax}'s proof for Lemma 1, and simply replace their estimator definitions with our estimators $\hat{\tau}, \hat{\delta}, \hat{\Delta}$, and $\hat{\lambda}$.
\end{proof}

\begin{lemma}\label{lemma:tilde-pi-minimax}
    For design $\pi \in \mathbb{H}$ and $\sigma \in S_N$, let $\tilde{\pi} = (N!)^{-1} \sum_{\sigma \in S_N} \pi(\sigma \cdot \mathbf{Z})$. Then, for a permutation-invariant schedule $\mathbb{Y}(\mathcal{Y})$, 
    \begin{align*}
        \max_{\underline{\mathbf{Y}} \in \mathbb{Y}(\mathcal{Y})} \{r(\tilde{\pi}; \underline{\mathbf{Y}})\} \leq \max_{\underline{\mathbf{Y}} \in \mathbb{Y}(\mathcal{Y})} \{r(\pi; \underline{\mathbf{Y}})\}
    \end{align*}
\end{lemma}
\begin{proof}
    See \cite{Basse2023Minimax}'s proof for Lemma 2, which invokes Lemma \ref{lemma:loss-perm-invar}.
\end{proof}

\begin{lemma}\label{lemma:tilde-pi-representation}
Let $\phi_{\mathbf{Z}}$ denote the design that assigns mass 1 at $\mathbf{Z}$.  Then,
    \begin{align*}
        \tilde{\pi} = \sum_{\mathbf{Z} \in \mathcal{Z}} \pi(\mathbf{Z})\tilde{\phi}_{\mathbf{Z}}
    ,\quad \text{where } \tilde{\phi}_{\mathbf{Z}}(\mathbf{Z}')
        = ({N!})^{-1} \sum_{\sigma \in S_N} \phi_{\mathbf{Z}}(\sigma \cdot \mathbf{Z}')
    \end{align*}
\end{lemma}
\begin{proof}
    See \cite{Basse2023Minimax}'s proof for Lemma 3.
\end{proof}

\begin{lemma}\label{lemma:variance}
    \begin{align*}
        \max_{\underline{\mathbf{Y}} \in \mathbb{Y}(\mathcal{Y})} \{r(\tilde{\pi}; \underline{\mathbf{Y}})\}
        = V^*
        \sum_{Z \in \mathcal{Z}}
        \pi(\mathbf{Z})
        \sum_{t=2}^T
        \left[
            \frac{2}{N_{1, t}(\mathbf{Z})}
            + \frac{3}{N_{0, t}(\mathbf{Z})}
            + \frac{2}{N_{w_{t, \infty}}(\mathbf{Z})}
            + \frac{1}{N_{w_{1, t}}(\mathbf{Z})}
        \right]
    \end{align*}
\end{lemma}
\begin{proof}
    We largely reiterate \cite{Basse2023Minimax}'s Lemma 4 proof, with modifications for our additional estimators and double wedge design. First, \cite{Basse2023Minimax} show that, by applying Lemma \ref{lemma:tilde-pi-representation}, we have
    \begin{align*}
        r(\tilde{\pi}; \underline{\mathbf{Y}})
        = \sum_{\mathbf{Z} \in \mathcal{Z}} \pi(\mathbf{Z})r(\tilde{\phi}_{\mathbf{Z}}, \underline{\mathbf{Y}})
    \end{align*}
    where $\tilde{\phi}(\mathbf{Z})$ is a completely randomized design that assigns $N_{1}(\mathbf{Z})$ decision-makers to $\mathbf{1}$, $N_{0}(\mathbf{Z})$ to $\mathbf{0}$, $N_{w_{t, \infty}}(\mathbf{Z})$ to $\mathbf{w}_{t, \infty}$ for $t=2, \dots, T$, and $N_{w_{1, t}}(\mathbf{Z})$ for $t=2, \dots, T$.
    
    We can then write the risk under $\tilde{\phi}(\mathbf{Z})$ as,
    \begin{align*}
        r(\tilde{\phi}_{\mathbf{Z}}; \underline{\mathbf{Y}}) &= 
            \sum_{t=2}^T 
                \left( \mathbb{E}[\hat{\tau}^{sdw}_t] - \tau_t \right)^2 
                + \left( \mathbb{E}[\hat{\Delta}^{sdw}_t] - \Delta_t \right)^2
                + \left( \mathbb{E}[\hat{\delta}^{sdw}_t] - \delta_t \right)^2
                + \left( \mathbb{E}[\hat{\lambda}^{sdw}_t] - \lambda_t \right)^2 \\
            &\quad + \sum_{t=2}^T 
                \left(V_{\tilde{\phi}_{\mathbf{Z}}}(\hat{\tau}^{sdw}_t) + V_{\tilde{\phi}_{\mathbf{Z}}}(\hat{\Delta}^{sdw}_t) + V_{\tilde{\phi}_{\mathbf{Z}}}(\hat{\delta}^{sdw}_t) + V_{\tilde{\phi}_{\mathbf{Z}}}(\hat{\lambda}^{sdw}_t) \right)\\
            &= \sum_{t=2}^T 
            V_{\tilde{\phi}_{\mathbf{Z}}}(\hat{\tau}^{sdw}_t) +
            V_{\tilde{\phi}_{\mathbf{Z}}}(\hat{\Delta}^{sdw}_t) + V_{\tilde{\phi}_{\mathbf{Z}}}(\hat{\delta}^{sdw}_t) + V_{\tilde{\phi}_{\mathbf{Z}}}(\hat{\lambda}^{sdw}_t) 
    \end{align*}
    where the second equality follows since the estimators are unbiased under Assumptions \ref{assump-partial-interf} - \ref{assump:ignorability}, and where
    \begin{align*}
        V_{\tilde{\phi}_{\mathbf{Z}}}({\hat{\tau}^{sdw}_t}) &= 
            \frac{V_1^{(t)}}{N_{1, t}(\mathbf{Z})} 
                + \frac{V_0^{(t)}}{N_{0, t}(\mathbf{Z})} 
                - \frac{V_{1, 0}^{(t)}}{N}\\
        V_{\tilde{\phi}_{\mathbf{Z}}}({\hat{\Delta}^{sdw}_t}) &= 
            \frac{V_1^{(t)}}{N_{1, t}(\mathbf{Z})} 
                + \frac{V_{w_{t, \infty}}^{(t)}}{N_{w_{t, \infty}}(\mathbf{Z})} 
                - \frac{V_{1, w_{t, \infty}}^{(t)}}{N}\\
        V_{\tilde{\phi}_{\mathbf{Z}}}({\hat{\delta}^{sdw}_t}) &= \frac{V_{w_{t, \infty}}^{(t)}}{N_{w_{t, \infty}}(\mathbf{Z})} 
                + \frac{V_0^{(t)}}{N_{0, t}(\mathbf{Z})} 
                - \frac{V_{w_{t, \infty}, 0}^{(t)}}{N} \\
        V_{\tilde{\phi}_{\mathbf{Z}}}({\hat{\lambda}^{sdw}_t}) &= \frac{V_{w_{1, t}}^{(t)}}{N_{w_{1, t}}(\mathbf{Z})} 
                + \frac{V_0^{(t)}}{N_{0, t}(\mathbf{Z})} 
                - \frac{V_{w_{1, t}, 0}^{(t)}}{N} \\
        V_h^{(t)} &= \frac{1}{N-1} \sum_{i=1}^N \left(Y_{it}(D_{it}(h)) - \bar{Y}_{it}(D_{it}(h)) \right)^2
            \quad\quad h=\mathbf{0}, \mathbf{1}, \mathbf{w}_{t, \infty}, \mathbf{w}_{1, t} \\
        V_{1, 0}^{(t)} &= 
            \frac{1}{N-1} \sum_{i=1}^N 
                \left(
                    [Y_{it}(D_{it}(\mathbf{1}^{\leq t})) - Y_{it}(D_{it}(\mathbf{0}^{\leq t}))]
                    - [\bar{Y}_{it}(D_{it}(\mathbf{1}^{\leq t})) - \bar{Y}_{it}(D_{it}(\mathbf{0}))]
                \right)^2\\
        V_{1, w_{t, \infty}}^{(t)} &= 
            \frac{1}{N-1} \sum_{i=1}^N 
                \left(
                    [Y_{it}(D_{it}(\mathbf{1}^{\leq t})) - Y_{it}(D_{it}(\mathbf{w}_{t, \infty}))]
                    - [\bar{Y}_{it}(D_{it}(\mathbf{1}^{\leq t})) - \bar{Y}_{it}(D_{it}(\mathbf{w}_{t, \infty})))]
                \right)^2\\
        V_{w_{t, \infty}, 0}^{(t)} &=
            \frac{1}{N-1} \sum_{i=1}^N 
                \left(
                    [Y_{it}(D_{it}(\mathbf{w}_{t, \infty})) - Y_{it}(D_{it}(\mathbf{0}^{\leq t}))]
                    - [\bar{Y}_{it}(D_{it}(\mathbf{w}_{t, \infty}))) - \bar{Y}_{it}(D_{it}(\mathbf{0}^{\leq t}))]
                \right)^2\\
        V_{w_{1, t}, 0}^{(t)} &= 
            \frac{1}{N-1} \sum_{i=1}^N 
                \left(
                    [Y_{it}(D_{it}(\mathbf{w}_{1, t})) - Y_{it}(D_{it}(\mathbf{0}^{\leq t}))]
                    - [\bar{Y}_{it}(D_{it}(\mathbf{w}_{1, t}))) - \bar{Y}_{it}(D_{it}(\mathbf{0}^{\leq t}))]
                \right)^2\\
    \end{align*}
Then the risk over $\tilde{\pi}$ can be written as
\begin{align*}
    r(\tilde{\pi}; \underline{\mathbf{Y}})
    &= \sum_{t=2}^T M_{1}^{(t)}(\tilde{\pi}; \underline{\mathbf{Y}}) + M_{2}^{(t)}(\underline{\mathbf{Y}})  \\
    \text{ where }
    M_{1}^{(t)}(\tilde{\pi}, \underline{\mathbf{Y}}) &=  
        2V_1^{(t)} \sum_{\mathbf{Z} \in \mathcal{Z}}\frac{\pi(\mathbf{Z})}{N_{1, t}(\mathbf{Z})}
        + 3V_0^{(t)} \sum_{\mathbf{Z} \in \mathcal{Z}}\frac{\pi(\mathbf{Z})}{N_{0, t}(\mathbf{Z})} \\
        &\quad\quad\quad+ 2V_{w_{t, \infty}}^{(t)} \sum_{\mathbf{Z} \in \mathcal{Z}}\frac{\pi(\mathbf{Z})}{N_{w_{t, \infty}}(\mathbf{Z})} 
        + V_{w_{1, t}}^{(t)} \sum_{\mathbf{Z} \in \mathcal{Z}}\frac{\pi(\mathbf{Z})}{N_{w_{1, t}}(\mathbf{Z})}  \\
    M_{2}^{(t)}(\underline{\mathbf{Y}}) &= - \frac{1}{N} \left(V_{1, 0}^{(t)} + V_{1, w_{t, \infty}}^{(t)} + V_{w_{t, \infty}, 0}^{(t)} + V_{w_{1, t}, 0}^{(t)} \right)
\end{align*}
Since $V_h^{(t)}$ depends only on $\mathbf{y}_t(h) = \left(Y_{1t}(D(h)), Y_{2t}(D(h)), \dots, Y_{Nt}(D(h)) \right)$, we can express $V_h^{(t)}$ as $V_h^{(t)} = V(\mathbf{y}_t(h))$, for some function $V(\cdot)$. Since $y_t(h) \in \mathcal{Y}$, we have that 
\begin{align*}
    \argmax_{\mathbf{y}_t(h) \in \mathcal{Y}} V(\mathbf{y}_t(h)) = \mathcal{Y}^{\text{opt}}, 
\end{align*}
where $\mathcal{Y}^{\text{opt}}$ is constant and may be set-valued. Furthermore, it holds that, for every $\mathbf{y} \in \mathcal{Y}^{\text{opt}}$, 
\begin{align*}
    \max_{\underline{\mathbf{Y}} \in \mathbb{Y}(\mathcal{Y})} V_h^{(t)} = V(\mathbf{y})
\end{align*}
For any choice of $\mathbf{y} \in \mathcal{Y}^{\text{opt}}$, we can construct the $N \times T$ matrix, $\mathbf{Y}^{\text{opt}}$ with all columns set to $\mathbf{y}$, and we define $\underline{Y}^{\text{opt}} = [\mathbf{Y}^{\text{opt}}, \dots, \mathbf{Y}^{\text{opt}}]$ as a schedule with $2T$ copies of $\mathbf{Y}^{\text{opt}}$ (in contrast to the $T+1$ copies in \cite{Basse2023Minimax}). Then, for $h = \mathbf{0}, \mathbf{1}, \mathbf{w}_{t, \infty}, \mathbf{w}_{1, t}$, it holds that
\begin{align*}
    \max_{\underline{\mathbf{Y}} \in \mathbb{Y}(\mathcal{Y})} V_h^{(t)}
    &= V_h^{(t)}(\underline{Y}^{\text{opt}}) = \frac{1}{N-1} \sum_{i=1}^N (\mathbf{y}_i - \bar{\mathbf{y}})^2,
\end{align*}
which is equal to some constant $V^*$, assuming the potential outcomes are bounded. It then holds that
\begin{align*}
    \mathbf{\underline{Y}}^{\text{opt}} \in 
        \argmax_{\underline{\mathbf{Y}} \in \mathbb{Y}(\mathcal{Y})} \sum_{t=2}^T M_{1}^{(t)}(\tilde{\pi}, \underline{\mathbf{Y}})
\end{align*}
since we can separately optimize each of the terms $V_{\mathbf{0}}^{(t)}, V_{\mathbf{1}}^{(t)}, V_{\mathbf{w}_{t, \infty}}^{(t)}, V_{\mathbf{w}_{1, t}}^{(t)}$ in $M_{1}^{(t)}(\tilde{\pi}, \underline{\mathbf{Y}})$.

Finally, by construction of $\underline{\mathbf{Y}}^{\text{opt}}$, we have that
\begin{align*}
    V_{1, 0}^{(t)}(\underline{\mathbf{Y}}^{\text{opt}}) =
    V_{1, \mathbf{w}_{t, \infty}}^{(t)}(\underline{\mathbf{Y}}^{\text{opt}}) =
    V_{\mathbf{w}_{t, \infty}, 0}^{(t)}(\underline{\mathbf{Y}}^{\text{opt}}) =
    V_{\mathbf{w}_{1, t}, 0}^{(t)}(\underline{\mathbf{Y}}^{\text{opt}}) = 0
\end{align*}
And since all variances must be non-negative, it holds that
\begin{align*}
    &\max_{\underline{\mathbf{Y}} \in \mathbb{Y}(\mathcal{Y})} \{- V_{l}^{(t)} \} = 0 \\
    &\text{for} \ V_l^{(t)} \in \{V_{1, 0}^{(t)},  V_{1, \mathbf{w}_{t, \infty}}^{(t)}, V_{\mathbf{w}_{t, \infty}, 0}^{(t)} , V_{\mathbf{w}_{1, t}, 0}^{(t)}\}
\end{align*} 
Thus, $\underline{\mathbf{Y}}^{\text{opt}}$ must also satisfy
\begin{align*}
    \underline{\mathbf{Y}}^{\text{opt}} \in 
        \argmax_{\underline{\mathbf{Y}} \in \mathbb{Y}(\mathcal{Y})} \sum_{t=2}^T M_{1}^{(t)}(\tilde{\pi}, \underline{\mathbf{Y}}) + M_2^{(t)}(\underline{\mathbf{Y}})
\end{align*}
Plugging in the constant $V^*$, and using $\max_{\mathbf{\underline{Y}} \in \mathbb{Y}(\mathcal{Y})} M_2^{(t)}(\mathbf{\underline{Y}}) = 0$ for all $t = 2, \dots, T$,
\begin{align*}
    \max_{\underline{\mathbf{Y}} \in \mathbb{Y}(\mathcal{Y})} 
    r(\tilde{\pi}; \underline{\mathbf{Y}})
    =V^* \sum_{Z \in \mathcal{Z}} \pi(\mathbf{Z}) 
     \sum_{t=2}^{T}\left[\frac{2}{N_{1, t}(\mathbf{Z})} + \frac{3}{N_{0, t}(\mathbf{Z})} + \frac{2}{N_{w_{t, \infty}}(\mathbf{Z})} + \frac{1}{N_{w_{1, t}}(\mathbf{Z})} \right]
\end{align*}
\end{proof}

\begin{proof}[Proof of Theorem \ref{thm:minimax}]
By Lemma \ref{lemma:tilde-pi-minimax}, if $\pi$ is minimax optimal, then $\tilde{\pi}$ is also minimax optimal, so the goal is to find $\tilde{\pi}$ that achieves minimax risk. Using Lemma \ref{lemma:variance}, it follows that a design $\pi^{\text{opt}}$ that solves $\min_{\pi} \max_{\underline{\mathbf{Y}} \in \mathbb{Y}(\mathcal{Y})} r(\tilde{\pi}; \underline{\mathbf{Y}})$ must satisfy
\begin{align*}
   &\text{for any } \mathbf{Z} \text{ such that } \pi^{\text{opt}}(\mathbf{Z}) > 0, \\
    &\left( 
        N_1, N_0, \{N_{w_{t, \infty}}\}_{t=2}^T, \{N_{w_{1, t}}\}_{t=2}^T
    \right) \\
     &= \underset{N_{1}, N_{0}, \{N_{w_{t, \infty}}\}, \{N_{w_{1, t}}\}}{\arg \min}
        \sum_{t=2}^T 
        \left[
            \frac{2}{N_{1, t}} 
            + \frac{3}{N_{0, t}} 
            + \frac{2}{N_{w_{t, \infty}}}
            + \frac{1}{N_{w_{1, t}}}
        \right]
\end{align*}
Note that the objective is over $N_1, N_0, \{N_{w_{t, \infty}}\}, \{N_{w_{1, t}}\}$, since $N_{1, t}$ and $N_{0, t}$ are induced quantities that are fully determined by the values of $N_1, N_0, \{N_{w_{t, \infty}}\}, \{N_{w_{1, t}}\}$. In the \cite{Basse2023Minimax} Theorem 1 proof, they show that, for any $\mathbf{Z}$ such that $\pi^{\text{opt}}(\mathbf{Z}) > 0$, $\tilde{\pi}^{\text{opt}}(\sigma \cdot \mathbf{Z}) = \tilde{\pi}^{\text{opt}}(\mathbf{Z})$ for any permutation $\sigma \in S_N$. This implies that $\tilde{\pi}$ is a completely randomized design over the permutations of $\mathbf{Z}$ where $\pi^{\text{opt}}(\mathbf{Z}) > 0$, which are exactly the assignments that satisfy the $\argmin$ objective above.

As in \cite{Basse2023Minimax}, we solve an integer relaxation of the objective using Lagrangian multipliers. We write the Lagrangian function as

\begin{align*}
    &\mathcal{L}\left( 
            N_1, N_0, \{N_{w_{t, \infty}}\}_{t=2}^T, \{N_{w_{1, t}}\}_{t=2}^T, v
        \right) \\
        &= \sum_{t=2}^T 
            \left[
                \frac{2}{N_{1, t}} 
                + \frac{3}{N_{0, t}} 
                + \frac{2}{N_{w_{t, \infty}}}
                + \frac{1}{N_{w_{1, t}}}
            \right]   
            + v \left[ N_0 + N_1 + \sum_{t=2}^T N_{w_{t, \infty}} + \sum_{t=2}^T N_{w_{1, t}} - N \right]
\end{align*}
Calculating the partial derivatives and setting to zero, we have
\begin{align}\label{eq-lagrange-n1}
    \frac{\partial \mathcal{L}}{\partial N_1} 
        = -\sum_{t=2}^{T}\frac{2}{N_{1, t}^2} + v
        = 0
\end{align}
\begin{align}\label{eq-lagrange-n0}
    \frac{\partial \mathcal{L}}{\partial N_0} 
        = -\sum_{t=2}^{T}\frac{3}{N_{0, t}^2} + v
        = 0
\end{align}
\begin{align}\label{eq-lagrange-n-onset}
    \frac{\partial \mathcal{L}}{\partial N_{w_{s, \infty}}} 
        = -\sum_{t=2}^{s-1}\frac{3}{N_{0, t}^2} - \frac{2}{N_{w_{s, \infty}}^2} + v
        = 0
\end{align}
\begin{align}\label{eq-lagrange-n-offset}
    \frac{\partial \mathcal{L}}{\partial N_{w_{1, s}}} 
        = -\sum_{t=2}^{s-1}\frac{2}{N_{1, t}^2} - \frac{1}{N_{w_{1, s}}^2} + v
        = 0
\end{align}

Using Equations \ref{eq-lagrange-n0} and \ref{eq-lagrange-n-onset} above, and setting $s=T$, we have
\begin{align*}
    \sum_{t=2}^{T}\frac{3}{N_{0, t}^2} &= \sum_{t=2}^{T-1}\frac{3}{N_{0, t}^2} + \frac{2}{N_{w_{T, \infty}}^2} \\
    \frac{3}{N_{0, T}^2} &= \frac{2}{N_{w_{T, \infty}}^2}\\
    \frac{3}{N_0^2} &= \frac{2}{N_{w_{T, \infty}}^2} & \text{since } N_{0, t} = N_0 \text{ at } t = T \\
    N_{w_{T, \infty}} &= \sqrt{\frac{2}{3}} N_0
\end{align*}
Let $c_t$ be a scalar multiplier such that $c_T = 1$, so we can write $N_{w_{t, \infty}} = \sqrt{\frac{2}{3}} N_0 c_t$. We can then rewrite $N_{0, t}$ in terms of $N_0$ and $c_t$ as 
\begin{align*}
    N_{0, t} 
        &= N_0 + \sum_{l=t+1}^T N_{w_{l, \infty}} \\
        &= N_0 + \sum_{l=t+1}^T \sqrt{\frac{2}{3}}N_0 c_l \\
        &= N_0 \left( 1 + \sqrt{\frac{2}{3}}\sum_{l=t+1}^Tc_l \right)
\end{align*}
Now, subtracting \ref{eq-lagrange-n-onset} at $s$ from \ref{eq-lagrange-n-onset} at $s+1$, we establish a recursive definition for $c_t$,
\begin{align*}
     \left[\sum_{t=2}^{s}\frac{3}{N_{0, t}^2} + \frac{2}{N_{w_{s+1, \infty}}^2}\right]
      &= \left[\sum_{t=2}^{s-1}\frac{3}{N_{0, t}^2} + \frac{2}{N_{w_{s, \infty}}^2}\right]  \\
     \frac{2}{N_{w_{s, \infty}}^2} 
        &= \frac{3}{N_{0, s}^2} + \frac{2}{N_{w_{s+1, \infty}}^2} \\
     \frac{2}{\frac{2}{3} N_0^2 c_s^2}
        &= \frac{3}{N_0^2\left(1 + \sqrt{\frac{2}{3}} \sum_{l=s+1}^{T}c_l\right)^2} + \frac{2}{\frac{2}{3}N_0^2c_{s+1}^2} \\
    c_s &= \left[ 
        \frac{1}{c_{s+1}^2} +  
        \frac{1}{(1 + \sqrt{\frac{2}{3}}\sum_{l=s+1}^T c_l)^2}
        \right]^{-1/2} \quad \text{ for all } s = 2, \dots, T-1
\end{align*}

Using the same logic with Equations \ref{eq-lagrange-n1} and \ref{eq-lagrange-n-offset}, we have
\begin{align*}
    N_{w_{1, t}} &= \sqrt{\frac{1}{2}} N_1 d_t \\
    d_s &= \left[ 
        \frac{1}{d_{s+1}^2} +  
        \frac{1}{(1 + \sqrt{\frac{1}{2}}\sum_{l=s+1}^T d_l)^2}        \right]^{-1/2}
        \quad \text{ for all } s = 2, \dots, T-1 \\
    d_T &= 1
\end{align*}

Next, we write $N_1$ in terms of $N_0$. We apply a telescoping sum to the recursion rule for $c_s$,
\begin{align*}
    \frac{1}{c_s^2} &= \frac{1}{c_{s+1}^2} + \frac{1}{N_{0, s}^2 / N_0^2} \\
    \sum_{s=2}^{T-1} \frac{1}{c_s^2} - \frac{1}{c_{s+1}^2} &= N_0^2 \sum_{s=2}^{T-1}\frac{1}{N_{0, s}^2} \\
    \frac{1}{c_2^2} - \frac{1}{c_T^2} &= N_0^2 \sum_{s=2}^{T-1}\frac{1}{N_{0, s}^2} \\
    \frac{1}{c_2^2} & = N_0^2 \sum_{s=2}^{T-1} \frac{1}{N_{0,s}^2} + 1 \\
    \frac{1}{c_2^2} &= N_0^2 \sum_{s=2}^{T} \frac{1}{N_{0,s}^2}
\end{align*}
The same telescoping for $d_s$ gives
\begin{align*}
     \frac{1}{d_2^2} &= N_1^2 \sum_{s=2}^{T} \frac{1}{N_{1,s}^2}
\end{align*}
Then, equating \ref{eq-lagrange-n1} and \ref{eq-lagrange-n0} and substituting, we have
\begin{align*}
    \sum_{t=2}^{T}\frac{2}{N_{1, t}^2} &= \sum_{t=2}^{T}\frac{3}{N_{0, t}^2} \\
    \frac{2}{N_1^2 d_2^2} &= \frac{3}{N_0^2 c_2^2} \\
    N_1 &= \sqrt{\frac{2}{3}} \frac{c_2}{d_2} N_0
\end{align*}

Finally, using the constraint that there are $N$ total decision-makers and writing all $N_{(\cdot)}$ in terms of $N_0$ gives us
\begin{align*}
    N &= N_0 + N_1 + \sum_{t=2}^T N_{w_{t, \infty}} + \sum_{t=2}^T N_{w_{1, t}} \\
    N &= N_0 + \sqrt{\frac{2}{3}}\frac{c_2}{d_2}N_0 + \sqrt{\frac{2}{3}}\sum_{t=2}^TN_0c_t + \sqrt{\frac{1}{3}}\frac{c_2}{d_2}\sum_{t=2}^TN_0d_t \\
    N_0 &= N \left(1 + \sqrt{\frac{2}{3}}\frac{c_2}{d_2} 
    + \sqrt{\frac{2}{3}}\sum_{t=2}^Tc_t 
    + \sqrt{\frac{1}{3}}\frac{c_2}{d_2}\sum_{t=2}^Td_t
     \right)^{-1}
\end{align*}

The objective is convex for $N_0, N_1, \{N_{w_{t, \infty}}\}, \{N_{w_{1, t}}\} > 0$, so the above is a minimum.

\end{proof}

\subsection{Bias Under Monotonicity in Exposure}\label{app:proofs:bias-monotonic}
\begin{proof}[Proof of Proposition \ref{prop:bias-monotonic}]
    Suppose $Y_{it}(D_{it}(\mathbf{Z}^{\leq t})) < Y_{it}(D_{it}(\mathbf{Z}'^{\leq t}))$ if and only if $\sum_{t'=1}^{t}Z_{it'} < \sum_{t'=1}^{t}Z'_{it'}$; that is, $Y_{it}(D_{it}(\mathbf{Z}^{\leq t}))$ is strictly increasing in the cumulative number of exposures to assistance. Then, in the decision-randomized design such that $P(\mathbf{Z}_i = \mathbf{z}) > 0$ for all $\mathbf{z} \in \{0, 1\}^T$, it holds that
    \begin{align*}
        \mathbb{E_\mathbf{Z}}[Y_{it}(D_{it}(\mathbf{Z}^{< t}, 1)) \mid \mathcal{D}] 
            &< Y_{it}(D_{it}(\mathbf{1}^{\leq t})) \\
        \mathbb{E_\mathbf{Z}}[Y_{it}(D_{it}(\mathbf{Z}^{< t}, 0)) \mid \mathcal{D}] 
            &> Y_{it}(D_{it}(\mathbf{0}^{\leq t}))
    \end{align*}
    where the expectation is taken over the random treatment assignments in the decision-randomized design, conditioned on the sample, $\mathcal{D}$.

    Thus,
    \begin{align*}
        \mathbb{E}[\hat{\tau}^{dr}] - \tau
        &= \frac{1}{T} \sum_{t=1}^T
            \mathbb{E}_{\mathcal{D}}
                \left[\mathbb{E}_{\mathbf{Z}}
                    [
                        Y_{t}(D_{t}(\mathbf{Z}^{<t}, 1)) \mid \mathcal{D}
                    ] - 
                    \mathbb{E}_{\mathbf{Z}}
                    [
                        Y_{t}(D_{t}(\mathbf{Z}^{<t}, 0)) \mid \mathcal{D}
                    ] 
                \right] \\
                & \quad\quad - 
                \frac{1}{T} \sum_{t=1}^T
                \mathbb{E}_{\mathcal{D}}
                [
                \left(
                    Y_{t}(D_{t}(\mathbf{1}^{\leq t}))
                    - 
                    Y_{t}(D_{t}(\mathbf{0}^{\leq t}))
                \right)
                ]\\
        &= \frac{1}{T} \sum_{t=1}^T
            \underbrace
            {\mathbb{E}_{\mathcal{D}} 
                \left[ 
                    \mathbb{E}_Z[Y_{t}(D_{t}(\mathbf{Z}^{<t}, 1)) \mid \mathcal{D}]
                    - 
                     Y_{t}(D_{t}(\mathbf{1}^{\leq t}))
                \right]}_{ =0 \text{ for } t=1, \ < 0 \text{ for all } t \geq 2}
                \\
            & \quad\quad -
            \underbrace
            {\mathbb{E}_{\mathcal{D}} 
                \left[ 
                    \mathbb{E}_Z[Y_{t}(D_{t}(\mathbf{Z}^{<t}, 0)) \mid \mathcal{D}]
                    - 
                     Y_{t}(D_{t}(\mathbf{0}^{\leq t}))
                \right]}_{ =0 \text{ for } t=1, \ > 0 \text{ for all } t \geq 2} \\
            &< 0 \text{ for all } T \geq 2
    \end{align*}
    A symmetric argument shows that $\mathbb{E}[\hat{\tau}^{dr}] - \tau > 0$ when $Y_{it}(D_{it}(\mathbf{Z}^{\leq t}))$ is strictly decreasing in the cumulative number of exposures to assistance.
\end{proof}

\subsection{Bias Expressions Under Illustrative Data-generating Processes}\label{app:proofs:bias}
First, notice that we can write $Y_{it}(D_{it}(\mathbf{Z}))$ as
    \begin{align*}
        Y_{it}(D_{it}(\mathbf{Z})) = Q_{it}(\mathbf{Z})A_{it} + (1 - Q_{it}(\mathbf{Z}))(1 - A_{it})
    \end{align*}
    Thus, by definition of the population-level version of $\tau_t$,
    \begin{align*}
        \tau_t 
        &= \mathbb{E}[Y_t(D_t(\mathbf{1}^{\leq t})) - Y_t(D_t(\mathbf{0}^{\leq t}))] \\
        &= \mathbb{E}[\left( Q_t(\mathbf{1}^{\leq t})A_t  + (1 - Q_t(\mathbf{1}^{\leq t}))(1 - A_t) \right)
            - \left( Q_t(\mathbf{0}^{\leq t})A_t  + (1 - Q_t(\mathbf{0}^{\leq t}))(1 - A_t) \right) ] \\
        &= \mathbb{E}[(2A_t - 1) (Q_t(\mathbf{1}^{\leq t}) - Q_t(\mathbf{0}^{\leq t}))]
    \end{align*}
Similarly, we can write $\mathbb{E}_{\mathbf{Z}, \mathcal{D}}[\hat{\tau}^{dr}]$ as
    \begin{align*}
        \mathbb{E}_{\mathbf{Z}, D}[\hat{\tau}^{dr}]
        &= \frac{1}{NT} \sum_{i=1}^{N}\sum_{t=1}^T
            \mathbb{E}[Y_{t} \mid Z_{t} = 1] - \mathbb{E}[Y_{t} \mid Z_{t} = 0] \\
        &= \frac{1}{T} \sum_{t=1}^T 
            \mathbb{E}[Y_{t}(D(\mathbf{Z}^{<t}, 1))] - \mathbb{E}[Y_{t}(D(\mathbf{Z}^{<t}, 0))] \\
        &= \frac{1}{T} \sum_{t=1}^T \mathbb{E} [(2A_t - 1)(Q_t(\mathbf{Z}^{< t}, 1) - Q_t(\mathbf{Z}^{< t}, 0))]
    \end{align*}
For the Automation Bias and Alert Fatigue data-generating processes, $A_t \indep Q_t(\mathbf{Z})$, so these can be further simplified to 
    \begin{align*}
        \tau_t &= (2\mu_A - 1) \mathbb{E} [Q_t(\mathbf{1}^{\leq t}) - Q_t(\mathbf{0}^{\leq t})] \\
         \mathbb{E}_{\mathbf{Z}, \mathcal{D}_{ab}}[\hat{\tau}^{dr}] &= (2\mu_A - 1) \frac{1}{T} \sum_{t=1}^T \mathbb{E} [Q_t(\mathbf{Z}^{<t}, 1) - Q_t(\mathbf{Z}^{<t}, 0)]
    \end{align*}

\begin{proof}[Proof of Proposition \ref{prop:ab-bias} (Automation Bias DGP)]
    We can write $\mathbb{E}_{\mathcal{D}_{ab}}[Q_t(\mathbf{1}^{\leq t})]$ and $\mathbb{E}_{\mathcal{D}_{ab}}[Q_t(\mathbf{0}^{\leq t})]$ as 
    \begin{align*}
        \mathbb{E}_{\mathcal{D}_{ab}}[Q_t(\mathbf{1}^{\leq t})] 
            &= \mathbb{E}[p_{Q_{ab}}(Z_t = 1, S_t|_\mathbf{Z}=\mathbf{1})] \\
            &= q_0 + \beta_Z + \beta_S S_t|_{\mathbf{Z} = \mathbf{1}} \\
            &= q_0 + \beta_Z + \beta_S (1 - \rho^{t-1}) \\
        \mathbb{E}_{\mathcal{D}_{ab}}[Q_t(\mathbf{0}^{\leq t})]
            &= \mathbb{E}[p_{Q_{ab}}(Z_t = 0, S_t|_\mathbf{Z}=\mathbf{0})] \\
            &= q_0
    \end{align*}
    where $S_t|_{\mathbf{Z} = \mathbf{1}}$ denotes the value of $S_t$ when $\mathbf{Z}^{\leq t} = \mathbf{1}$ and is determined by solving the recurrence relation
    \begin{align*}
        S_t 
        &= \rho S_{t-1} + (1 - \rho)Z_{t-1} \\
        &= \rho S_{t-1} + (1 - \rho) ,\quad S_1 = 0
    \end{align*}
    We thus have 
    \begin{align*}
        \tau_T &= (2\mu_A - 1) (\beta_Z + \beta_S(1 - \rho^{T-1})) \\
        \tau &= \frac{1}{T}\sum_{t=1}^T \tau_t 
            = (2\mu_A - 1)
            \left ( 
                \beta_Z + \beta_S 
                    \left( 
                        1 - \frac{1 - \rho^T}{T(1 - \rho)}
                    \right)
            \right)
    \end{align*}
    We can write $\mathbb{E}_{\mathbf{Z}, \mathcal{D}_{ab}}[Q(\mathbf{Z}^{<t}, 1)]$ and $\mathbb{E}_{\mathbf{Z}, \mathcal{D}_{ab}}[Q(\mathbf{Z}^{<t}, 0)]$ as
    \begin{align*}
        \mathbb{E}_{\mathbf{Z}, \mathcal{D}_{ab}}[Q(\mathbf{Z}^{<t}, 1)] 
        &= \mathbb{E}_{\mathbf{Z}}[p_{Q_{ab}}(Z_t = 1, S_t|_\mathbf{Z})] \\
        &= q_0 + \beta_Z + \beta_S \mathbb{E}_{\mathbf{Z}}[S_t|_\mathbf{Z}]] \\
        &= q_0 + \beta_Z + p\beta_S(1 - \rho^{t-1}) \\
        \mathbb{E}_{\mathbf{Z}, \mathcal{D}_{ab}}[Q(\mathbf{Z}^{<t}, 0)]
        &= \mathbb{E}_{\mathbf{Z}}[p_{Q_{ab}}(Z_t = 0, S_t|_\mathbf{Z})] \\
        &= q_0
    \end{align*}
    where $\mathbb{E}_{\mathbf{Z}}[S_t|_\mathbf{Z}]$ is determined by solving the recurrence relation
    \begin{align*}
        \mathbb{E}_{\mathbf{Z}}[S_t|_\mathbf{Z}] &= \rho \mathbb{E}_{\mathbf{Z}}[S_{t-1}|_\mathbf{Z}] + (1 - \rho) \mathbb{E}[Z_{t-1}] \\
        &= \rho \mathbb{E}_{\mathbf{Z}}[S_{t-1}|_\mathbf{Z}] + p(1 - \rho)
    \end{align*}
    We thus have
    \begin{align*}
        \mathbb{E}_{\mathbf{Z}, \mathcal{D}_{ab}}[\hat{\tau}^{dr}] 
        = (2\mu_A - 1) \left(\beta_Z + p  \beta_S \left(1 - \frac{1- \rho^T}{T(1 - \rho)}\right) \right)
    \end{align*}
    The bias expressions then follow.
\end{proof}

\begin{proof}[Proof of Proposition \ref{prop:af-bias} (Alert Fatigue DGP)]
     We can write $\mathbb{E}_{\mathcal{D}_{af}}[Q_t(\mathbf{1}^{\leq t})]$ and $\mathbb{E}_{\mathcal{D}_{af}}[Q_t(\mathbf{0}^{\leq t})]$ as
     \begin{align*}
         \mathbb{E}_{\mathcal{D}_{af}}[Q_{t}(\mathbf{1}^{\leq t})] 
         &= \mathbb{E}[p_{Q_{af}}(Z_t = 1, S_t|_{\mathbf{Z} = 1})] \\
         &= q_0 + \beta_Z + \beta_S \mathbb{E} [S_t |_{\mathbf{Z} = 1}] \\
         &= q_0 + \beta_Z + \beta_S (1 - \mu_A) \left(1 - \rho^{t-1}\right) \\
          \mathbb{E}_{\mathcal{D}_{af}}[Q_{t}(\mathbf{0}^{\leq t})]
        &= q_0
     \end{align*}
     where $\mathbb{E}[S_t|_{\mathbf{Z} = \mathbf{1}}]$ is determined by solving the recurrence relation,
    \begin{align*}
        \mathbb{E}[S_t \mid_{\mathbf{Z} = 1}] &= \rho \mathbb{E}[S_{t-1} \mid_{\mathbf{Z} = 1}] + (1 - \mathbb{E}[A_{t-1}])(1 - \rho) \\
        &=\rho \mathbb{E}[S_{t-1} \mid_{\mathbf{Z} = 1}] + (1 - \mu_A)(1 - \rho) , \quad S_1 = 0
    \end{align*}
We thus have
\begin{align*}
    \tau_T &= (2\mu_A - 1) \left( \beta_Z + \beta_S (1 - \mu_A) \left(1 - \rho^{T-1}\right) \right) \\
    \tau &= \frac{1}{T}\sum_{t=1}^T \tau_t 
    = (2\mu_A - 1)
                \left(\beta_Z + \beta_S (1 - \mu_A)
                    \left(1 - \frac{1 - \rho^{T}}{T(1 - \rho)}
                    \right)
                \right)
\end{align*}
    We can write $\mathbb{E}_{\mathbf{Z}, \mathcal{D}_{af}}[Q(\mathbf{Z}^{<t}, 1)]$ and $\mathbb{E}_{\mathbf{Z}, \mathcal{D}_{af}}[Q(\mathbf{Z}^{<t}, 0)]$ as
    \begin{align*}
        \mathbb{E}_{\mathbf{Z}, \mathcal{D}_{af}}[Q(\mathbf{Z}^{<t}, 1)] 
        &= \mathbb{E}[p_{Q_{af}}(Z_t = 1, S_t|_\mathbf{Z})] \\
        &= q_0 + \beta_Z + \beta_S \mathbb{E}_{\mathbf{Z}}[S_t|_\mathbf{Z}] \\
        &= q_0 + \beta_Z + p\beta_S(1 - \mu_A) (1 - \rho^{t-1}) \\
        \mathbb{E}_{\mathbf{Z}, \mathcal{D}_{af}}[Q(\mathbf{Z}^{<t}, 0)]
        &= \mathbb{E}_{\mathbf{Z}}[p_{Q_{af}}(Z_t = 0, S_t|_\mathbf{Z})] \\
        &= q_0
    \end{align*}
    where $\mathbb{E}_{\mathbf{Z}}[S_t|_\mathbf{Z}]$ is determined by solving the recurrence relation
    \begin{align*}
        \mathbb{E}_{\mathbf{Z}}[S_t|_\mathbf{Z}] &= \rho \mathbb{E}_{\mathbf{Z}}[S_{t-1}|_\mathbf{Z}] + (1 - \rho)\mathbb{E}[Z_{t-1}(1 - A_{t-1})] \\
        &= \rho \mathbb{E}_{\mathbf{Z}}[S_{t-1}|_\mathbf{Z}] + p(1 - \mu_A)(1 - \rho) \quad \text{ since } A_{it} \indep Z_{it}
    \end{align*}
    We thus have
    \begin{align*}
        \mathbb{E}_{\mathbf{Z}, \mathcal{D}_{af}}[\hat{\tau}^{dr}] 
        = (2\mu_A - 1) \left(\beta_Z + p \beta_S (1 - \mu_A) \left(1 - \frac{1- \rho^T}{T(1 - \rho)}\right) \right)
    \end{align*}
    The bias expressions then follow.
\end{proof}

\begin{proof}[Proof of Proposition \ref{prop:cr-bias} (Calibrated Reliance DGP)]
    We can write $\mathbb{E}_{\mathcal{D}_{cr}}[Q_t(\mathbf{1}^{\leq t})A_t]$, $\mathbb{E}_{\mathcal{D}_{cr}}[Q_t(\mathbf{0}^{\leq t})A_t]$, $\mathbb{E}_{\mathcal{D}_{cr}}[Q_t(\mathbf{1}^{\leq t})]$, and $\mathbb{E}_{\mathcal{D}_{cr}}[Q_t(\mathbf{0}^{\leq t})]$ as
    \begin{align*}
        \mathbb{E}_{\mathcal{D}_{cr}}[Q_t(\mathbf{1}^{\leq t})] 
            &= \mathbb{E}[p_{Q_{cr}}(Z_t = 1, S_t |_{\mathbf{Z} = \mathbf{1}})] \\
            &= q_0 + \beta_Z + \beta_S (2\mathbb{E}[A_t] - 1) S_t |_{\mathbf{Z} = \mathbf{1}} \\
            &= q_0 + \beta_Z + \beta_S(2\mu_A - 1) (1 - \rho^{t-1}) \\
        \mathbb{E}_{\mathcal{D}_{cr}}[Q_t(\mathbf{1}^{\leq t})A_t] 
            &= \mathbb{E}[p_{Q_{cr}}(Z_t = 1, S_t |_{\mathbf{Z} = \mathbf{1}})A_t] \\
            &= q_0\mathbb{E}[A_t] + \beta_Z\mathbb{E}[A_t]  + \beta_S {E}[2A_t^2 - A_t] S_t |_{\mathbf{Z} = \mathbf{1}}  \\
            &= \mu_A\left(q_0 + \beta_Z + \beta_S (1 - \rho^{t-1})\right)
    \end{align*}
    where the last equality follows by simplifying $2A_t^2 - A_t = 2A_t - A_t = A_t$ for $A_t \in \{0, 1\}$ and by solving the recurrence relation for $S_t$, which is the same as that for the Automation Bias data-generating process.
    \begin{align*}
        \mathbb{E}_{\mathcal{D}_{cr}}[Q_t(\mathbf{0}^{\leq t})] 
            &= \mathbb{E}[p_{Q_{cr}}(Z_t = 0, S_t |_{\mathbf{Z} = \mathbf{0}})] \\
            &= q_0 \\
        \mathbb{E}_{\mathcal{D}_{cr}}[Q_t(\mathbf{0}^{\leq t})A_t] 
            &= \mathbb{E}[p_{Q_{cr}}(Z_t = 0, S_t |_{\mathbf{Z} = \mathbf{0}})A_t] \\
            &= q_0\mathbb{E}[A_t] \\
            &= \mu_Aq_0
    \end{align*}
    We thus have
    \begin{align*}
        \tau_T
        &= 2\mu_A\left(\beta_Z + \beta_S (1 - \rho^{T-1})\right) - (\beta_Z + \beta_S(2\mu_A - 1) (1 - \rho^{T-1})) \\
        &= (2\mu_A - 1)\beta_Z + \beta_S(1 - \rho^{T-1}) \\
        \tau &= \frac{1}{T}\sum_{t=1}^T \tau_t = (2\mu_A - 1) \beta_Z + \beta_S
                    \left(1 - \frac{1 - \rho^{T}}{T(1 - \rho)}
                    \right)
    \end{align*}
     We can write $\mathbb{E}_{\mathbf{Z}, \mathcal{D}_{cr}}[Q_t(\mathbf{Z}^{<t}, 1)A_t], \mathbb{E}_{\mathbf{Z}, \mathcal{D}_{cr}}[Q_t(\mathbf{Z}^{<t}, 1)]$, $\mathbb{E}_{\mathbf{Z}, \mathcal{D}_{cr}}[Q_t(\mathbf{Z}^{<t}, 0)A_t], \mathbb{E}_{\mathbf{Z}, \mathcal{D}_{cr}}[Q_t(\mathbf{Z}^{<t}, 0)]$ as
     \begin{align*}
         \mathbb{E}_{\mathbf{Z}, \mathcal{D}_{cr}}[Q_t(\mathbf{Z}^{<t}, 1)A_t]
         &= \mathbb{E}_{A_t}[\mathbb{E}_{\mathbf{Z}}[Q_t(\mathbf{Z}^{<t}, 1)A_t \mid A_t]] \\
         &= \mathbb{E}_{A_t}[\mathbb{E}_{\mathbf{Z}}[p_{cr}(Z_t = 1, S_t |_{\mathbf{Z}})A_t \mid A_t]] \\
         &= \mathbb{E}_{A_t}[q_0A_t + \beta_ZA_t + \beta_SA_t(2A_t - 1)\mathbb{E}_{\mathbf{Z}}[S_t|_{\mathbf{Z}}]] \\
         &= \mu_A(q_0 + \beta_Z + p\beta_S(1 - \rho^{t-1})) \\
        \mathbb{E}_{\mathbf{Z}, \mathcal{D}_{cr}}[Q_t(\mathbf{Z}^{<t}, 1)] 
        &= \mathbb{E}[p_{cr}(Z_t = 1, S_t |_{\mathbf{Z}})] \\
        &= q_0 + \beta_Z + \beta_S(2A_t - 1)\mathbb{E}_{\mathbf{Z}}[S_t|_{\mathbf{Z}}] \\
        &= q_0 + \beta_Z + p\beta_S(2\mu_A - 1)(1 - \rho^{t-1}) \\
        \mathbb{E}_{\mathbf{Z}, \mathcal{D}_{cr}}[Q_t(\mathbf{Z}^{<t}, 0)A_t]
         &= \mathbb{E}_{A_t}[\mathbb{E}_{\mathbf{Z}}[Q_t(\mathbf{Z}^{<t}, 0)A_t \mid A_t]] \\
         &= \mathbb{E}_{A_t}[\mathbb{E}_{\mathbf{Z}}[p_{cr}(Z_t = 0, S_t |_{\mathbf{Z}})A_t \mid A_t]] \\
         &= \mathbb{E}_{A_t}[q_0A_t] \\
         &= \mu_Aq_0 \\
        \mathbb{E}_{\mathbf{Z}, \mathcal{D}_{cr}}[Q_t(\mathbf{Z}^{<t}, 0)] 
        &= \mathbb{E}[p_{cr}(Z_t = 0, S_t |_{\mathbf{Z}})] \\
        &= q_0
     \end{align*}
     where $\mathbb{E}_{\mathbf{Z}}[S_t|_{\mathbf{Z}}]$ is determined by solving the recurrence relation
     \begin{align*}
         \mathbb{E}[S_t|_{\mathbf{Z}}] 
         &= \rho\mathbb{E}[S_{t-1}|_{\mathbf{Z}}] + (1-\rho)\mathbb{E}[Z_{t-1}] \\ 
         &= \rho\mathbb{E}[S_{t-1}|_{\mathbf{Z}}] + p(1-\rho)
     \end{align*}
    We thus have
     \begin{align*}
         \mathbb{E}_{\mathbf{Z}, \mathcal{D}_{cr}}[\hat{\tau}^{dr}] 
         &= (2\mu_A - 1) \beta_Z + p \beta_S \left( 1 - \frac{1 - \rho^{T}}{T({1 - \rho})} \right)
     \end{align*}
     The bias expressions then follow.
\end{proof}

\subsection{Conditional Effects}\label{app:proofs:conditional}

\begin{proof}[Proof of Proposition \ref{prop:cond-ident} (Identifiability of $\tau^{YI}, \delta^{YI}$)]\label{app:proofs:conditional:identifiability}
    We follow the argument of \cite{Imbens1994LATE}, with the concordance indicator $Q_t$ playing the role of treatment and $\mathbf{Z}^{\leq t}$ the role of instrument.

    First, note that we can decompose $\tau_t$ as
    \begin{align*}
        \tau_t 
            &= \mathbb{E} [Y_t(D_t(\mathbf{1}^{\leq t})) - Y_t(D_t(\mathbf{0}^{\leq t})) \mid Q_t(\mathbf{1}^{\leq t})\ = 1, Q_t(\mathbf{0}^{\leq t})\ = 1] \cdot P(Q_t(\mathbf{1}^{\leq t})\ = 1, Q_t(\mathbf{0}^{\leq t})\ = 1) \\
            &\quad+ \mathbb{E} [Y_t(D_t(\mathbf{1}^{\leq t})) - Y_t(D_t(\mathbf{0}^{\leq t})) \mid Q_t(\mathbf{1}^{\leq t})\ = 1, Q_t(\mathbf{0}^{\leq t})\ = 0] \cdot P(Q_t(\mathbf{1}^{\leq t})\ = 1, Q_t(\mathbf{0}^{\leq t})\ = 0) \\
            &\quad+ \mathbb{E} [Y_t(D_t(\mathbf{1}^{\leq t})) - Y_t(D_t(\mathbf{0}^{\leq t})) \mid Q_t(\mathbf{1}^{\leq t})\ = 0, Q_t(\mathbf{0}^{\leq t})\ = 0] \cdot P(Q_t(\mathbf{1}^{\leq t})\ = 0, Q_t(\mathbf{0}^{\leq t})\ = 0) \\
            &\quad+ \mathbb{E} [Y_t(D_t(\mathbf{1}^{\leq t})) - Y_t(D_t(\mathbf{0}^{\leq t})) \mid Q_t(\mathbf{1}^{\leq t})\ = 0, Q_t(\mathbf{0}^{\leq t})\ = 1] \cdot P(Q_t(\mathbf{1}^{\leq t})\ = 0, Q_t(\mathbf{0}^{\leq t})\ = 1) \\
            &= \mathbb{E} [Y_t(D_t(\mathbf{1}^{\leq t})) - Y_t(D_t(\mathbf{0}^{\leq t})) \mid Q_t(\mathbf{1}^{\leq t})\ = 1, Q_t(\mathbf{0}^{\leq t})\ = 1] \cdot P(Q_t(\mathbf{1}^{\leq t})\ = 1, Q_t(\mathbf{0}^{\leq t})\ = 1) \\
            &\quad+ \mathbb{E} [Y_t(D_t(\mathbf{1}^{\leq t})) - Y_t(D_t(\mathbf{0}^{\leq t})) \mid Q_t(\mathbf{1}^{\leq t})\ = 1, Q_t(\mathbf{0}^{\leq t})\ = 0] \cdot P(Q_t(\mathbf{1}^{\leq t})\ = 1, Q_t(\mathbf{0}^{\leq t})\ = 0) \\
            &\quad+ \mathbb{E} [Y_t(D_t(\mathbf{1}^{\leq t})) - Y_t(D_t(\mathbf{0}^{\leq t})) \mid Q_t(\mathbf{1}^{\leq t})\ = 0, Q_t(\mathbf{0}^{\leq t})\ = 0] \cdot P(Q_t(\mathbf{1}^{\leq t})\ = 0, Q_t(\mathbf{0}^{\leq t})\ = 0) \\
            &= \mathbb{E} [Y_t(D_t(\mathbf{1}^{\leq t})) - Y_t(D_t(\mathbf{0}^{\leq t})) \mid Q_t(\mathbf{1}^{\leq t})\ = 1, Q_t(\mathbf{0}^{\leq t})\ = 0] \cdot P(Q_t(\mathbf{1}^{\leq t})\ = 1, Q_t(\mathbf{0}^{\leq t})\ = 0) \\
            &= \tau_t^{YI} \cdot P(Q_t(\mathbf{1}^{\leq t})\ = 1, Q_t(\mathbf{0}^{\leq t})\ = 0)
    \end{align*}
    The first equality is an application of the law of total expectation and the partial interference (\ref{assump-partial-interf}) and non-anticipation (\ref{assump-nonanticip}) assumptions. The second equality follows because of the monotonicity assumption (\ref{assump:monotonicity}), and the third equality follows because
    \begin{align*}
        &\mathbb{E} [Y_t(D_t(\mathbf{1}^{\leq t})) - Y_t(D_t(\mathbf{0}^{\leq t})) \mid Q_t(\mathbf{1}^{\leq t})\ = 0, Q_t(\mathbf{0}^{\leq t})\ = 0] \cdot P(Q_t(\mathbf{1}^{\leq t})\ = 0, Q_t(\mathbf{0}^{\leq t})\ = 0) \\
        &= \mathbb{E} [Y_t(1 - R_t) - Y_t(1 - R_t) \mid Q_t(\mathbf{1}^{\leq t})\ = 0, Q_t(\mathbf{0}^{\leq t})\ = 0]\cdot P(Q_t(\mathbf{1}^{\leq t})\ = 0, Q_t(\mathbf{0}^{\leq t})\ = 0) \\
        &= 0 
    \end{align*}
    and
        \begin{align*}
        &\mathbb{E} [Y_t(D_t(\mathbf{1}^{\leq t})) - Y_t(D_t(\mathbf{0}^{\leq t})) \mid Q_t(\mathbf{1}^{\leq t})\ = 1, Q_t(\mathbf{0}^{\leq t})\ = 1] \cdot P(Q_t(\mathbf{1}^{\leq t})\ = 1, Q_t(\mathbf{0}^{\leq t})\ = 1) \\
        &= \mathbb{E} [Y_t(R_t) - Y_t(R_t) \mid Q_t(\mathbf{1}^{\leq t})\ = 1, Q_t(\mathbf{0}^{\leq t})\ = 1]\cdot P(Q_t(\mathbf{1}^{\leq t})\ = 1, Q_t(\mathbf{0}^{\leq t})\ = 1) \\
        &= 0 
    \end{align*}
    where we can replace $\mathbb{E}[Y_t(D_t(\mathbf{z}^{\leq t})) \mid Q_t(\mathbf{z}^{\leq t}) = 1]$  with $\mathbb{E}[Y_t(R_t) \mid Q_t(\mathbf{z}^{\leq t}) = 1]$ and $\mathbb{E}[Y_t(D_t(\mathbf{z}^{\leq t})) \mid Q_t(\mathbf{z}^{\leq t}) = 0]$  with $\mathbb{E}[Y_t(1 - R_t) \mid Q_t(\mathbf{z}^{\leq t}) = 0]$ under the full mediation assumption (\ref{assump-mediation}). 
    
    Again, under the monotonicity assumption (\ref{assump:monotonicity}), we have that
    \begin{align*}
        P(Q_t(\mathbf{1}^{\leq t})\ = 1, Q_t(\mathbf{0}^{\leq t})\ = 1) 
        &= P(Q_t(\mathbf{0}^{\leq t}) = 1) - P(Q_t(\mathbf{1}^{\leq t})\ = 0, Q_t(\mathbf{0}^{\leq t})\ = 1) \\
        &= P(Q_t(\mathbf{0}^{\leq t}) = 1) \\
        P(Q_t(\mathbf{1}^{\leq t})\ = 0, Q_t(\mathbf{0}^{\leq t})\ = 0) 
        &= P(Q_t(\mathbf{1}^{\leq t}) = 0) - P(Q_t(\mathbf{1}^{\leq t})\ = 0, Q_t(\mathbf{0}^{\leq t})\ = 1) \\
        &= P(Q_t(\mathbf{1}^{\leq t}) = 0)
    \end{align*}
    Thus, we have that
    \begin{align*}
        P(Q_t(\mathbf{1}^{\leq t})\ = 1, Q_t(\mathbf{0}^{\leq t})\ = 0) 
        &= 1 - P(Q_t(\mathbf{0}^{\leq t}) = 1) - P(Q_t(\mathbf{1}^{\leq t}) = 0) \\
        &= (1 - P(Q_t(\mathbf{0}^{\leq t}) = 1)) - (1 - P(Q_t(\mathbf{1}^{\leq t}) = 1)) \\
        &= P(Q_t(\mathbf{1}^{\leq t}) = 1) - P(Q_t(\mathbf{0}^{\leq t}) = 1)
    \end{align*}
    So $\tau_t^{YI}$ becomes
    \begin{align*}
        \tau_t^{YI} 
        &= \frac{\mathbb{E}[Y_t(D_t(\mathbf{1}^{\leq t})) - Y_t(D_t(\mathbf{0}^{\leq t}))]}{P(Q_t(\mathbf{1}^{\leq t}) = 1) - P(Q_t(\mathbf{0}^{\leq t}) = 1)} \\
        &= \frac
        {
        \mathbb{E}[Y_t(D_t(\mathbf{1}^{\leq t})) \mid \mathbf{Z}^{\leq t} = \mathbf{1}^{\leq t} ] - \mathbb{E}[Y_t(D_t(\mathbf{0}^{\leq t})) \mid \mathbf{Z}^{\leq t} = \mathbf{0}^{\leq t}]
        }
        {
        P(D_t(\mathbf{1}^{\leq t}) = R_t \mid \mathbf{Z}^{\leq t} = \mathbf{1}^{\leq t}) - P(D_t(\mathbf{0}^{\leq t}) = R_t \mid \mathbf{Z}^{\leq t} = \mathbf{0}^{\leq t})
        } \\
        &= \frac
        {
        \mathbb{E}[Y_t \mid \mathbf{Z}^{\leq t} = \mathbf{1}^{\leq t}] - E[Y_t \mid \mathbf{Z}^{\leq t} = \mathbf{0}^{\leq t}]}
        {\mathbb{E}[Q_t \mid \mathbf{Z}^{\leq t} = \mathbf{1}^{\leq t}] - E[Q_t \mid \mathbf{Z}^{\leq t} = \mathbf{0}^{\leq t}]}
    \end{align*}
    where the second equality follows by ignorability (\ref{assump:ignorability}) and the third equality follows by consistency.
    
    The proof for identifiability of $\delta_t^{YI}$ is identical, replacing the contrast between $\mathbf{Z}^{\leq t} = \mathbf{1}^{\leq t}$ and $\mathbf{Z}^{\leq t} = \mathbf{0}^{\leq t}$ with the contrast between $\mathbf{Z}^{\leq t} = (\mathbf{0}^{< t}, 1)$ and $\mathbf{Z}^{\leq t} = \mathbf{0}^{\leq t}$. 
\end{proof}

\begin{proof}[Proof for Example \ref{app:ex:ab-af-yielding-effects} (Conditional yielding effects for example data-generating processes)]
    As shown in the derivation for the bias of the decision-randomized estimator under the Automation Bias DGP,
    \begin{align*}
        \tau_t = (2\mu_A - 1) \mathbb{E} [Q_t(\mathbf{1}^{\leq t}) - Q_t(\mathbf{0}^{\leq t})] 
    \end{align*}
    It follows directly that
    \begin{align*}
        \delta_t = (2\mu_A - 1) \mathbb{E} [Q_t(\mathbf{0}^{<t}, 1) - Q_t(\mathbf{0}^{\leq t})]
    \end{align*}
    Therefore, using the formulae for $\tau_t^{YI}$ and $\delta_t^{YI}$,
    \begin{align*}
        \tau_t^{YI} &= \frac{\tau_t}{\mathbb{E} [Q_t(\mathbf{1}^{\leq t}) - Q_t(\mathbf{0}^{\leq t})]} = 2\mu_A - 1 \\
        \delta_t^{YI} &= \frac{\delta_t}{\mathbb{E} [Q_t(\mathbf{0}^{<t}, 1) - Q_t(\mathbf{0}^{\leq t})]} = 2\mu_A - 1 \\
    \end{align*}
    The argument proves the same result for the Alert Fatigue DGP.

    For the Calibrated Reliance DGP, because $A_t$ and $Q_t(\mathbf{Z})$ are not independent, we obtain a different result. First, from the derivation for the bias of the decision-randomized estimator under the Calibrated Reliance DGP, we have that 
    \begin{align*}
        \tau_t &= (2\mu_A - 1) \beta_Z + \beta_S(1 - \rho^{t-1}) \\
        \mathbb{E}[Q_t(\mathbf{1}^{\leq t}) - Q_t(\mathbf{0}^{\leq t})] &= \beta_Z + \beta_S(2\mu_A - 1)(1 - \rho^{t-1})
    \end{align*}
    A trivially similar derivation shows that
    \begin{align*}
        \delta_t &= (2\mu_A - 1)\beta_Z \\
        \mathbb{E}[Q_t(\mathbf{0}^{<t}, 1) - Q_t(\mathbf{0}^{\leq t})] &= \beta_Z 
    \end{align*}
    Thus,
    \begin{align*}
        \tau_t^{YI} &= \frac{\tau_t}{\mathbb{E} [Q_t(\mathbf{1}^{\leq t}) - Q_t(\mathbf{0}^{\leq t})]} = \frac{(2\mu_A - 1)\beta_Z + \beta_S(1 - \rho^{t-1})}{\beta_Z + \beta_S(2\mu_A - 1)(1 - \rho^{t-1})} \\
        \delta_t^{YI} &= \frac{\delta_t}{\mathbb{E} [Q_t(\mathbf{0}^{<t}, 1) - Q_t(\mathbf{0}^{\leq t})]} = 2\mu_A - 1\\
    \end{align*}
\end{proof}

\begin{proof}[Proof of Proposition \ref{prop:profile} (Profiling) ]
    The proof follows that of \cite{Marbach2020Profiling}.

    First, applying the law of total expectation and the partial interference (\ref{assump-partial-interf}) and non-anticipation (\ref{assump-nonanticip}) assumptions,  we have a similar decomposition as the one in the identifiability proof for $\tau^{YI}$, namely
    \begin{align*}
        \mathbb{E}[X_t]
            &= \mathbb{E} [X_t \mid Q_t(\mathbf{1}^{\leq t})\ = 1, Q_t(\mathbf{0}^{\leq t})\ = 1] \cdot P(Q_t(\mathbf{1}^{\leq t})\ = 1, Q_t(\mathbf{0}^{\leq t})\ = 1) \\
            &\quad+ \mathbb{E} [X_t \mid Q_t(\mathbf{1}^{\leq t})\ = 1, Q_t(\mathbf{0}^{\leq t})\ = 0] \cdot P(Q_t(\mathbf{1}^{\leq t})\ = 1, Q_t(\mathbf{0}^{\leq t})\ = 0) \\
            &\quad+ \mathbb{E} [X_t \mid Q_t(\mathbf{1}^{\leq t})\ = 0, Q_t(\mathbf{0}^{\leq t})\ = 0] \cdot P(Q_t(\mathbf{1}^{\leq t})\ = 0, Q_t(\mathbf{0}^{\leq t})\ = 0) \\
            &\quad+ \mathbb{E} [X_t \mid Q_t(\mathbf{1}^{\leq t})\ = 0, Q_t(\mathbf{0}^{\leq t})\ = 1] \cdot P(Q_t(\mathbf{1}^{\leq t})\ = 0, Q_t(\mathbf{0}^{\leq t})\ = 1) \\
            &=\mathbb{E} [X_t \mid Q_t(\mathbf{1}^{\leq t})\ = 1, Q_t(\mathbf{0}^{\leq t})\ = 1] \cdot P(Q_t(\mathbf{1}^{\leq t})\ = 1, Q_t(\mathbf{0}^{\leq t})\ = 1) \\
            &\quad+ \mathbb{E} [X_t \mid Q_t(\mathbf{1}^{\leq t})\ = 1, Q_t(\mathbf{0}^{\leq t})\ = 0] \cdot P(Q_t(\mathbf{1}^{\leq t})\ = 1, Q_t(\mathbf{0}^{\leq t})\ = 0) \\
            &\quad+ \mathbb{E} [X_t \mid Q_t(\mathbf{1}^{\leq t})\ = 0, Q_t(\mathbf{0}^{\leq t})\ = 0] \cdot P(Q_t(\mathbf{1}^{\leq t})\ = 0, Q_t(\mathbf{0}^{\leq t})\ = 0) 
    \end{align*}
    Note that we can write, for the always aligned decisions,
    \begin{align*}
        \mathbb{E}[X_t \mid Q_t(\mathbf{1}^{\leq t}) = 1, Q_t(\mathbf{0}^{\leq t}) = 1]
        &= \mathbb{E}[X_t \mid  Q_t(\mathbf{0}^{\leq t}) = 1]  \\
        & = \mathbb{E}[X_t \mid D_t(\mathbf{0}^{\leq t}) = R_t] \\
        & = \mathbb{E}[X_t \mid D_t(\mathbf{0}^{\leq t}) = R_t, \mathbf{Z}^{\leq t} = \mathbf{0}^{\leq t}] \\
        &= \mathbb{E}[X_t \mid Q_t = 1, \mathbf{Z}^{\leq t}  = \mathbf{0}^{\leq t}]
    \end{align*}
    The first equality follows by the monotonicity assumption \ref{assump:monotonicity} -- since $Q_t(\mathbf{1}^{\leq t}) \geq Q_t(\mathbf{0}^{\leq t})$, if it is known that $Q_t(\mathbf{0}^{\leq t}) = 1$, it implies that $Q_t(\mathbf{1}^{\leq t}) = 1$, so we can drop $Q_t(\mathbf{1}^{\leq t}) = 1$ from the expectation. The second equality follows by definition of $Q_t(\mathbf{0}^{\leq t})$, the third equality by the ignorability assumption \ref{assump:ignorability-x}, and the fourth equality by a standard consistency assumption.

    By the same logic, for the never aligned decisions,
    \begin{align*}
        \mathbb{E}[X_t \mid Q_t(\mathbf{1}^{\leq t}) = 0, Q_t(\mathbf{0}^{\leq t}) = 0]
        &= \mathbb{E}[X_t \mid  Q_t(\mathbf{1}^{\leq t}) = 0]  \\
        & = \mathbb{E}[X_t \mid D_t(\mathbf{1}^{\leq t}) = 1 - R_t] \\
        & = \mathbb{E}[X_t \mid D_t(\mathbf{1}^{\leq t}) = 1 - R_t, \mathbf{Z}^{\leq t}  = \mathbf{1}^{\leq t}] \\
        &= \mathbb{E}[X_t \mid Q_t = 0, \mathbf{Z}^{\leq t}  = \mathbf{1}^{\leq t}]
    \end{align*}

    As shown in the previous proof of identifiability for $\tau^{YI}$, we have that
    \begin{align*}
        P(Q_t(\mathbf{1}^{\leq t})\ = 1, Q_t(\mathbf{0}^{\leq t})\ = 1)
        & = P(Q_t(\mathbf{0}^{\leq t}) = 1)\\
        P(Q_t(\mathbf{1}^{\leq t})\ = 0, Q_t(\mathbf{0}^{\leq t})\ = 0)
        &= P(Q_t(\mathbf{1}^{\leq t}) = 0)
    \end{align*}
    By ignorability and consistency, we have
    \begin{align*}
        P(Q_t(\mathbf{0}^{\leq t}) = 1) 
        &= P(D_t(\mathbf{0}^{\leq t}) = R_t) \\
        &= P(D_t(\mathbf{0}^{\leq t}) = R_t \mid \mathbf{Z}^{\leq t} = \mathbf{0}^{\leq t}) \\
        &= P(D_t = R_t \mid \mathbf{Z}^{\leq t} = \mathbf{0}^{\leq t}) \\
        &= P(Q_t = 1 \mid \mathbf{Z}^{\leq t} = \mathbf{0}^{\leq t}) \\
        P(Q_t(\mathbf{1}^{\leq t}) = 0) 
        &= P(D_t(\mathbf{1}^{\leq t}) = 1 - R_t) \\
        &= P(D_t(\mathbf{1}^{\leq t}) = 1 - R_t \mid \mathbf{Z}^{\leq t} = \mathbf{1}^{\leq t}) \\
        &= P(D_t = 1 - R_t \mid \mathbf{Z}^{\leq t} = \mathbf{1}^{\leq t}) \\
        &= P(Q_t = 0 \mid \mathbf{Z}^{\leq t} = \mathbf{1}^{\leq t})
    \end{align*} 
    Plugging these values back into the decomposition expression for $\mathbb{E}[X_t]$ gives us
    \begin{align*}
        \mathbb{E}[X_t] 
        &= \mathbb{E} [X_t \mid Q_t = 1, \mathbf{Z}^{\leq t} = \mathbf{0}^{\leq t}] \cdot P(Q_t = 1 \mid \mathbf{Z}^{\leq t} = \mathbf{0}^{\leq t}) \\
            &\quad+ \mathbb{E} [X_t \mid Q_t = 0, \mathbf{Z}^{\leq t} = \mathbf{1}^{\leq t}] \cdot P(Q_t = 0 \mid \mathbf{Z}^{\leq t} = \mathbf{1}^{\leq t}) \\
            &\quad+ \mathbb{E} [X_t \mid Q_t(\mathbf{1}^{\leq t})\ = 1, Q_t(\mathbf{0}^{\leq t})\ = 0] \cdot P(Q_t(\mathbf{1}^{\leq t})\ = 1, Q_t(\mathbf{0}^{\leq t})\ = 0) 
    \end{align*}
    Rearranging to solve for $\mathbb{E} [X_t \mid Q_t(\mathbf{1}^{\leq t})\ = 1, Q_t(\mathbf{0}^{\leq t})\ = 0]$ and using the fact that $P(Q_t(\mathbf{1}^{\leq t})\ = 1, Q_t(\mathbf{0}^{\leq t})\ = 0) = \mathbb{E}[Q_t \mid \mathbf{Z}^{\leq t} = \mathbf{1}^{\leq t}] - \mathbb{E}[Q_t \mid \mathbf{Z}^{\leq t} = \mathbf{0}^{\leq t}]$ as shown in the proof of identifiability for $\tau^{YI}$ gives us the desired result. The same proof argument, replacing $Q_t(\mathbf{1}^{\leq t})$ with $Q_t(\mathbf{0}^{< t}, 1)$, gives the result for $\mathbb{E}[X_t \mid I_t = YI]$.  
\end{proof}

\begin{proof}[Proof of Example \ref{app:ex:ab-af-cr-profile}]
    Given the data-generating process for $X_{it}$, we have that
    \begin{align*}
        \mathbb{E}[X_t] &= 2\gamma - 1 \\
        \mathbb{E}[X_t \mid Q_t(\mathbf{1}^{\leq t}) = 0]
        &= P(X_t =1 \mid Q_t(\mathbf{1}^{\leq t}) = 0) - P(X_t = -1 \mid Q_t(\mathbf{1}^{\leq t}) = 0) \\
        &= \frac
        {
        P(Q_t(\mathbf{1}^{\leq t}) = 0 \mid X_t = 1)P(X_t = 1) - P(Q_t(\mathbf{1}^{\leq t}) = 0 \mid X_t = -1)P(X_t = -1)
        }
        {
        P(Q_t(\mathbf{1}^{\leq t}) = 0)
        } \\
        &= \frac
        {
        \gamma P(Q_t(\mathbf{1}^{\leq t}) = 0 \mid X_t = 1) - (1 - \gamma) P(Q_t(\mathbf{1}^{\leq t}) = 0 \mid X_t = -1)
        }
        {
        P(Q_t(\mathbf{1}^{\leq t}) = 0)
        } \\
        \mathbb{E}[X_t \mid Q_t(\mathbf{0}^{\leq t}) = 1]
        &= P(X_t =1 \mid Q_t(\mathbf{0}^{\leq t}) = 1) - P(X_t = -1 \mid Q_t(\mathbf{0}^{\leq t}) = 1) \\
        &= \frac
        {
        P(Q_t(\mathbf{0}^{\leq t}) = 1 \mid X_t = 1)P(X_t = 1) - P(Q_t(\mathbf{0}^{\leq t}) = 1 \mid X_t = -1)P(X_t = -1)
        }
        {
        P(Q_t(\mathbf{0}^{\leq t}) = 1)
        } \\
        &= \frac
        {
        \gamma P(Q_t(\mathbf{0}^{\leq t}) = 1 \mid X_t = 1) - (1 - \gamma) P(Q_t(\mathbf{0}^{\leq t}) = 1 \mid X_t = -1)
        }
        {
        P(Q_t(\mathbf{0}^{\leq t}) = 1)
        }
    \end{align*}
    We can then write $\mathbb{E}[X_t \mid H_t = YI]$ as
    \begin{align*}
        \mathbb{E}[X_t \mid H_t = YI]
        = \frac{1}{\mathbb{E}[Q_t(\mathbf{1}^{\leq t}) - Q_t(\mathbf{0}^{\leq t})]} \cdot
        \left[
            (2\gamma - 1) - (E_{NA} + E_{AA})
        \right]
    \end{align*}
    where 
    \begin{align*}
        E_{NA} &:= \gamma P(Q_t(\mathbf{1}^{\leq t}) = 0 \mid X_t = 1) - (1 - \gamma) P(Q_t(\mathbf{1}^{\leq t}) = 0 \mid X_t = -1) \\
        E_{AA} &:= \gamma P(Q_t(\mathbf{0}^{\leq t}) = 1 \mid X_t = 1) - (1 - \gamma) P(Q_t(\mathbf{0}^{\leq t}) = 1 \mid X_t = -1)
    \end{align*}
    \textit{Automation Bias and Alert Fatigue profiles}. For both the Automation Bias and Alert Fatigue data-generating processes, $Q_t(\mathbf{Z}^{\leq t}) \indep {X_t}$. Thus,
   \begin{align*}
       E_{NA} &= (2\gamma - 1) P(Q_t(\mathbf{1}^{\leq t}) = 0) \\
       &= (2\gamma - 1)(1 - \mathbb{E}[Q_t(\mathbf{1}^{\leq t}) ])\\
       E_{AA} &= (2\gamma - 1) P(Q_t(\mathbf{0}^{\leq t}) = 1) \\
       &= (2\gamma - 1)(\mathbb{E}[Q_t(\mathbf{0}^{\leq t})])\\
        \Rightarrow 
        E_{NA} + E_{AA} &= (2\gamma - 1)(1 - (\mathbb{E}[Q_t(\mathbf{1}^{\leq t})] - \mathbb{E}[Q_t(\mathbf{0}^{\leq t})])) \\
        \Rightarrow
        \mathbb{E}[X_t \mid H_t = YI] 
        &= \frac{1}{\mathbb{E}[Q_t(\mathbf{1}^{\leq t})] - \mathbb{E}[Q_t(\mathbf{0}^{\leq t})]}
        \left[ 
            (2\gamma - 1) - 
            (2\gamma - 1)(1 - (\mathbb{E}[Q_t(\mathbf{1}^{\leq t})] - \mathbb{E}[Q_t(\mathbf{0}^{\leq t})]))
        \right] \\
        &= 2\gamma - 1
    \end{align*}
    Applying the same logic and replacing $Q_t(\mathbf{1}^{\leq t})$ with $Q_t(\mathbf{0}^{< t}, 1)$ gives us $\mathbb{E}[X_t \mid I_t = YI] = 2\gamma - 1.$

    \textit{Calibrated Reliance profile}. For the Calibrated Reliance data-generating process, the same independence assumption between $Q_t(\mathbf{Z}^{\leq t})$ and $X_t$ does not hold. We therefore plug in the parameters for the data-generating process to derive expressions for $\mathbb{E}[Q_t(\mathbf{1}^{\leq t}) - Q_t(\mathbf{0}^{\leq t})]$, $E_{AA}$ and $E_{NA}$. 

    As shown in the proof of Example \ref{app:ex:ab-af-yielding-effects} above,
    \begin{align*}
        \mathbb{E}[Q_t(\mathbf{1}^{\leq t}) - Q_t(\mathbf{0}^{\leq t})]
        = \beta_Z + \beta_S (2 \mu_A - 1) (1 - \rho^{t-1})
    \end{align*}

    For $E_{NA}$ and $E_{AA}$, we have
    \begin{align*}
        E_{NA}
        &= \gamma(1 - \mathbb{E}[Q_t(\mathbf{1}^{\leq t}) \mid X_t = 1]) - (1 - \gamma)(1 - \mathbb{E}[Q_t(\mathbf{1}^{\leq t}) \mid X = -1]) \\
        &= \gamma 
        \left(
        1 - q_0 - \beta_Z - \beta_S(2 \mathbb{E}[A_t \mid X_t = 1] - 1)(1 - \rho^{t-1}) 
        \right)
        \\
        &\quad- (1 - \gamma)
         \left(
        1 - q_0 - \beta_Z - \beta_S(2 \mathbb{E}[A_t \mid X_t = -1] - 1)(1 - \rho^{t-1}) 
        \right) \\
        E_{AA} 
        &= \gamma(\mathbb{E}[Q_t(\mathbf{0}^{\leq t}) \mid X_t = 1]) - (1 - \gamma)(\mathbb{E}[Q_t(\mathbf{0}^{\leq t}) \mid X = -1]) \\
        &= \gamma q_0 - (1 - \gamma)q_0 \\
    \Rightarrow
    E_{NA} + E_{AA}  
    &= (2\gamma - 1) 
    - (2\gamma - 1) \beta_Z \\
    &\quad - \beta_S(1 - \rho^{t-1})
    [\gamma (2\mathbb{E}[A_t \mid X_t = 1] - 1) 
     - (1 - \gamma)(2\mathbb{E}[A_t \mid X_t = -1] - 1) 
    ] \\
    &= (2\gamma - 1) \cdot 
    \left[
    1
    - \beta_Z
    - \beta_S (1 - \rho^{t-1})(2\mu_A - 1)
    \right] \\
    &\quad- 4 ( 1 - \gamma) (\mu_A - \mu_A \mid_{X = -1})\beta_S ( 1 - \rho^{t-1})
    \end{align*}
    where $\mu_A \mid_{X = -1} := \mathbb{E}[A_t \mid X_t = -1]$

Plugging these values back into the expression for $\mathbb{E}[X_t \mid H_t = YI]$,
\begin{align*}
    \mathbb{E}[X_t \mid H_t = YI] &= 
    \frac{1}
    {
    \beta_Z + \beta_S(2\mu_A - 1) (1 - \rho^{t-1})
    }\\
    &\quad \cdot
    \left[
    (2\gamma - 1) \right. \\
    &\left. \quad -
    (2\gamma - 1)\cdot \left(1 - \beta_Z - \beta_S(1 - \rho^{t-1})(2\mu_A - 1)\right) \right. \\
    &\left. \quad + 
    4(1 - \gamma)(\mu_A - \mu_A \mid_{X=-1})\beta_S(1 - \rho^{t-1})
    \right] \\
    &= (2 \gamma - 1) + 
        \frac 
        {4(1 - \gamma)(\mu_A - \mu_A \mid_{X=-1})\beta_S(1 - \rho^{t-1})}
        {\beta_Z + \beta_S(2\mu_A - 1) (1 - \rho^{t-1})}
\end{align*}

To determine $\mathbb{E}[X_t \mid I_t = YI]$, we replace $Q_t(\mathbf{1}^{\leq t})$ with $Q_t(\mathbf{0}^{< t}, 1)$. As shown in the proof of Example \ref{app:ex:ab-af-yielding-effects} above, 
\begin{align*}
    \mathbb{E}[Q_t(\mathbf{0}^{<t}, 1) - Q_t(\mathbf{0}^{\leq t})] = \beta_Z
\end{align*}
The value of $E_{AA}$ is the same as above. For $E_{NA}$, we have 
\begin{align*}
    E_{NA} 
    &= \gamma(1 - \mathbb{E}[Q_t(\mathbf{0}^{<t}, 1) \mid X_t = 1]) - (1 - \gamma)(1 - \mathbb{E}[Q_t(\mathbf{0}^{<t}, 1) \mid X_t = -1]) \\
    &= \gamma (1 - q_0 - \beta_Z) - (1 - \gamma)(1 - q_0 - \beta_Z) \\
    \Rightarrow E_{NA} + E_{AA} &= (2\gamma - 1) \cdot (1 - \beta_Z)
\end{align*}
Plugging back into the expression for $\mathbb{E}[X_t \mid I_t = YI]$,
\begin{align*}
    \mathbb{E}[X_t \mid I_t = YI]
    &= \frac{1}{\beta_Z} \cdot [(2\gamma - 1) - (2\gamma-1) \cdot (1 - \beta_Z)] = 2\gamma - 1
\end{align*}

\end{proof}

\clearpage

\section{Included Papers from Elicit-Assisted Review}\label{sec:elicit-bib}
{\renewcommand{\bibsection}{}\bibliographysec{refs_elicit}}

\end{document}

%% file: math_commands.tex
\usepackage{amsmath,amsfonts,bm}

\def\eqref#1{equation~\ref{#1}}

\def\1{\bm{1}}

\def\rmZ{{\mathbf{Z}}}

\DeclareMathAlphabet{\mathsfit}{\encodingdefault}{\sfdefault}{m}{sl}
\SetMathAlphabet{\mathsfit}{bold}{\encodingdefault}{\sfdefault}{bx}{n}

\DeclareMathOperator*{\argmax}{arg\,max}
\DeclareMathOperator*{\argmin}{arg\,min}

\newcommand{\indep}{\perp \!\!\! \perp}